\documentclass[12pt,oneside]{report}

\usepackage[utf8]{inputenc}
\usepackage{graphicx} 
\usepackage[a4paper,left=40mm,right=20mm,top=25mm,bottom=25mm]{geometry} 
\usepackage{lineno}
\usepackage{fancyhdr} 
\usepackage{amsthm} 
\usepackage{amssymb} 
\usepackage{amsmath} 
\usepackage{color}
\usepackage{setspace} 
\usepackage{afterpage} 
\usepackage{lipsum} 
\usepackage[square, numbers, comma, sort&compress]{natbib} 
\usepackage[nottoc,notlot,notlof]{tocbibind} 
\usepackage{mathbbol} 
\usepackage{bm} 
\usepackage{hyperref} 
\usepackage{url} 
\usepackage[table]{xcolor} 
\usepackage{tabto}

\hypersetup{
    colorlinks=true,
    linkcolor=blue,
    filecolor=magenta,      
    urlcolor=cyan,
    citecolor=cyan}
\newtheorem{theorem}{Theorem}[section]
\newtheorem{corollary}{Corollary}[theorem]
\newtheorem{lemma}[theorem]{Lemma}
\newcommand\blankpage{%
   \null
   \thispagestyle{empty}%
   \addtocounter{page}{0}%
   \newpage}
   
\newcommand{\ket}[1]{\left | #1 \right\rangle}
\newcommand{\bra}[1]{\left \langle #1 \right |}
\newcommand{\openone}[0]{\mathbb{1}}
\def\identity{\leavevmode\hbox{\small1\kern-3.8pt\normalsize1}}
\newcommand{\proj}[1]{\ket{#1}\bra{#1}}
\definecolor{timberwolf}{rgb}{0.86, 0.84, 0.82}
\definecolor{trolleygrey}{rgb}{0.5, 0.5, 0.5}
   
\begin{document}
\pagenumbering{roman} 
\thispagestyle{empty}
{\color{white}{space}}
\vspace{0.5cm}

\begin{figure}[h]
\centering
\includegraphics[width=0.8\textwidth]{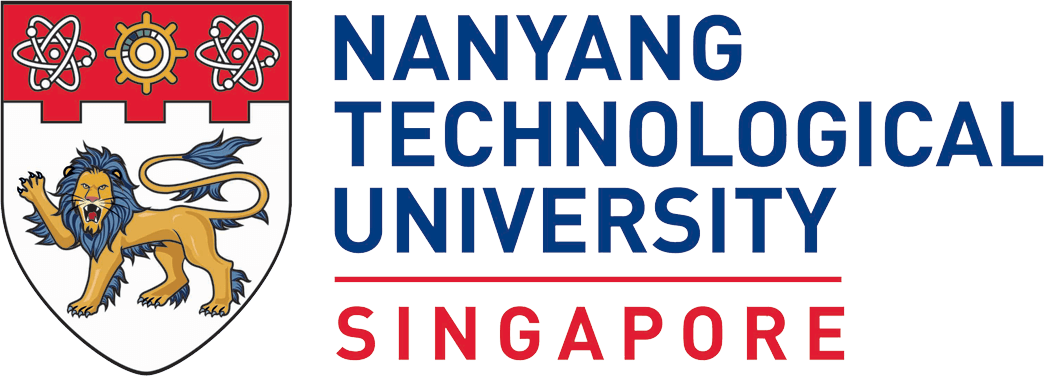}
\end{figure}
\vspace{2cm}

\begin{center}
\large \bf DISTRIBUTION OF QUANTUM ENTANGLEMENT:\\ PRINCIPLES AND APPLICATIONS\end{center}
\vspace{3cm}

\begin{center}
\large \bf I WAYAN GEDE TANJUNG KRISNANDA
\end{center}
\vspace{3cm}


\begin{center}
\large \bf SCHOOL OF PHYSICAL AND MATHEMATICAL SCIENCES\\
\end{center}
\vspace{3cm}

\begin{center}
\large \bf 2020
\end{center}

\afterpage{\blankpage}

\thispagestyle{empty}
{\color{white}{space}}
\vspace{0.5cm}

\begin{center}
\large \bf DISTRIBUTION OF QUANTUM ENTANGLEMENT:\\ PRINCIPLES AND APPLICATIONS
\end{center}
\vspace{3cm}

\begin{center}
\large \bf I WAYAN GEDE TANJUNG KRISNANDA
\end{center}
\vspace{3cm}

\begin{center}
\large SCHOOL OF PHYSICAL AND MATHEMATICAL SCIENCES\\
\end{center}
\vspace{6cm}

\begin{center}
\large A thesis submitted to the Nanyang Technological\\ University in partial fulfilment of the requirement for the\\ degree of Doctor of Philosophy
\end{center}
\vspace{1cm}

\begin{center}
\large \bf 2020
\end{center}
\thispagestyle{empty}


\afterpage{\blankpage}

\onehalfspacing 

\pagenumbering{arabic} 
\chapter*{Abstract}
\setcounter{page}{1} 
\addcontentsline{toc}{chapter}{Abstract}

Quantum entanglement is a form of correlation between quantum particles that has now become a crucial part in quantum information and communication science. 
For example, it has been shown to enable or enhance quantum processing tasks such as quantum cryptography, quantum teleportation, and quantum computing. 
However, quantum entanglement is prone to decoherence as a result of interactions with environmental scatterers, making it an expensive resource. 
Therefore, it is crucial to understand its creation.

We centre our attention to a situation where one would like to distribute quantum entanglement between principal particles that are apart.
In this case, it is necessary to use ancillary systems that are communicated between them or interact with them continuously.
Cubitt \emph{et al}. showed that the ancillary systems need not be entangled with the principal particles in order to distribute entanglement.
This has been demonstrated experimentally in the case of communicated ancillary particles and it is now known that the bound on the distributed  entanglement is given by a communicated quantum discord.
On the other hand, little is understood about the setting with continuous interactions, despite its abundant occurrence in nature.

The main focus of this thesis is to study the distribution of quantum entanglement via continuous interactions with ancillary particles, which I will call mediators. 
First, basic concepts and tools that are helpful for this thesis will be introduced.
This includes the description of quantum objects within the framework of quantum mechanics, their dynamics, and important properties.
Next, I will present my work regarding the necessary conditions for entanglement distribution, the factors that are relevant for the distributed amount, and the speed limit to achieving maximum entanglement gain.
Finally, I present some notable applications that can benefit from our work.
This includes, among others, indirect probing of the quantum nature of optomechanical mirrors, photosynthetic organisms, and gravitational interactions.



\newpage

\chapter*{\color{white}{Dedication}}
To my mother and stepfather \thispagestyle{empty}
 

\chapter*{Acknowledgements}
\addcontentsline{toc}{chapter}{Acknowledgements}

I would like to give a \emph{gazillion} thanks to my supervisor A/Prof.~Paterek for all his thoughts and guidance. 
I really appreciate the time we spent discussing and arguing throughout my PhD programme. 
I also thank my co-supervisor Dr. Zuppardo for helping me to get started in this research. 
Special thanks to Prof.~Paternostro for his advice and collaboration for most parts of our projects. 

I thank my current supervisor Prof.~Dumke and also Prof.~Brukner for taking interest in our projects and approving my transitions.

Many thanks to our colleagues Mr.~Ganardi, Dr.~Lee, Prof.~Kim, Dr.~Marletto, Prof.~Vedral, Dr.~Noh, A/Prof.~Streltsov, Prof.~Bandyopadhyay, Dr.~Banerjee, Dr.~Deb, Dr.~Halder, Dr.~Modi, Prof.~Sengupta, Mr.~Pal, Ms.~Batra, and Prof.~Mahesh for their collaboration.

I thank Prof. Winter, Dr. Coles, and Prof. Gr\"{o}blacher for stimulating discussions, which improved our articles. 

I wish to express my gratitude to Prof.~Laskowski for hospitality at the University of Gda\'{n}sk, Dr.~Dutta for hospitality and collaboration at the Korea Institute for Advanced Study, and Prof. Liew for hospitality at Nanyang Technological University. 

I had the opportunity to co-supervise undergraduate students who were keen in doing research. 
I would like to thank Guo Yao Tham, Wen Yu Kon, Koo Sui Ho Edmund, Parth Patel, Swetha Sridhar, Jeremy Goh Swee Kang, Chee Mun Yin, and Tan Xue Yi for their contributions to our projects.

I thank my colleagues Dr.~Miller, Dr.~Lake, Minh Tran, Kai Sheng, Zhao Zhuo, Liu Zheng, Lamia Varawala, and Ankit Kumar for the time we spent together. 

Thank you to my sister, my mother, my stepfather, and my friends Andhita, David, Hendra, Wiswa, and Zhonglin for their company. 
I also thank Mary and Steven for their hospitality.
Towards the end of my PhD, there is a special person who added new colours in my life. 
I would like to express my gratefulness to my girlfriend Qiannan Zhang.

I wish to acknowledge the funding support for our projects by the National Research Foundation Singapore and Singapore Ministry of Education Academic Research Fund Tier 2 Project No. MOE2015-T2-2-034.  

\afterpage{\blankpage} 

\newpage
\addcontentsline{toc}{chapter}{Table of contents}
\tableofcontents





\chapter*{Abbreviations}\label{C_abb}
\addcontentsline{toc}{chapter}{Abbreviations}

\[
\begin{array}{lcl}  
\text{am} &\text{    }& \text{Air molecules}\\
\text{BChl} &\text{    }& \text{Bacteriochlorophyll}\\
\text{cg} &\text{    }& \text{Casimir-gravity}\\
\text{CM} &\text{    }& \text{Covariance matrix}\\
\text{CV} &\text{    }& \text{Continuous variable}\\
\text{det} &\text{    }& \text{Determinant}\\
\text{di} &\text{    }& \text{Direct interaction}\\
\text{FWHM} &\text{    }& \text{Full-width at half-maximum}\\
\text{gnd} &\text{    }& \text{Ground}\\
\text{inf} &\text{    }& \text{Infimum}\\
\text{int} &\text{    }& \text{Interaction}\\
\text{LIGO} &\text{    }& \text{Laser Interferometer Gravitational-Wave Observatory}\\
\text{loc} &\text{    }& \text{Local}\\
\text{LOCC} &\text{    }& \text{Local operations and classical communication}\\
\text{max} &\text{    }& \text{Maximum}\\
\text{min} &\text{    }& \text{Minimum}\\
\text{ph} &\text{    }& \text{Photon}\\
\text{PPT} &\text{    }& \text{Positive under partial trasposition}\\
\text{QC, qc} &\text{    }& \text{Quantum-classical}\\
\text{QSL} &\text{    }& \text{Quantum speed limit}\\
\text{RED} &\text{    }& \text{Relative entropy of discord}\\
\text{REE} &\text{    }& \text{Relative entropy of entanglement}\\
\text{RHS} &\text{    }& \text{Right hand side}\\
\text{SEC} &\text{    }& \text{System-environment correlation}\\
\text{sep} &\text{    }& \text{Separable}\\
\text{SDE} &\text{    }& \text{Standard deviation of energy}\\
\text{sq} &\text{    }& \text{Squeezed}\\
\text{sup} &\text{    }& \text{Supremum}\\
\text{th} &\text{    }& \text{Thermal}\\
\text{tr} &\text{    }& \text{Trace}
\end{array}
\]

\afterpage{\blankpage} 
\chapter*{List of notations}\label{D_sym}
\addcontentsline{toc}{chapter}{List of symbols}
A list of general symbols used throughout this thesis is provided below. 
Some related notations are defined in groups for better clarity.
Special usage in some Chapters will be mentioned.
Note that trivial notations are not included.
\[
\begin{array}{lcl}  
A,B,C & \text{    } &\mbox{Generally used to indicate principal quantum systems ($A$ \& $B$)}\\
 & \text{    } & \text{and a mediating system ($C$). Also, multiple systems together are}\\
 & \text{    } & \text{denoted simply as $AB$, $ABC$, etc.}\\
C_{X:Y}, \tilde C_{X:Y} & \text{    } &\mbox{Classical correlation between systems $X$ \& $Y$ and its lower}\\
& \text{    } &\mbox{bound.}\\
D & \text{    } & \text{Used to denote a diagonal matrix associated with Brownian}\\
& \text{    } & \text{noises and decay rates.}\\
D_{X|Y} & \text{    } & \text{Quantum discord between systems $X$ and $Y$, given the}\\
 & \text{    } & \text{measurement outcomes on system $Y$.}\\
D(\rho,\sigma) & \text{    } & \text{General distance measure between quantum states $\rho$ and $\sigma$.}\\
& \text{    } & \text{$D_{\text{tr}}(\rho,\sigma)$ denotes trace distance measure.}\\
d,d_X & \text{    } & \text{Dimension of the Hilbert space of a quantum system and of}\\
 & \text{    } & \text{a quantum system $X$. Not to be confused with an object's length.}\\
E_{X:Y} & \text{    } & \text{Quantum entanglement (entropy of entanglement or relative}\\
 & \text{    } & \text{entropy of entanglement) between systems $X$ and $Y$. It is also}\\
 & \text{    } & \text{used to denote logarithmic negativity for the applications part of}\\
 & \text{    } & \text{the thesis, for simplicity. Not to be confused with another usage}\\
 & \text{    } & \text{of $E$ which denotes energy or expectation value of energy.}\\
 \bm{E}_m & \text{    } & \text{Used to denote strength of driving lasers (has frequency unit).}\\
F_m,Q_n & \text{    } & \text{Zero-mean Gaussian noise operators for CV mode $m$ and $n$,}\\
 & \text{    } & \text{respectively. Also, $X_m$ and $Y_m$ are used to denote noise}\\
 & \text{    } & \text{quadratures. These are only used in Chapter~\ref{Chapter_probing}.}\\
 g_{mn},G_{mn}& \text{    } & \text{Individual and collective coupling between mode $m$ and $n$. We }\\
 & \text{    } & \text{also use $\tilde G_{n}$ to indicate base coupling. These quantities are only}\\
 & \text{    } & \text{defined in Chapter~\ref{Chapter_probing}.}\\
G_{0J}& \text{    } & \text{Coupling constant defined in Chapter~\ref{Chapter7}. Note also its effective}\\
& \text{    } & \text{coupling $G_J$.}\\

\end{array}
\]

\newpage
\[
\begin{array}{lcl}  
H, H_X,H^j_X & \text{    } & \text{General Hamiltonian, Hamiltonian applies only on system $X$, }\\
 & \text{    } & \text{and its $j$th variant. A lowercase letter can also indicate variant,}\\
 & \text{    } & \text{e.g., $H_j$. Also, $H_{\text g}$ is used to indicate gravitational term.}\\
\langle H \rangle, \Delta H& \text{    } & \text{Expectation value of energy with respect to the ground state }\\
& \text{    } & \text{level, i.e., $\mbox{tr}(\rho H)-E_g$ and energy variance. The dimensionless}\\
& \text{    } & \text{quantities are written as $\langle M \rangle$ and $\Delta M$, respectively.}\\
I_{X:Y} & \text{    } & \text{Mutual information between systems $X$ and $Y$. Not to be}\\
& \text{    } & \text{confused with $\mathcal{I}_{X:Y}$, see below.}\\
j,j^{\dagger} & \text{    } &\mbox{Annihilation and creation operator, respectively,}\\
 & \text{    } & \text{for a general $d$-dimensional quantum system.}\\
 j_{\text{in}} & \text{    } &\mbox{Used to denote input noise in Chapter~\ref{Chapter7}. The corresponding}\\
& \text{    } &\mbox{noise quadratures are denoted by $x_{\text{in},J}$ and $y_{\text{in},J}$.}\\
K & \text{    } & \text{A drift matrix. Note also the notation $W_{\pm}(t)=\exp{(\pm Kt)}$ or}\\
& \text{    } & \text{$M(t)=\exp{(Kt)}$.}\\
L_{X:Y} & \text{    } & \text{Logarithmic negativity between systems $X$ and $Y$. Not to be}\\
& \text{    } & \text{confused with $L$ which denotes distance.}\\
L\rho, L_X\rho & \text{    } & \text{Lindblad operation on a quantum state $\rho$ and with the operation}\\
& \text{    } & \text{only on system $X$.}\\
m,m_X& \text{    } & \text{Mass of an object. Not to be confused with dummy index.}\\
n_r& \text{    } & \text{Refractive index.}\\
n(t)& \text{    } & \text{Noise vector used in Chapter~\ref{Chapter7}.}\\
\bar n, \bar N & \text{    } & \text{Mean excitation numbers. For example, $\bar n=(e^\beta-1)^{-1}$ is}\\
& \text{    } & \text{mean thermal phonon with $\beta=\hbar \omega/k_B T$.}\\
N_{X:Y} & \text{    } & \text{Negativity between systems $X$ and $Y$.}\\
P & \text{    } & \text{Used to indicate pressure. Not to be confused with power $\bm{P}$.}\\
Q,Q_X,Q_X^j & \text{    } & \text{A quantum operation, operation applied only on system $X$,}\\
& \text{    } & \text{and its $j$th variant.}\\
Q_{X:Y} & \text{    } & \text{General correlation quantifier between systems $X$ and $Y$.}\\
r_1,r_2,r_3& \text{    } & \text{Interaction rates for different configurations in Chapter~\ref{Chapter_gravity}.}\\
r_{\text{cg}}& \text{    } & \text{The ratio between Casimir and gravitational energy terms,}\\
& \text{    } & \text{which are relevant for entanglement gain in Chapter~\ref{Chapter_gravity}.}\\
R,R_X & \text{    } & \text{Radius, radius of object $X$. Not to be confused with reflectivity}\\
& \text{    } & \text{of mirrors, $R_{1,2}$, in Chapter~\ref{Chapter_probing}.}\\
S,S^{\dagger}& \text{    } & \text{Squeezing operators with squeezing strength $s$. Subscripts can}\\
& \text{    } & \text{be used to indicate the application on a quantum system, }\\
& \text{    } & \text{e.g., $s_A,s_B$.}\\
\end{array}
\]

\newpage
\[
\begin{array}{lcl}  
S(\rho), S(\rho_X) & \text{    } & \text{von Neumann entropy of state $\rho$ and its subsystem $\rho_X$,}\\
& \text{    } & \text{also can be simply written as $S_X$.}\\
S_{X|Y} & \text{    } & \text{Conditional von Neumann entropy.}\\
S(\rho||\sigma) & \text{    } & \text{Relative entropy distance between states $\rho$ and $\sigma$.}\\
T & \text{    } & \text{Dimensionless time, $T=\omega t$ or $T=\Omega t$, where $\omega$ and $\Omega$}\\
 & \text{    } & \text{being frequencies. Not to be confused with temperature. As a}\\ 
& \text{    } & \text{power, it denotes transposition. Also $T_X$ denotes partial}\\
 & \text{    } & \text{transposition with respect to system $X$.}\\
u & \text{    } & \text{Used to define quadrature vector, e.g., $u=(\bm{X}_A, \bm{P}_A, \bm{X}_B, \bm{P}_B)^T$.}\\
& \text{    } & \text{Not to be confused with dummy index.}\\
U, U_X & \text{    } & \text{Unitary operator in general and only on system $X$.}\\
V, V_{ij}& \text{    } & \text{Covariance matrix and its elements. The elements can also be}\\
& \text{    } & \text{denoted by $2\times 2$ matrices. For example, for a two-mode system,}\\
& \text{    } & \text{the components are $I_A,I_B, L,$ and $L^T$. A useful quantity is}\\
& \text{    } & \text{defined in this regard, i.e., $\Sigma=\mbox{det}I_A+\mbox{det}I_B-2\:\mbox{det}L$. After}\\
& \text{    } & \text{partial transposition, a covariance matrix is written as $\tilde V$.}\\
V_m& \text{    } & \text{Mode volume, only defined in Chapter~\ref{Chapter_probing}.}\\

x, p, \bm{X}, \bm{P} & \text{    } & \text{Position and momentum. The bold uppercase notations are used}\\
 & \text{    } & \text{to indicate the corresponding dimensionless quantities (or}\\
  & \text{    } & \text{quadratures). Also, subscripts are used to indicate their systems, }\\
  & \text{    } & \text{e.g., $x_A,p_B,\bm{X}_A$. For field operators, quadratures are denoted}\\ 
  & \text{    } & \text{by $\bm{X}$ and $\bm{Y}$.}\\

X, Y & \text{    } & \text{Generally used to indicate quantum systems, but they can}\\
 & \text{    } & \text{indicate parties each consisting of one or more quantum systems.}\\
X:Y & \text{    } & \text{Used to indicate symmetrical partition of systems $X$ and $Y$.}\\
X|Y & \text{    } & \text{Indicates asymmetrical partition of systems $X$ and $Y$. In}\\
& \text{    } & \text{particular, it denotes a situation given the knowledge of system}\\
& \text{    } & \text{$Y$, e.g., measurement outcomes.}\\
\mathcal{C} & \text{    } & \text{A non-convex set containing quantum-classical states.}\\
\mathcal{E}_{\text{c}} & \text{    } & \text{Casimir energy between two objects.}\\
\mathcal{F}(\rho,\sigma) & \text{    } & \text{Fidelity between quantum states $\rho$ and $\sigma$}.\\
\mathcal{F}_{\text{i}} & \text{    } & \text{Finesse.}\\
\mathcal{I}_{X:Y} & \text{    } & \text{A distance from a state to the closest product state of the form}\\
& \text{    } & \text{$\sigma_X\otimes \sigma_Y$ as quantified by a general correlation quantifier $D(\cdot,\cdot)$}.\\
\mathcal{L},\mathcal{L}_{X} & \text{    } & \text{General map, applies only on system $X$. It is used in association}\\
& \text{    } & \text{with Lindblad operators.}\\
\mathcal{M} & \text{    } & \text{A non-convex set containing product states.}\\
\mathcal{P}(\rho) & \text{    } & \text{Purity of a quantum state $\rho$.}\\

\end{array}
\]

\newpage
\[
\begin{array}{lcl}  
\mathcal{Q} & \text{    } & \text{Mechanical quality factor.}\\
\mathcal{S} & \text{    } & \text{A convex set containing separable states. Also used to indicate}\\
& \text{    } & \text{a set in general for general illustration.}\\
\openone, \openone_X& \text{    } & \text{Identity operator, identity operator applied only on system $X$.}\\
|\psi \rangle,|\psi \rangle_X & \text{    } & \text{A pure quantum state, state of system $X$.}\\
& \text{    } & \text{A lowercase subscript indicates its variant, e.g., $|\psi \rangle_j$ or $|\psi_j \rangle$.}\\
\Lambda, \Lambda_X& \text{    } & \text{A general completely positive trace-preserving map,}\\
& \text{    } & \text{applies only on system $X$. Not to be confused with frequency of}\\
& \text{    } & \text{lasers $\Lambda_m$ in Chapter~\ref{Chapter_probing}.}\\
\bm{\Lambda}_{\text{ph}}, \bm{\Lambda}_{\text{am}}& \text{    } & \text{Decoherence coefficients as a result of scattering photons and air}\\
& \text{    } & \text{molecules.}\\
\eta_{\text g}& \text{    } & \text{Dimensionless figure of merit for gravity-mediated}\\
& \text{    } & \text{entanglement gain between two oscillators. $\bm{\sigma}$ is the}\\
& \text{    } & \text{corresponding figure of merit of entanglement between two}\\
& \text{    } & \text{released masses.}\\
\Theta(\rho,\sigma)& \text{    } & \text{Bures angle defined as $\Theta(\rho,\sigma)=\arccos( \mathcal{F}(\rho,\sigma)).$}\\
\Xi_X& \text{    } & \text{Charge of quantum battery for system $X$.}\\
\omega,\Omega& \text{    } & \text{Used to denote frequencies. In Chapter~\ref{Chapter7} Section~\ref{SC_faf}, we also}\\
& \text{    } & \text{use $g$.}\\
\Omega_N& \text{    } & \text{Symplectic form composed of $N$ CV modes. Not to be confused}\\
& \text{    } & \text{with the frequency $\Omega$.}\\
\lambda& \text{    } & \text{Wavelength. Not to be confused with linear mass density or}\\
& \text{    } & \text{complex coefficients.}\\
\Delta_{0J}& \text{    } & \text{Cavity-laser detuning. Note also its effective detuning $\Delta_j$.}\\
\delta_{mn}& \text{    } & \text{Kronecker delta.}\\
\delta(t-t^{\prime})& \text{    } & \text{Direct delta.}\\
\kappa_{\nu},\upsilon(t),l(t)& \text{    } & \text{Constant vector, noise vector, and a vector defined as }\\
& \text{    } & \text{$l(t)=\upsilon(t)+\kappa_{\nu}$.}\\
\xi,\xi_X& \text{    } & \text{Brownian noise, affecting system $X$. Also used to denote a}\\
& \text{    } & \text{combined annihilation or creation operators in Chapter~\ref{Chapter7}.}\\
& \text{    } & \text{Section~\ref{SC_faf}.}\\
\gamma,\gamma_X& \text{    } & \text{Decay rate, affecting system $X$. Note that a notation $\kappa$ is also}\\
& \text{    } & \text{used to indicate decay in Chapter~\ref{Chapter_probing} and \ref{Chapter7}.}\\
\tau_{\text{e}},\tau_{\text{ph}},\tau_{\text{am}}& \text{    } & \text{Entangling time, coherence time due to photon scattering, and}\\
& \text{    } & \text{due to air molecules.}\\

\end{array}
\]

\newpage
\[
\begin{array}{lcl}  
\tau(\rho,\sigma)& \text{    } & \text{The time it takes for a state $\rho$ to evolve to another state $\sigma$.}\\
& \text{    } & \text{Its dimensionless quantity is denoted as $\Gamma(\rho,\sigma)$ or simply $\Gamma$.}\\
& \text{    } & \text{Note, for example, $\Gamma_{\text{di}}$ stands for dimensionless time bound for}\\
& \text{    } & \text{direct interactions scenario.}\\

\Gamma_n& \text{    } & \text{FWHM of bacterial spectrum in Chapter~\ref{Chapter_probing}.}\\
\mu& \text{    } & \text{Dipole moment of two-level transition. Not to be confused with}\\
& \text{    } & \text{dummy index.}\\
\nu_k,\tilde \nu_k& \text{    } & \text{The $k$th symplectic eigenvalue of a covariance matrix,}\\
& \text{    } & \text{and its value after partial transposition. Not to be confused}\\
& \text{    } & \text{with the frequency constant $\nu$ in Chapter~\ref{Chapter_gravity}}\\
\rho,\rho_X,\rho_X^j& \text{    } &\text{General density matrix, density matrix of system $X$, and its $j$th}\\
& \text{    } & \text{variant. A lowercase subscript also denotes its variant, e.g., $\rho_j$.}\\
 & \text{    } & \text{Not to be confused with mass density in Chapter~\ref{Chapter_gravity}.}\\
 & \text{    } & \text{Note that $\sigma$ is also often used to denote a density matrix.}\\
\sigma^j,\sigma_X^j& \text{    } & \text{$j$-Pauli matrix ($j=x,y,z$). The subscript indicates its}\\
& \text{    } & \text{application on system $X$ in the context of multiple two-level}\\
& \text{    } & \text{systems.}\\
\sigma^{+(-)},\sigma_X^{+(-)}& \text{    } & \text{Two-level system raising (lowering) operator in general and for}\\
& \text{    } & \text{system $X$.}\\
\Pi^j_X& \text{    } & \text{The $j$th projection operator of system $X$.}
\end{array}
\]


\afterpage{\blankpage} 

\chapter{Introduction} 

\label{Chapter0} 

\lhead{Chapter 1. \emph{Introduction}} 

\emph{This chapter presents the motivation, objectives, and organisation of this thesis.\footnote{Note that this chapter serves as a brief introduction. 
Some terminologies used here will be explained more thoroughly in Chapter~\ref{Chapter1}.}
It begins with the importance of quantum entanglement, leading to the need for understanding the process of its creation. 
In particular, we focus on mediated interaction scenarios where the aim is to create or distribute entanglement between two quantum objects by utilising an ancillary system.
This sets the stage for the thesis.
Finally, I will describe the objectives of our study and how the thesis is organised.}

\clearpage 

\section{Motivation}
Entanglement is a form of quantum correlation first recognised by Einstein, Podolsky, and Rosen~\cite{epr}, and Schr\"{o}dinger~\cite{schrodinger1935discussion} in 1935.
This correlation has proven crucial and resulted in key proposals for enabling tasks that are not possible in the classical regime.
This includes, to name a few, quantum cryptography~\cite{cryptography}, quantum dense coding~\cite{dense-coding}, and quantum teleportation~\cite{teleportation}.
The rapid utilisation and manipulation of quantum entanglement for the purpose of quantum information tasks have made it a crucial resource, some argued, as real as energy~\cite{RevModPhys.81.865}.
Therefore, the study of its creation (or distribution) is a necessity.
In this section, we will briefly review protocols for creating quantum entanglement between two principal objects, here denoted as, $A$ and $B$.
In particular, we describe a scenario where $A$ interacts directly with $B$ and the one where the interactions are mediated by an ancillary system $C$.
For the latter case, the nature of the interactions can be discrete or continuous, which will be the main focus of this thesis.

\subsection{Direct interactions}

A straightforward way to create quantum correlations is to engineer direct interactions between the principal objects, see Fig.~\ref{FIG_ch2_direct}.
This dynamics can either be closed (unitary) or open (e.g., following the Lindblad master equation). 
The general Hamiltonian considered in this case is of the form $H_{AB}=\sum_{j} H_A^j \otimes H_B^j,$ where the subscripts indicate the corresponding objects. 
Note also that the Hamiltonian components can be local, i.e., $H_A\otimes \openone_B$ or $\openone_A\otimes H_B$, which we will simply write as $H_A$ and $H_B$ hereafter.
As an example, let us consider two qubits that are interacting with Hamiltonian $H=\hbar \omega \:\sigma_A^x\otimes \sigma_B^x$, where $\hbar \omega$ is the energy unit and $\sigma_{A(B)}^x$ is the $x$-Pauli matrix with the subscript denoting the corresponding object. 
For the initial condition, let us consider the state $\ket{\psi(0)}=\ket{0}_A\ket{0}_B$, where $\ket{0}_{A(B)}$ is the eigenstate (ground state) of the $z$-Pauli matrix.
One can see that the dynamics from the Hamiltonian will create excitations in both objects $A$ and $B$, i.e., the state at time $t$ is given by $\ket{\psi(t)}=\alpha(t)\ket{0}_A\ket{0}_B+\beta(t)\ket{1}_A\ket{1}_B$, where $\alpha(t)$ and $\beta(t)$ are complex coefficients.
This is a form of state where $A$ is said to be entangled with $B$, as the state of one object cannot be written separately from the other.\footnote{A detailed analysis will be given later in this thesis.}
\begin{figure}[h]
\centering
\includegraphics[scale=0.4]{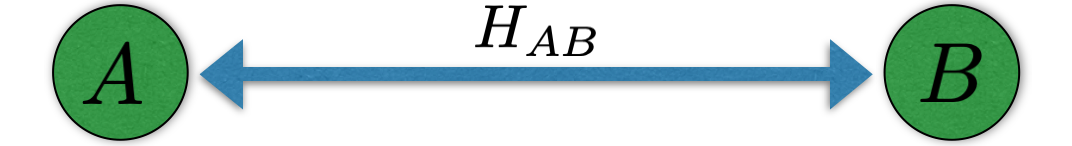} 
\caption{Direct interaction setting. Two quantum particles $A$ and $B$ are interacting with Hamiltonian $H_{AB}$.}
\label{FIG_ch2_direct} 
\end{figure}

\subsection{Indirect interactions}

In some cases, the principal objects are not situated together, e.g., $A$ is in Alice's lab and $B$ is in Bob's lab (see Fig.~\ref{FIG_ch2_indirect}).
Therefore, they require an ancillary system to mediate interactions between them.
We will refer to this as indirect interaction setting.
Two scenarios can be applied here based on the nature of the interactions.
A discrete interaction scenario (Fig.~\ref{FIG_ch2_indirect}a) involves performing operations (or discrete interactions) on objects $A$ and $C$ in Alice's lab, then $C$ is sent to Bob's lab, after which Bob can perform operations on $B$ and $C$.
On the other hand, a continuous interaction scenario (Fig.~\ref{FIG_ch2_indirect}b) considers continuous (in time) interactions between $A$ and $C$ as well as between $B$ and $C$.
We note that both scenarios have been presented by Cubitt \emph{et al.}~\cite{cubitt}.

\begin{figure}[h]
\centering
\includegraphics[scale=0.45]{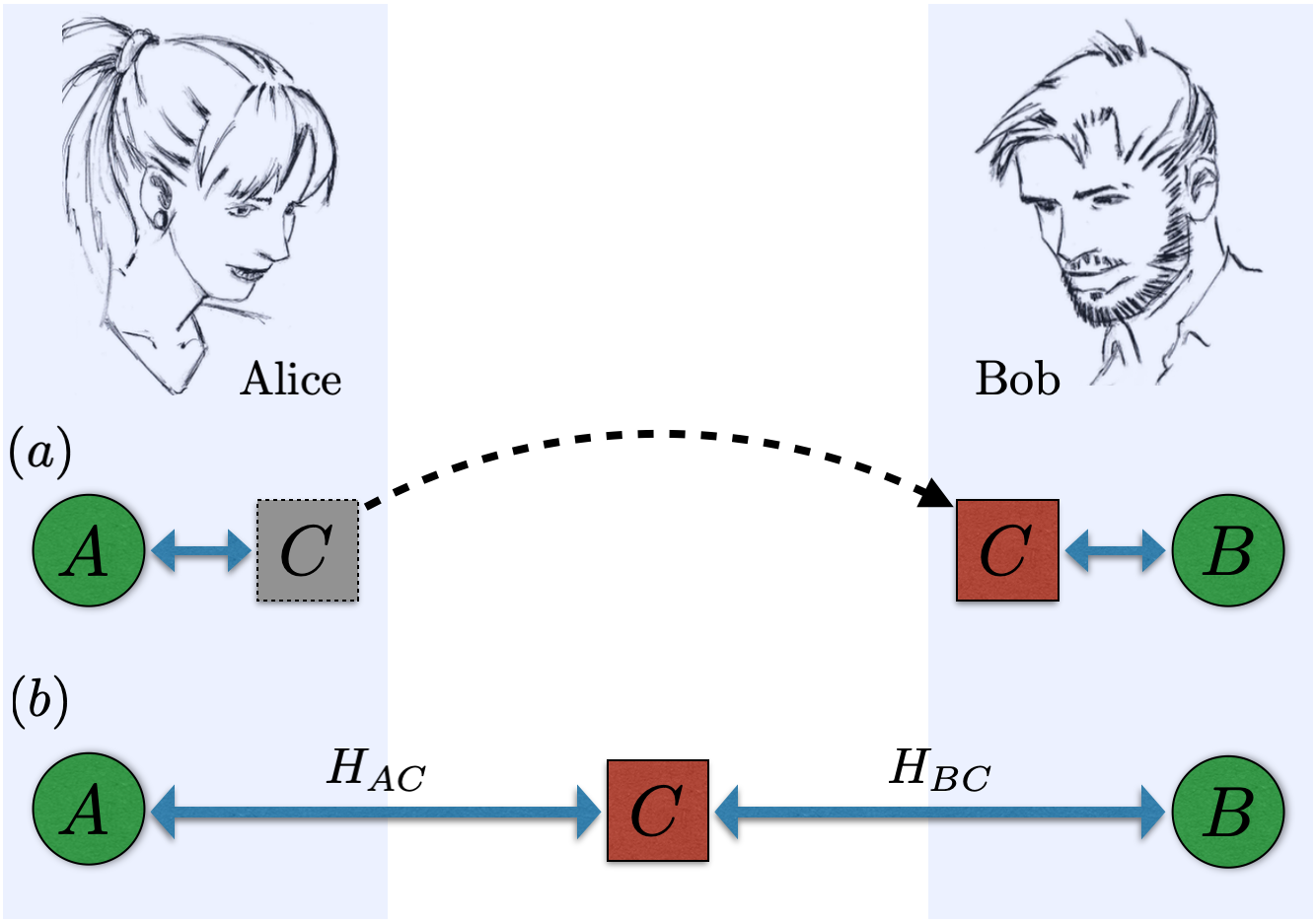} 
\caption{Indirect interaction setting. Quantum particles $A$ and $B$ are in different locations. In discrete interaction scenario (a), an ancillary system $C$ first interacts with $A$, then it is sent to Bob, after which it also interacts with $B$. In scenario (b), the ancillary system is interacting continuously with $A$ and $B$. Note that in both scenarios, $A$ and $B$ are not interacting directly, i.e., there is no Hamiltonian $H_{AB}$.}
\label{FIG_ch2_indirect} 
\end{figure}

\subsubsection{Discrete interactions}

In the discrete interaction scenario, a simple way to create entanglement between $A$ and $B$ is to first entangle $A$ and $C$, after which one can perform entanglement swapping (e.g., see Ref.~\cite{swapping} and its implementations~\cite{exp_swap1,exp_swap2}) on $B$ and $C$, resulting in entanglement between the principal objects.
This scenario requires $C$ to be entangled in the process.
Interestingly, the distribution of entanglement is also possible via a separable ancillary system~\cite{cubitt}.
The results of Cubitt \emph{et al.} have attracted attention, resulting in entanglement distribution proposals~\cite{mista1,mista2,mista3,akay}.
The distribution via separable mediators has also been verified experimentally~\cite{exp0,edssexp1,edssexp2,edssexp3}.
Additionally, Refs.~\cite{bounds1, bounds2} have shown that the entanglement gain in this scenario is bounded by another form of non-classical correlation, known as quantum discord, between the principal particles $AB$ and the mediator $C$.
It demonstrates that the change of quantum entanglement as a result of relocating the object $C$ from Alice's lab to Bob's lab is bounded by the quantum discord carried by $C$. 
This way, some form of quantum correlation is required to distribute entanglement.
We also note the distribution in the presence of noise with the use of different entanglement quantifiers~\cite{streltsov2015progress,excessive}.

\subsubsection{Continuous interactions}

On the other hand, the continuous interaction scenario assumes that the total energy of the objects is described by Hamiltonian of the form $H=H_{AC}+H_{BC}$.
We note $A$ and $B$ do not interact directly, and therefore there is no $H_{AB}$ component.
More explicitly, the total Hamiltonian can be written as $H=\sum_j H_A^j\otimes H_{C_1}^j+\sum_k H_B^k\otimes H_{C_2}^k$ (local Hamiltonian terms are included).
Note that the subscripts $1$ and $2$ are used to indicate that the terms $H_{C_1}^j$ and $H_{C_2}^k$ can be different in general.
For most parts of this thesis, we allow each object to interact with its own local environment, with the corresponding dynamics following the Lindblad master equation.
This way, the continuous interaction scenario corresponds to ample situations in nature, some of which will be considered in this thesis.
Yet, this entanglement distribution method is not as well understood as its counterpart, the discrete interaction scenario.

\section{Objectives}
As first shown by Cubitt \emph{et al.}, the continuous scenario allows entanglement distribution via a separable mediator~\cite{cubitt}. 
Investigating further questions in this direction is therefore the goal of this thesis.
In particular, one seeks answers to the following:
\begin{enumerate}
\item The requirements for entanglement distribution (if not entangled mediator). 
\item The amount of distributed entanglement and the relevant factors.
\item The speed limit of distributing quantum entanglement.
\end{enumerate}
In addition to answering these questions, we will also present the implications of our study and explore applications that can benefit from it.

\section{Thesis organisation}
This thesis is organised as follows. 
Firstly, Chapter \ref{Chapter1} covers preliminaries that are useful for the core of the thesis. 
Readers who are interested mainly in new results can proceed to Chapters \ref{Chapter_revealing}--\ref{Chapter7}. 
Chapter \ref{Chapter_revealing} presents our results on the resource necessary for distribution of quantum entanglement via continuous interactions. 
An important and general application will be constructed that is aimed at indirect revelation of non-classical properties of unknown objects. 
Another resource for entanglement distribution will be presented in Chapter \ref{Chapter_detecting}. 
More specifically, we link the amount of distributed correlations to one of the most elementary non-classical features of quantum observables, that is non-commutativity.  
Chapter \ref{Chapter_speedup} considers the speed of entanglement distribution for direct and indirect interaction cases. 
We present conditions for the indirect interaction setting to saturate the maximum entangling speed that is derived from the quantum speed limit. 
Major applications of our theorems will be presented in Chapter \ref{Chapter_gravity} and \ref{Chapter_probing}. 
Chapter \ref{Chapter_gravity} addresses one of the most unsettled quests in physics, that is showcasing quantum nature of gravitational interactions. 
In particular, we propose experimental setups where gravity can mediate observable quantum entanglement between two otherwise non-interacting masses. 
Chapter \ref{Chapter_probing} covers the problem of detecting quantum features in living organisms by proposing an experimental setup where one can probe non-classical properties of photosynthetic organisms, such as green sulphur bacteria.
Finally, other applications will be presented in Chapter \ref{Chapter7}.

\afterpage{\blankpage} 

\chapter{Quantum objects, their dynamics, and properties} 
\label{Chapter1} 

\lhead{Chapter 1. \emph{Introduction}} 

\emph{This chapter provides preliminaries that will be useful for the remainder of this thesis.
First, I describe quantum objects and their states, both for discrete systems and continuous variable systems.
Additionally, quantum systems evolve in time due to interactions between them or simply as a result of local energy. 
In this regard, I provide the dynamics of quantum objects both for closed and open (interactions with environment) systems.
Finally, I proceed to review important properties that quantum objects can possess. 
This includes quantities such as von Neumann entropy, purity, and fidelity, as well as correlations such as quantum entanglement, quantum discord, and mutual information.
More in-depth analysis can be found in Refs.~\cite{nielsenchuang,RevModPhys.81.865,discord4,weedbrook2012gaussian}.
I will also introduce the notion of classicality and non-classicality that will be used throughout the thesis.
}

\clearpage 

\section{Quantum objects}

Within the framework of quantum mechanics, a physical object is fully described by its wave function (or, more abstractly a state $\ket{\psi}$). 
Knowing the state of the object, one can then calculate expectation values of observables such as position, momentum, and others.
In this section, we focus on describing the state of objects for two cases.
One is when an object ``lives" in a finite dimensional Hilbert space, such as the spin state of an electron, while the other is when the description requires infinite dimension, such as an atom in a harmonic potential. 
Indeed, the latter has infinitely many energy levels.
Finally, I will also cover the evolution of these objects both where they are isolated and having interactions with environment.

\subsection{Discrete systems}
For quantum systems that only require finite-dimensional space, their pure state is superposition 
\begin{equation}\label{EQ_purestate}
\ket{\psi}=\sum_{n=0}^{d-1} \lambda_n \ket{n},
\end{equation}
where, e.g., $\{\ket{n}\}$ correspond to discrete energy levels,\footnote{We will also call them eigenstates and they form an orthonormal basis, i.e., $\langle n|n^{\prime}\rangle=\delta_{nn^{\prime}}$.} $d$ is the dimension of the object's Hilbert space, and $\lambda_n$ is a complex coefficient. 
As a simple example, consider a two-dimensional system such as the spin state of an electron $\ket{\psi_{\text{s}}}=\lambda_0 \ket{0}+\lambda_1 \ket{1}$, where $\ket{0}$ ($\ket{1}$) denotes spin up (down).
Note that normalisation of the state requires $|\lambda_0|^2+|\lambda_1|^2=1$, which can be interpreted as the electron having spin up with probability $|\lambda_0|^2$ and spin down with probability $|\lambda_1|^2$.

A quantum system can also be in a mixture of states (simply, mixed state or density matrix) 
\begin{equation}\label{EQ_mixedstate}
\rho=\sum_m p_m\: \ket{\psi_m}\bra{\psi_m},
\end{equation}
where $\ket{\psi_m}$ is any pure state.
The normalisation condition requires $\sum_m p_m=1$ and $p_m\ge0$ gives the probability of observing the state $\ket{\psi_m}$.
This requirement is referred to as positive semidefiniteness, i.e., a physical state $\rho$ has to have real non-negative eigenvalues with the sum equal to unity.
It turns out that any such $\rho$ has a spectral decomposition 
\begin{equation}\label{EQ_mixedstateeig}
\rho=\sum_{n=0}^{d-1} p_n |n\rangle \langle n|,
\end{equation}
where $\{|n\rangle \}$ form a basis in which $\rho$ is diagonal.
Note that in a matrix representation, the mixed state in Eq.~(\ref{EQ_mixedstateeig}) has dimension $d\times d$, while the pure state $\ket{\psi}$ in Eq.~(\ref{EQ_purestate}) is a vector with dimension $d$. 
In this thesis, we will use $\ket{\psi}$ or $\rho$ to describe the state of quantum objects, from which one can calculate important quantities. 
Some examples will be given later in this chapter.

\subsubsection{Two-level systems}
A special case where $d=2$ (simply, a qubit) will be used often in this thesis.
Now we provide useful notations for the states and operators.
\begin{equation}
\ket{0}=\left( \begin{array}{c} 1\\ 0\end{array}\right),\:\:\: \ket{1}=\left( \begin{array}{c} 0\\ 1\end{array}\right),\:\:\:\ket{\pm}=\frac{1}{\sqrt{2}}\left( \begin{array}{c} 1\\ \pm 1\end{array}\right),\:\:\:\ket{y_{\pm}}=\frac{1}{\sqrt{2}}\left( \begin{array}{c} 1\\ \pm i\end{array}\right).
\end{equation}
Note that $\{\ket{0},\ket{1}\}$, $\{\ket{\pm}\}$, and $\{\ket{y_{\pm}}\}$ each form an eigenbasis of the Pauli matrix
\begin{equation}
\sigma^z=\left( \begin{array}{cc} 1&0\\ 0&-1\end{array}\right),\:\:\: \sigma^x=\left( \begin{array}{cc} 0&1\\ 1&0\end{array}\right),\:\:\: \sigma^y=\left( \begin{array}{cc} 0&-i\\ i&0\end{array}\right),
\end{equation}
respectively. 
Additionally, we also use the raising and lowering operators, respectively written as $\sigma^+=\ket{1}\bra{0}$ and $\sigma^-=\ket{0}\bra{1}$.

\subsection{Continuous variable systems}

Technically, infinite dimensional (simply, continuous variable) systems can be described using the state in (\ref{EQ_purestate}) or (\ref{EQ_mixedstateeig}) with infinite sums. 
The calculations of observables follow accordingly.
However, for continuous variable (CV) Gaussian systems, which belong to a special class of states mostly considered for applications later in this thesis, the description of the corresponding object can be made simpler.
Indeed, a complete specification of the object can be incorporated in the covariance matrix (CM), denoted by $V$, whose components are defined as
\begin{equation}\label{EQ_cmdef}
V_{ij}=\frac{\langle u_iu_j+u_ju_i\rangle}{2} -\langle u_i\rangle \langle u_j\rangle,
\end{equation}
where $u=(\bm{X},\bm{P})^T$ is a quadrature vector\footnote{Note that we adopt the term quadrature vector from literature. However, one should note that $u$ is not a vector in the normal algebraic sense. For instance, the cross product $u\times u \ne 0$, which is a consequence of the non-commuting relation between $\bm{X}$ and $\bm{P}$.} with $\bm{X}$ and $\bm{P}$ being dimensionless position and momentum quadratures.\footnote{In this example, $V$ is a single-mode covariance matrix. For an $N$-mode CM, one writes the quadrature vector $u=(\bm{X}_1,\bm{P}_1,\bm{X}_2,\bm{P}_2,\cdots,\bm{X}_N,\bm{P}_N)^T$.} 

In this thesis, we use the following relations 
\begin{equation}\label{EQ_relation}
\bm{X}=\frac{a+a^{\dagger}}{\sqrt{2}}, \:\:\:\:\bm{P}=\frac{a-a^{\dagger}}{\sqrt{2}i},
\end{equation}
where $a$ ($a^{\dagger}$) is the annihilation (creation) operator of the corresponding infinite dimensional object.
We note that an $N$-mode covariance matrix has $N$ symplectic eigenvalues $\{\nu_k\}_{k=1}^N$, which can be calculated from the spectrum of a matrix $|i\Omega_N V|$, where
\begin{equation}
\Omega_{N}=\bigoplus^{N}_{k=1} \left( \begin{array}{cc} 0&1\\ -1 &0\end{array}\right),
\end{equation}
is known as symplectic form. 
A physical covariance matrix has 
\begin{equation}\label{EQ_cmcon}
\nu_k\ge1/2
\end{equation}
for all $k=1,2,\cdots,N$.
This is essentially a consequence of the uncertainty relation~\cite{weedbrook2012gaussian}.\footnote{Note that some literature uses the relations in (\ref{EQ_relation}) without the factor $\sqrt{2}$, and therefore the corresponding condition in (\ref{EQ_cmcon}) becomes $\nu_k\ge1$.}

We now provide some examples of states and the corresponding CM that will be useful in this thesis. 
This includes coherent, thermal, and squeezed thermal states of a harmonic oscillator
\begin{equation}
\ket{\alpha}= e^{-\frac{|\alpha|^2}{2}}\sum_{j=0}^\infty \frac{\alpha^j}{\sqrt{j!}}\ket{j},\:\:\:\:
\rho_{\text{th}}=\frac{1}{\bar n+1}\sum_{j=0}^{\infty}\left( \frac{\bar n}{\bar n+1}\right)^j \ket{j}\bra{j},\:\:\:\:
\rho_{\text{sq}}= S\rho_{\text{th}}S^{\dagger},
\end{equation}
where the mean phonon number $\bar n=(\exp(\beta)-1)^{-1}$ with $\beta=\hbar \omega/k_BT$, $\hbar \omega$ the energy unit of the oscillator, and $T$ the temperature.
Also, $|j\rangle$ denotes the $j$th eigenstate of the operator $a^{\dag}a$.
The single-mode squeezing operator is given by
\begin{equation}
S=\exp\left(\frac{\eta^{*}}{2}aa-\frac{\eta}{2} {a^{\dagger}a^{\dagger}}\right),
\end{equation} 
where $\eta=s\exp\left(i\theta \right)$ is a complex number with $s$ and $\theta$ being the squeezing strength and angle respectively.
The corresponding CM is obtained by computing the expectation value of observables in Eq.~(\ref{EQ_cmdef}) with respect to the state of the object.
One can confirm
\begin{equation}
V_{\alpha}=\frac{1}{2}\openone,\:\:\:\:\: V_{\text{th}}=\frac{2\bar n+1}{2}\openone,\:\:\:\:\:V_{\text{sq}}=\frac{2\bar n+1}{2} \left( \begin{array}{cc} e^{2s}&0\\ 0&e^{-2s}\end{array}\right),
\end{equation}
where $\openone$ here denotes a $2\times 2$ identity matrix and we have assumed the case $\theta=0$.\footnote{This corresponds to squeezing the momentum (position) quadrature for $s>0$ ($s<0$).}
Moreover, the symplectic eigenvalues are given by $1/2$, $(2\bar n+1)/2$, and $(2\bar n+1)/2$ respectively.
Note that all the covariance matrices above are physical, i.e., they satisfy the condition in Eq.~(\ref{EQ_cmcon}).

\subsection{Quantum dynamics}
Here we describe the evolution of quantum objects, i.e., how their state $\ket{\psi}$ (or, in general $\rho$) changes in time. 

\subsubsection{Closed systems}
For an isolated system, the dynamical equation reads
\begin{equation}\label{EQ_udynamics}
\ket{\psi(t)}=U\ket{\psi(0)},\:\:\:\:\: \rho(t)=U\rho(0)U^{\dagger},
\end{equation}
for pure and mixed states respectively, where $U=\exp(-iHt/\hbar)$ is a unitary operator and $H$ is the Hamiltonian that represents the corresponding total energy of the object.\footnote{We have also assumed that the Hamiltonian here is time-independent which may not apply in general. In this thesis, we always try to move into a frame where the Hamiltonian is time-independent (e.g., by changing to a rotating frame) and, wherever convenient, solve the problem in Heisenberg picture.}
In general, the Hamiltonian can be composed of many terms, i.e., $H=\sum_j H_j$.
Note also that during the evolution, we have $\text{tr}(\rho(t))=\text{tr}(U\rho(0)U^{\dagger})=\text{tr}(\rho(0)U^{\dagger}U)=\text{tr}(\rho(0))=1$, where we used the cyclic property of trace and unitary property $U^{\dagger}U=UU^{\dagger}=\openone$.

Let us now provide a tool for the expression of the unitary operator that will be crucial for some parts of this thesis.
In particular, we consider a unitary operator with Hamiltonian $H=H_1+H_2$. 
We utilise the Trotter expansion (see Appendix~\ref{A_trotter} for a simple derivation) and write the operator as
\begin{equation}
U=\exp\left(-\frac{i(H_1+H_2)t}{\hbar}\right)=\lim_{n\rightarrow \infty} \left( \exp\left(-\frac{iH_1\Delta t}{\hbar}\right)\exp\left(-\frac{iH_2\Delta t}{\hbar}\right)\right)^n,
\end{equation}
where $\Delta t=t/n$.
Hence, one can think of the evolution as if it consisted of sequences of two unitary operations with Hamiltonian $H_2$ and $H_1$ (or, in reverse), each for a time $\Delta t$.
For a special case, in which the Hamiltonians commute,  $[H_1,H_2]=0$, the corresponding operator reads
\begin{equation}
U=\exp\left(-\frac{iH_1t}{\hbar}\right)\exp\left(-\frac{iH_2t}{\hbar}\right).
\end{equation}

\subsubsection{Open systems}

For quantum systems that are open, i.e., having interactions with their environment, we will consider the evolution of the state following the Lindblad equation
\begin{equation}\label{EQ_lindblad}
\frac{d\rho(t)}{d t}=-\frac{i}{\hbar}[H,\rho(t)]+L \rho(t),
\end{equation}
where $L\rho(t)\equiv \sum_k Q^k\rho(t) Q^{k\dag}-\frac{1}{2}\{Q^{k\dag}Q^k,\rho(t)\}$, $H$ is the Hamiltonian of the object, and $Q^k$ describes the operation on the object as a result of interactions with its environment.\footnote{For consistency, we have used the superscript $k$ to denote the variant of the operator $Q$ and not the power.}
Note that, for simplicity, the strength of the interaction is absorbed in $Q^k$ such that the term $L \rho(t)$ has the unit frequency.
For the case where the strength goes to zero, the 2nd term on the RHS of Eq.~(\ref{EQ_lindblad}) vanishes and the dynamics becomes unitary as in Eq.~(\ref{EQ_udynamics}).
Also, by using the cyclic property of trace, it is clear that $\text{tr}([H,\rho(t)])$ and $\text{tr}(L\rho(t))$ both equal zero, and consequently, the Lindblad equation is trace-preserving.

\section{Properties of quantum states}

Having known the state of quantum objects and its evolution in time, we now proceed to some quantities that can be computed from it.
This includes von Neumann entropy, purity, and fidelity between quantum states.
We will focus on definitions and list of properties that will be utilised later in this thesis.

\subsection{von Neumann entropy}

The amount of disorder in a quantum state can be quantified using
\begin{equation}\label{EQ_vnent}
S(\rho)=-\text{tr}(\rho \log_2(\rho)).\footnote{In this thesis, we shall call this the von Neumann entropy as it is commonly done in literature.}
\end{equation}

For a general mixed state, e.g., the one in Eq.~(\ref{EQ_mixedstateeig}), the von Neumann entropy (simply, entropy) is given by $-\sum_n p_n\log_2(p_n)$.
As $\{p_n\}$ form a probability distribution, the entropy of a quantum object is non-negative.
Some other useful properties of von Neumann entropy are listed below.
\begin{enumerate}
\item $S(\rho)=0$ if and only if the state $\rho$ is pure.
\item For a $d$-dimensional object, the maximum entropy is $\log_2(d)$.
\item $S(U\rho U^{\dagger})=S(\rho)$ (invariant under unitary operations).
\item For two systems, $S(\rho_X\otimes \rho_Y)= S(\rho_X)+S(\rho_Y)$ (additive).
\item For three systems, $S(\rho_{XYZ})+S(\rho_{Y})\le S(\rho_{XY})+S(\rho_{YZ})$ (strongly subadditive).
\end{enumerate}

\subsection{Purity}

Some quantum states are more mixed than the others. 
In order to quantify the purity of a state $\rho$, we use 
\begin{equation}
\mathcal{P}(\rho)=\text{tr}(\rho^2).
\end{equation}
For a general state with dimension $d$, it acquires values $1/d\le \mathcal{P} \le 1$.
Indeed, $\mathcal{P}=1$ for a pure state $\rho=\ket{\psi}\bra{\psi}$ and $\mathcal{P}=1/d$ for a maximally mixed state $\rho=\openone/d$.

We note that in a unitary dynamics, one has
\begin{equation}
\text{tr}(\rho(t)^2)=\text{tr}(U\rho(0)U^{\dagger}U\rho(0)U^{\dagger})=\text{tr}(\rho(0)^2),
\end{equation}
with the help of the cyclic property of trace and the unitary property.
Therefore, unitary dynamics is purity-preserving.
While the purity of the state $\rho$ is invariant, the purity of its subsystems can change.
We will see below some examples showing this case.

As the von Neumann entropy characterises the disorder in the system, one can also quantify the purity as $1-S(\rho)/\log_2(d)$.
In this case, its value ranges from a minimum of $0$ to a maximum of unity.

\subsection{Fidelity}

In order to assess how close a quantum state is from another, one of the quantifiers we use is the Uhlmann \emph{root} fidelity
\begin{equation}\label{EQ_rootfidel}
\mathcal{F}(\rho,\sigma)=\mbox{tr}\left( \sqrt{\sqrt{\rho}\sigma\sqrt{\rho}} \right),
\end{equation}
where $\rho$ and $\sigma$ are two density matrices having the same dimension~\cite{fidelity1,fidelity2}.
In a way, the fidelity above is related to a probability of identifying one state as the other.
Some properties of fidelity are listed below
\begin{enumerate}
\item $\mathcal{F}(\rho,\sigma)=\mathcal{F}(\sigma,\rho)$ (symmetrical).
\item Its value follows $0\le \mathcal{F}(\rho,\sigma)\le 1$ and $\mathcal{F}(\rho,\rho)=1$.
\item For pure states $\rho=\ket{\psi}\bra{\psi}$ and $\sigma=\ket{\phi}\bra{\phi}$, $\mathcal{F}(\rho,\sigma)=|\langle \psi|\phi \rangle|$.
\item $\mathcal{F}(\Lambda [\rho],\Lambda[\sigma])\ge \mathcal{F}(\rho,\sigma)$ for any trace-preserving completely positive map $\Lambda$.
\end{enumerate}

One might also consider a measure introduced by Helstrom that shows the best success probability to distinguish two quantum states, i.e., 
\begin{equation}\label{EQ_helstrom}
p_{\text{best}}=\frac{1}{2}\left(1+D_{\text{tr}}(\rho,\sigma)\right),
\end{equation}
where $D_{\text{tr}}(\rho,\sigma)=||\rho-\sigma||_1/2$ is known as the trace distance and $||\cdot||_1$ is trace norm \cite{helstrom1969}. 
The probability in Eq.~(\ref{EQ_helstrom}) is then related to the fidelity in Eq.~(\ref{EQ_rootfidel}) through the following relation \cite{nielsenchuang}
\begin{equation}
1-\mathcal{F}(\rho,\sigma)\le D_{\text{tr}}(\rho,\sigma)\le \sqrt{(1-\mathcal{F}(\rho,\sigma)^2)}.
\end{equation}

\section{Correlations between quantum states}
Two or more quantum objects can have correlations between them that are beneficial for quantum information tasks.
These correlations can be in many forms ranging from purely classical to having quantum nature.
This section reviews types of correlations used in this thesis.
In particular, this includes quantum entanglement, quantum discord, and mutual information. 
Important properties of the quantifiers that will be useful for this thesis will be mentioned.
We will also introduce the notion of classicality and non-classicality.

\subsection{Quantum entanglement}

In this thesis, we will only consider bipartite entanglement, i.e., between party $X$ and $Y$.\footnote{In general, each party can consist of more than one physical object.}
For pure states, separable (not entangled) states are those which can be written in a product form
\begin{equation}\label{EQ_puresep}
\ket{\psi}=\ket{\psi_X}\otimes \ket{\psi_Y},
\end{equation}
where the subscripts denote the corresponding party.
Therefore, if the description of the whole system $XY$ does not follow Eq.~(\ref{EQ_puresep}), the parties $X$ and $Y$ are entangled. 
For mixed states, separability is defined as a mixture of product states
\begin{equation}\label{EQ_mixedsep}
\rho_{\text{sep}}=\sum_j p_j\: \rho_X^j\otimes \rho_Y^j,
\end{equation}
where $\{p_j\}$ form a probability distribution.
States which are not of this form therefore possess quantum entanglement.
We note an important property of quantum entanglement -- monotonicity, i.e., it cannot be created or increased on average under local operations and classical communication (LOCC).
It is clear that both the states in (\ref{EQ_puresep}) and (\ref{EQ_mixedsep}) can be prepared via LOCC. 

As examples, we list below the Bell states (maximally entangled two-qubit states)
\begin{equation}\label{EQ_bellss}
\ket{\psi_{\pm}}=\frac{1}{\sqrt{2}}\left(\ket{01}\pm\ket{10} \right),\:\:\:\: \ket{\phi_{\pm}}=\frac{1}{\sqrt{2}}\left(\ket{00}\pm\ket{11} \right).
\end{equation} 
Note the simplified notation, e.g., $\ket{0}_X\otimes \ket{1}_Y=\ket{01}$, will be used throughout this thesis.
In general, some states are more entangled than the others, and therefore are more useful for performing tasks.
Below we provide some exemplary entanglement quantifiers.

\subsubsection{Entropy of entanglement}

For pure states, one can quantify the amount of entanglement using the entropy of entanglement, which is defined below.
Schmidt decomposition allows one to write a pure state $\ket{\psi}$ of system $XY$ as
\begin{equation}
\ket{\psi}=\sum_j \lambda_j \ket{x_j}\ket{y_j},
\end{equation}
where $\{\ket{x_j}\}$ and $\{\ket{y_j}\}$ are orthonormal on system $X$ and $Y$ respectively.
The state of each party is given by partially tracing out the other.
For example,
\begin{eqnarray}
\rho_X&=&\text{tr}_Y\left( \ket{\psi}\bra{\psi}\right)=\sum_j |\lambda_j|^2 \ket{x_j}\bra{x_j},\nonumber \\
\rho_Y&=&\text{tr}_X\left( \ket{\psi}\bra{\psi}\right)=\sum_j |\lambda_j|^2 \ket{y_j}\bra{y_j}
\end{eqnarray}
The information regarding the correlation between $X$ and $Y$ ($\lambda_j$) is carried in the states $\rho_X$ and $\rho_Y$.
The entropy of entanglement is defined as
\begin{equation}
E_{X:Y}=S(\rho_X)=S(\rho_Y).
\end{equation}

\subsubsection{Relative entropy of entanglement}

By using the so-called distance approach based on relative entropy, one can quantify the amount of entanglement in a state as its distance from the closest separable state.
In particular, we have the relative entropy of entanglement (REE)
\begin{equation}
E_{X:Y}=\inf_{\sigma\in \mathcal{S}} S(\rho||\sigma),
\end{equation}
where the relative entropy $S(\rho||\sigma)=-\text{tr}(\rho \log(\sigma))-S(\rho)$ and $\mathcal{S}$ is a convex set which contains separable states of Eq.~(\ref{EQ_mixedsep})~\cite{vedral1997quantifying}.
Note that for pure states, REE reduces to the entropy of entanglement.
In general, REE is hard to compute as it involves optimisation. 

\subsubsection{Negativity and logarithmic negativity}

For separable states in Eq.~(\ref{EQ_mixedsep}), it has been shown that the eigenvalues are positive after partial transposition with respect to one party~\cite{peres1996}.
This is commonly referred to as the PPT criterion.
If a state is not PPT, one argues that the corresponding negative eigenvalues quantify the amount of entanglement in the system.
These negative eigenvalues show that the corresponding density matrix after partial transposition is not physical.
In particular, one defines a computable quantifier of entanglement known as negativity
\begin{equation}
N_{X:Y}=\frac{||\rho^{T_X}||_1-1}{2},
\end{equation}
where $||\rho^{T_X}||_1$ is the trace norm of the state after partial transposition with respect to party $X$~\cite{zyczkowski1998volume,lee2000entanglement,lee2000partial,negativity}.\footnote{Note that the partial transposition can be done with respect to party $Y$, i.e., it is symmetrical.}

A logarithmic quantifier, which stems from negativity, can be defined as 
\begin{equation}
L_{X:Y}=\log_2(||\rho^{T_X}||_1)=\log_2\left( 2N_{X:Y}+1\right),
\end{equation}
and is known as logarithmic negativity.
In this thesis, we will mostly used the logarithmic negativity to quantify quantum entanglement between CV systems.
In this case, the expression is given by
\begin{equation}
L_{X:Y}=\max \big \{0,-\log_2{(2\tilde \nu_{\min})}\big \}, 
\end{equation}
where $\tilde \nu_{\min}$ is the minimum symplectic eigenvalue after partial transposition of the state (or covariance matrix) with respect to one party~\cite{negativity, adesso2004extremal}.
Similar to the PPT criterion, entanglement is shown by the violation of the physicality requirement of the covariance matrix, i.e., $\tilde \nu_{\min}<1/2$.

\subsection{Quantum discord}

A quantum state can easily be perturbed by measurements. 
Classical states are distinguishable, and therefore should be represented by orthogonal states.
For a bipartite system, the so-called quantum-classical (QC) state is defined as
\begin{equation}\label{EQ_qcstate}
\rho_{\text{qc}}=\sum_j p_j \:\rho_X^j \otimes \ket{y_j}\bra{y_j},
\end{equation}
where $\{\ket{y_j}\}$ form an orthonormal basis.
We note that there exists a von Neumann measurement on the side of $Y$ that does not perturb the whole system, i.e., $\sum_j \Pi_Y^j \rho \Pi_Y^j=\rho$, where $\Pi_Y^j=\ket{y_j}\bra{y_j}$ is a projector and $\sum_j \Pi_Y^j=\openone_Y$.\footnote{To clarify, this is a \emph{non-selective} von Neumann measurement and will be referred to as simply von Neumann measurement henceforward.}
Other states will be perturbed by any von Neumann measurement, characterising the quantumness in the correlation.
Quantum discord is a non-classical correlation of states that cannot be written as in Eq.~(\ref{EQ_qcstate})~\cite{discord1,discord2}.
We note that entanglement is a stronger type of quantum correlation in the following sense.
As any quantum-classical state in Eq.~(\ref{EQ_qcstate}) is a type of separable state in Eq.~(\ref{EQ_mixedsep}), all entangled states contain quantum discord, but the reverse is not always true.
As an example, consider the state
\begin{equation}\label{EQ_dexp}
\rho=\frac{1}{2}\big(\ket{0}\bra{0}\otimes \ket{0}\bra{0}+\ket{1}\bra{1}\otimes \ket{+}\bra{+} \big).
\end{equation}
This state is clearly separable, but is discorded since it cannot be written as (\ref{EQ_qcstate}) (with the orthonormal basis on the side of system $Y$).
One can see that by exchanging the system $X\leftrightarrow Y$ in Eq.~(\ref{EQ_dexp}), the resulting state is of the form (\ref{EQ_qcstate}) and therefore has zero discord. 
This emphasises that discord is not symmetrical.



\subsubsection{Relative entropy of discord}

Based on two expressions of mutual information that are equal in the classical regime but give different values in the quantum regime, a quantifier of quantum discord was first presented as the difference~\cite{discord2}.
In this thesis, however, we will use another form of quantifier using the distance approach based on relative entropy.
In particular, quantum discord is a distance from a given state $\rho$ to the closest quantum-classical state of Eq.~(\ref{EQ_qcstate}), i.e.,
\begin{equation}
D_{X|Y}=\inf_{\sigma \in \mathcal{C}}S(\rho||\sigma),
\end{equation}
where $\mathcal{C}$ is a non-convex set containing the quantum-classical states.
This quantifier is known as the relative entropy of discord (RED)~\cite{modi2010unified} or the one-way quantum deficit~\cite{deficit}.
Note that one can also rewrite the expression of discord as an optimisation over von Neumann measurement on the side of one party, i.e., $D_{X|Y}=\inf_{\{\Pi_Y^j\}}S(\rho||\sum_j\Pi_Y^j\rho \Pi_Y^j)$.
Consequently, this means that for the states in (\ref{EQ_qcstate}), there is always a measurement $\{\Pi_Y^j=\ket{y_j}\bra{y_j}\}$ (not perturbing the system) that minimises the RED.
It is also apparent that RED is not a symmetrical quantity, i.e., in general, $D_{X|Y}\ne D_{Y|X}$.
This can be seen in the state of Eq.(\ref{EQ_dexp}), where $D_{X|Y}$ is nonzero while $D_{Y|X}$ equals zero.
For pure states, RED reduces to the entropy of entanglement.

\subsubsection{The flags condition}

Here we note a useful property for the states of Eq.~(\ref{EQ_qcstate}).
Let us consider three quantum objects $A$, $B$, and $C$ with the state written as
\begin{equation}
\rho=\sum_j p_j\: \rho_{AB}^j \otimes \ket{c_j}\bra{c_j}.
\end{equation}
The entanglement of this quantum-classical state follows
\begin{equation}
E_{A:BC}(\rho)=\sum_j p_j\: E_{A:B}(\rho_{AB}^j).
\end{equation}
This equation is known as the flags condition~\cite{flags}. 
Furthermore, if one uses the REE as entanglement quantifier, the closest separable state to $\rho$ reads
\begin{equation}
\sigma=\sum_j p_j\: \sigma_{AB}^j\otimes \ket{c_j}\bra{c_j},
\end{equation}
where $\sigma_{AB}^j$ is the closest separable state to $\rho_{AB}^j$~\cite{flags}.

\subsubsection{Classicality \& non-classicality}

This section is dedicated to setting the notion of classicality and non-classicality that will be used throughout this thesis.
These definitions are correlation-based in the sense that they are related to the type of correlation between the object under consideration and other objects.

We will use the term classical for an object if there is a von Neumann measurement on it that does not perturb the whole (bigger) system. 
As an example, we say that system $Y$ is classical for the QC states in Eq.~(\ref{EQ_qcstate}), which have zero discord.
Note that this way, classicality is not a property of an object alone, rather it is a property of a bigger system, with the object as one of its components. 

Note that if we were to consider only the object in question, i.e., by tracing out the remaining others, one can always find a spectral decomposition of the corresponding density matrix such that it is expressed as a sum of weighted orthogonal (a signature of classicality) projectors.
Therefore, a single object (regardless of having correlations with other objects) is always classical at a particular time in a particular \emph{reference basis}.
Should this basis change in time, the system is developing coherence, which is quantum in nature \cite{roqcoherence}. 
This is simply shown by nonzero off-diagonal elements in the density matrix, as represented in the \emph{reference basis}.

In contrast with classicality, non-classicality is related to discorded states.
Therefore, being ``quantum" or non-classical is associated with the ability to have non-classical correlations with other objects.
Furthermore, for discorded states, e.g., with $D_{X|Y}\ne 0$, one can prepare quantum coherence on the side of $Y$ by only operating on system $X$ \cite{coherence-distillation}.
In this way, our notion of non-classicality can be thought of as a property of an object alone. 
Also, as entangled objects possess quantum discord, they are automatically non-classical.

\subsection{Mutual information}

In addition to non-classical correlations mentioned above, quantum systems can be classically correlated. 
As a measure of total correlation, we use the mutual information
\begin{equation}\label{EQ_minfo}
I_{X:Y}=S(\rho_X)+S(\rho_Y)-S(\rho),
\end{equation}
where $\rho_X$ and $\rho_Y$ are marginals (reduced states) of $\rho$~\cite{groisman2005}.
One can also phrase the mutual information using the distance approach. 
In particular, one writes $I_{X:Y}=\inf_{\sigma \in \mathcal{M}} S(\rho||\sigma)$, where $\mathcal{M}$ is a non-convex set containing product states.
As the closest product state to $\rho$ is simply given by its marginals, the expression simplifies to Eq.~(\ref{EQ_minfo})~\cite{modi2010unified}.
Also note that the mutual information is equal to twice the entropy of entanglement for pure states.

Finally, all considered correlations are summarised in Fig.~\ref{FIG_ch1_uni}.
In this unified view, one can clearly see that $E_{X:Y}\le D_{X|Y} \le I_{X:Y}$.
This is a consequence of the relation for the sets $\mathcal{M}\in \mathcal{C} \in \mathcal{S}$.

We should also mention that the presence of quantum discord, e.g., $D_{X|Y}\ne 0$, means that the total correlations between systems $X$ and $Y$, i.e., mutual information, cannot all be accessed by measuring system $Y$ alone \cite{discord2}.
In the context of this thesis, we will show how to reveal quantum discord without performing measurements on system $Y$, which is one of the main results in Chapter~\ref{Chapter_revealing}.

\begin{figure}[h]
\centering
\includegraphics[scale=0.6]{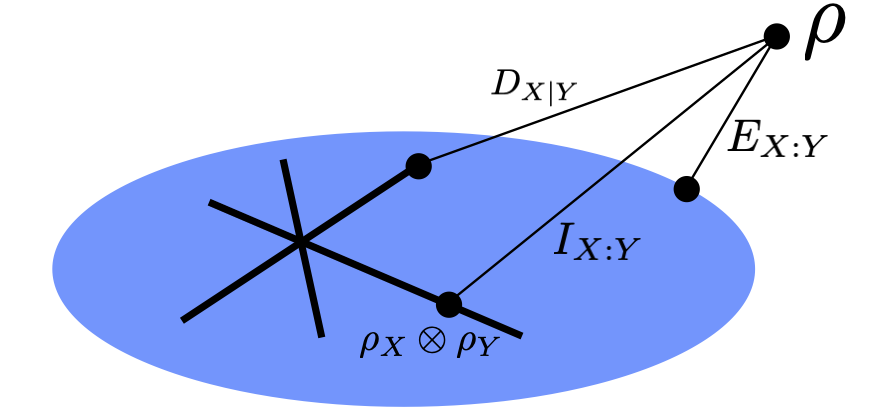} 
\caption{Geometrical illustration of correlations. The correlation quantifiers shown include relative entropy of entanglement $E_{X:Y}$, relative entropy of discord $D_{X|Y}$, and mutual information $I_{X:Y}$. The blue shaded area represents the convex set of separable states and the solid lines represent the non-convex set of quantum-classical states that also consist of product states, e.g., $\rho_X\otimes \rho_Y$.}
\label{FIG_ch1_uni} 
\end{figure}

\subsubsection{Conditional entropy}

Conditional entropy is a useful quantity that we will utilise later in this thesis. 
It is defined as 
\begin{equation}
S_{X|Y}=S(\rho)-S(\rho_Y).
\end{equation}
We note that for separable states in Eq.~(\ref{EQ_mixedsep}), the conditional entropy follows $S_{X|Y},S_{Y|X}\ge 0$~\cite{edmundfyp}. 
Consequently, the negative value, $S_{X|Y},S_{Y|X}<0$, implies quantum entanglement between party $X$ and party $Y$.
This way, one notices a difference between the classical regime where the conditional entropy is always non-negative and the quantum regime where it can be negative. 
For example, any pair of maximally entangled qubits in Eq.~(\ref{EQ_bellss}) gives $S_{X|Y}=S_{Y|X}=-1$.



\chapter{Revealing non-classicality of inaccessible objects} 

\label{Chapter_revealing} 

\lhead{Chapter 3. \emph{Revealing non-classicality of inaccessible objects}} 
\emph{In this chapter, the resource for entanglement gain in a scenario where two principal objects are continuously interacting via a mediating system will be investigated.\footnote{Parts of this chapter are reproduced from our published article of Ref.~\cite{krisnanda2017}, $\copyright$ [2019] American Physical Society. Where applicable, changes made will be indicated.}
I begin with an application-related motivation of showcasing quantum features of unknown objects that are not accessible by direct experimentation and a construction of a general scenario that will be considered throughout this chapter. 
Next, by taking non-correlated systems as the initial condition, I will show that quantum discord between the mediator and the principal objects is a necessary requirement.
On the other hand, an interesting scenario will be presented where the entanglement is distributed via a classical mediator, i.e., zero discord. 
This scheme, however, needs some correlation already present in the initial state.
By utilising this knowledge, I will devise a detection method capable of revealing a quantum property of an inaccessible object, as quantified by the quantum discord.
Our protocols use minimalistic assumptions: no information is required regarding the dimensionality of the objects involved, the initial state, and the explicit form of the interactions.
Also, all the objects can be open systems and no measurements are applied on the inaccessible object.
}

\clearpage

\section{Motivation and objectives}

What should be known about an inaccessible object to conclude that it is ``not classical"? 
In this thesis, inspired by quantum communication scenarios, I show the verification that such object can be used to increase quantum entanglement between remote probing particles (that individually interact with it) is a sufficient indication.

Specifically, we prove that such gain in quantum entanglement is only possible if, during its evolution, the object shares with the probes quantum correlations in the form of quantum discord~\cite{discord1,discord2,discord3,discord4,discord5}.
In turn, the presence of quantum discord between the probes and the object entails a non-classical feature of the object itself.
According to the definition of discord, two or more subsystems share quantum correlations if there is no von Neumann measurement on one of them that keeps the total state unchanged.
This can only happen when non-orthogonal (indistinguishable) states are involved in the description of the physical configuration of the measured subsystem.
This indistinguishability is the non-classical feature that we aim to detect. 
We formulate analytical criteria revealing such non-classicality based on operations performed only on the probing objects, and without any detailed modelling of the inaccessible object in question.

We emphasise that the non-classicality is revealed under a set of minimal assumptions.
Namely: (i) The object may remain inaccessible at all times, i.e., it needs not be directly measured. 
In particular its quantum state and Hilbert space dimension can remain unknown throughout the whole assessment. 
Our method is thus valid when the object is an elementary system or an arbitrarily complex one;
(ii) The details of the interactions between the object and the probes may also remain unspecified;
(iii) Every party can be open to its own local environment.
These properties make our method applicable to a large number of experimentally relevant situations. 
We demonstrate the revealing power of our criteria in Chapters \ref{Chapter_gravity}, \ref{Chapter_probing}, and \ref{Chapter7} for the detection of quantum property of gravitational interaction, photosynthetic organisms, optomechanical mirror in the membrane-in-the-middle setting, and more.


The general scenario we consider in this chapter is depicted in Fig.~\ref{FIG_re_ABC}.
System $C$ is assumed to be the inaccessible object and to mediate the interaction between two remote probes, labeled $A$ and $B$. Therefore, here, we refer to system $C$ as the {\it mediator}.
It is essential for our method that the probes are not directly coupled and only interact via the mediator. 
This means that the Hamiltonian for the process under scrutiny can be written as $H_{AC} + H_{BC}$, with $H_{JC}$ being the interaction Hamiltonian between the mediator $C$ and probe $J=A,B$.

\begin{figure}[h]
\centering
\includegraphics[scale=0.6]{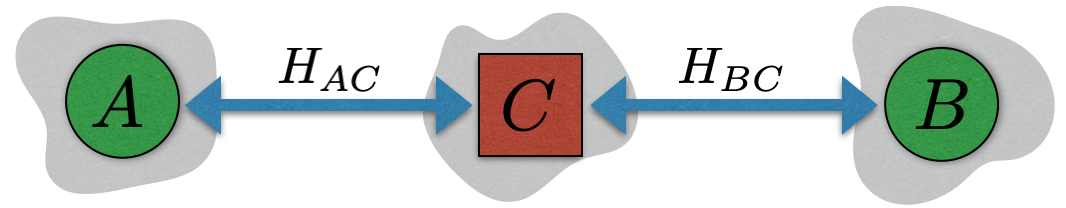} 
\caption{General setting for revealing non-classicality of inaccessible objects. 
An inaccessible object $C$ is mediating interactions between otherwise non-interacting probing objects $A$ and $B$. 
In general, the state of $C$ can be unknown and no measurement can be conducted on it. 
In the text, we present conditions for which entanglement gain between $A$ and $B$ implies non-classicality of $C$, in the form of quantum discord $D_{AB|C}$.
Note that our methods do not assume the dimensionality of any object and the explicit form of the interacting Hamiltonians $H_{AC}+H_{BC}$. 
Also, all the objects can be open to their own local environment (grey shaded area).}
\label{FIG_re_ABC} 
\end{figure}

\section{Instrumental discord}

Our work is developed in the context of entanglement distribution with continuous interactions~\cite{cubitt}.
We first focus on the partition $A:BC$ and demonstrate a result which will be instrumental to design our criteria for the inference of non-classicality of $C$ based on entanglement dynamics in $AB$ only.
Recall that previous studies on the resources allowing for entanglement distribution showed that any three-body density matrix, i.e., the state of $ABC$ at any time $t$ in the present context, satisfies the inequality~\cite{bounds1, bounds2}:
\begin{equation}
\label{ebound}
|E_{A:BC}(t) - E_{AC:B}(t) | \le D_{AB|C}(t).
\end{equation}
Here $E_{X:Y}$ is the relative entropy of entanglement in the partition $X:Y$~\cite{vedral1997quantifying}, and $D_{X|Y}$ is the relative entropy of discord~\cite{modi2010unified}, also known as the one-way quantum deficit~\cite{deficit}.
Note that relative entropy of discord is in general not symmetric, i.e., $D_{X|Y} \ne D_{Y|X}$.
Eq.~(\ref{ebound}) shows that the change in entanglement due to the relocation of $C$ is bounded by the quantum discord carried by it.

Let us start from the simple case where the overall probes-mediator system is closed (which allows us to ignore for now the grey-colored shadows in Fig.~\ref{FIG_re_ABC}).
If the interaction Hamiltonians $H_{JC}$ satisfy $[H_{AC},H_{BC}] = 0$, the evolution operator from the initial time $t=0$ to some finite time $\tau$ is just 
$U = U_{BC} U_{AC}$, where $U_{JC} = \exp{(- i H_{JC} \tau/\hbar)}$.
This situation is equivalent to first interacting $C$ with $A$ and then $C$ with $B$ (or in reversed order).
However, note that the density matrix $\rho' = U_{AC} \rho_0 U_{AC}^\dagger$ obtained by ``evolving'' the initial state through $U_{AC}$ only does not describe the state of the system at $\tau$.
Nevertheless, we now show the relevance of the properties of state $\rho'$ for entanglement gain.

Consider the following forms of Eq.~(\ref{ebound}) written for the initial state $\rho_0$ and the instrumental state $\rho'$, respectively
\begin{eqnarray}
E_{AC:B}(0) - E_{A:BC}(0) &\le& D_{AB|C}(0), \nonumber \\
E_{A:BC}' - E_{AC:B}' &\le& D_{AB|C}'.
\end{eqnarray}
Note that $E_{AC:B}(0) = E_{AC:B}'$, because $U_{AC}$ is a local unitary operator in this partition.
The state at time $\tau$ is given by $\rho_\tau = U_{BC} \rho' U_{BC}^\dagger$, and thus $E_{A:BC}(\tau) = E_{A:BC}'$, this time owing to $U_{BC}$ being local.
Summing the above inequalities we obtain a bound on the entanglement gain
\begin{equation}
\label{ebound_finite}
E_{A:BC}(\tau) - E_{A:BC}(0) \le D_{AB|C}(0) + D_{AB|C}'.
\end{equation}
This opens up the possibility to create entanglement at time $\tau$ without producing discord at both $t=0$ and $\tau$, but rather by utilising non-classicality in the instrumental state.
In other words, the gain of entanglement in $A:BC$ could be mediated by object $C$, which gets non-classically correlated by $U_{AC}$ 
and then decorrelated by $U_{BC}$. Therefore, $C$ is only classically correlated at times $t=0$ and $\tau$. 
We now give a concrete example of this type of entanglement creation.

Consider all the objects to be qubits and take the interaction Hamiltonian as
\begin{equation}
\label{xhamiltonian}
H=(\sigma^x_A\otimes \openone \otimes \sigma^x_C+\openone \otimes \sigma^x_B \otimes \sigma^x_C)\:\hbar \omega ,
\end{equation}
where $\sigma^j~(j=x,y,z)$ is the Pauli-$j$ matrix and $\omega$ is the frequency. 
As initial state we choose the classically correlated state
\begin{eqnarray}
\rho_0 &=& \tfrac{1}{2}| 011 \rangle \langle 011 | + \tfrac{1}{2}| 100 \rangle \langle 100 |,
\end{eqnarray}
where, e.g., $\sigma^z|0\rangle=|0\rangle$. 
One can now readily check that the relative entropy of entanglement $E_{A:BC}$ grows from $0$ to $1$ in the timespan from $T=0$ to $T = \pi/4$,\footnote{For simplicity, we define a dimensionless time variable $T=\omega t$.} whereas discord $D_{AB|C}$ remains zero at these two times. 
The gain is indeed due to non-classical correlations of the instrumental state: applying only $U_{AC}$ for a time $T = \pi/4$ produces discord $D_{AB|C}' = 1$. 
A summary of this dynamics is presented in Fig. \ref{FIG_re_exp1}.

\begin{figure}[h]
\centering
\includegraphics[scale=0.55]{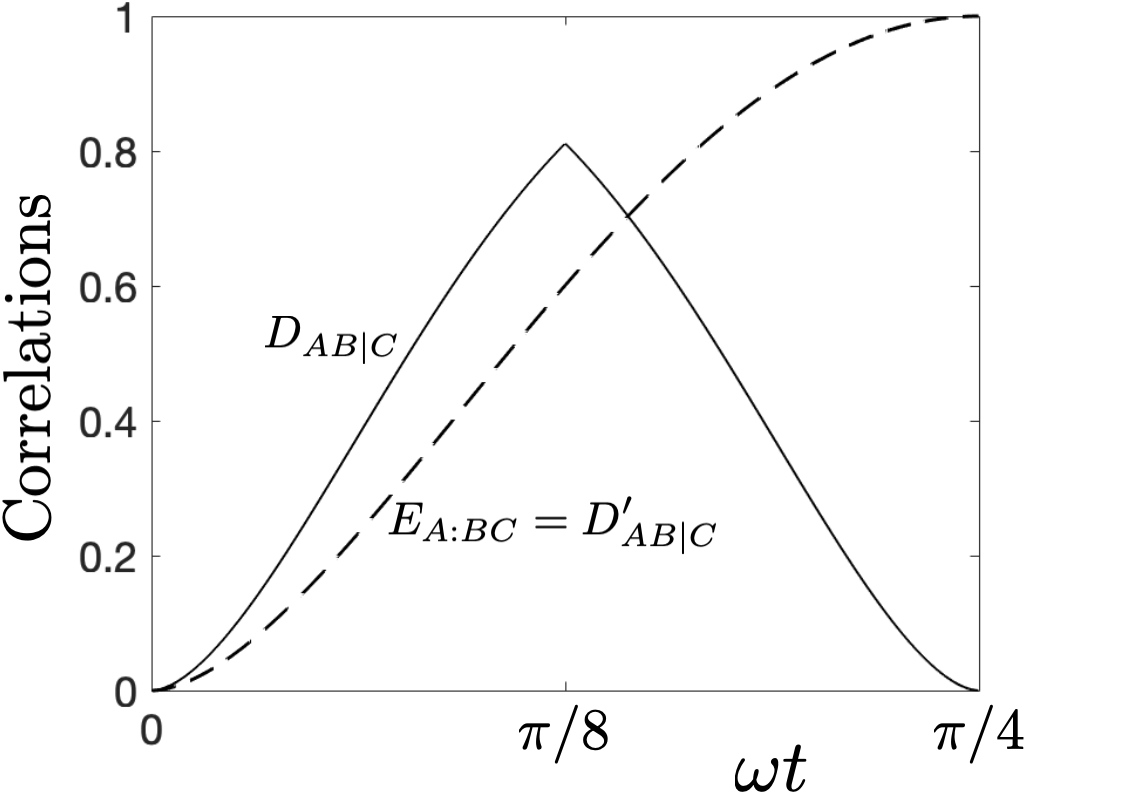} 
\caption{Entanglement distribution utilising instrumental discord. 
The Hamiltonian is taken to be $H=(\sigma^x_A\otimes \openone \otimes \sigma^x_C+\openone \otimes \sigma^x_B \otimes \sigma^x_C)\: \hbar \omega$ with initial state $\tfrac{1}{2}| 011 \rangle \langle 011 | + \tfrac{1}{2}| 100 \rangle \langle 100 |$.
Note the maximum quantum discord is $D_{AB|C}\approx 0.81$ at $\omega t=\pi/8$.
This shows that the entanglement gain is not bounded by the maximum discord.}
\label{FIG_re_exp1} 
\end{figure}

For completeness, I provide the detailed calculations below.
First, we note that the Hamiltonian has commuting components, i.e., $[H_{AC},H_{BC}]=0$ and therefore we write the evolution operator as $U=U_{BC}U_{AC}$.
Next, entanglement in the partition $A:BC$ is calculated as follows
\begin{eqnarray}
E_{A:BC}(T)&=&E_{A:BC}^{\prime} \nonumber \\
&=&\frac{1}{2}E_{A:C}(U_{AC} \ket{01}\bra{01}U_{AC}^{\dagger})+\frac{1}{2}E_{A:C}(U_{AC} \ket{10}\bra{10}U_{AC}^{\dagger}),
\end{eqnarray}
where now $E_{A:BC}^{\prime}=E_{A:BC}(U_{AC} \rho_0 U_{AC}^{\dagger})$ for a time $T$ and the steps are justified as follows.
In the first line, we use the fact that $U_{BC}$ is local in the partition $A:BC$.
Next, by applying $U_{AC}$ operator to the initial state, the state of $B$ is still classical, i.e., represented by an orthogonal basis $\{\ket{0},\ket{1}\}$.
Therefore, we arrive at the second line with the help of the flags condition \cite{flags}.
We note that the states in the second line are pure and one can equate the entanglement $E_{A:C}$ with the corresponding von Neumann entropy of object $A$ or $C$. 
As an example, by using $U_{AC}=\cos(T)\openone-i\sin(T)\sigma_A^x\sigma_C^x$, one has for the first term
\begin{eqnarray}
\text{tr}_B\left(U_{AC} \ket{01}\bra{01}U_{AC}^{\dagger}\right)=\cos^2(T)\ket{0}\bra{0}+\sin^2(T)\ket{1}\bra{1},
\end{eqnarray}
giving an entropy $-\sin^2(T)\log_2(\sin^2(T))-\cos^2(T)\log_2(\cos^2(T))$. 
The calculation for the second term gives the same result, and therefore the entanglement dynamics is simply given by
\begin{equation}
E_{A:BC}(T)=-\sin^2(T)\log_2(\sin^2(T))-\cos^2(T)\log_2(\cos^2(T)).
\end{equation}

Now I will show that the evolution of $D_{AB|C}^{\prime}$ is equal to $E_{A:BC}(T)$ as follows.
First, note that $E_{A:BC}(T)=E_{A:BC}^{\prime}=E_{AB:C}^{\prime}$, where the second equality is from the flags condition as a result of classical $B$.
By showing that the state $\rho^{\prime}$ has the closest separable state $\sigma_{AB:C}^{\prime}$ that is classical on $C$, we confirm that $D_{AB|C}^{\prime}=E_{AB:C}^{\prime}=E_{A:BC}(T)$. 
In order to do this, we recall that the flags condition also states that the closest separable state of a quantum--classical density matrix $\rho^{\prime}$ (classical $B$) is given by 
\begin{equation}
\sigma_{AB:C}^{\prime}=\frac{1}{2}\sigma_{A:C}^{\prime1} \otimes \ket{1}_B\bra{1}+\frac{1}{2}\sigma_{A:C}^{\prime2}\otimes \ket{0}_B\bra{0},
\end{equation}
where $\sigma_{A:C}^{\prime1}$ and $\sigma_{A:C}^{\prime2}$ are the closest separable states respectively to $U_{AC} \ket{01}\bra{01}U_{AC}^{\dagger}$ and $U_{AC} \ket{10}\bra{10}U_{AC}^{\dagger}$. 
Next, we use the fact that the closest separable state of a pure state is given by the density matrix that is de-phased in the local Schmidt basis, to arrive at $\sigma_{A:C}^{\prime1}=\cos^2(T)\ket{01}\bra{01}+\sin^2(T)\ket{10}\bra{10}$ and $\sigma_{A:C}^{\prime2}=\cos^2(T)\ket{10}\bra{10}+\sin^2(T)\ket{01}\bra{01}$.
By putting the expressions together, one realises that $\sigma_{AB:C}^{\prime}$ is also classical on the side of $C$, hence proving the claimed statement above.

Finally, the dynamics of $D_{AB|C}$ is computed numerically. 
We note that the von Neumann measurement on $C$ minimising the discord is given by the basis $\{\ket{0},\ket{1}\}$ for $T=[0,\pi/8]$ and $\{\ket{y_+},\ket{y_-}\}$ (the eigenbasis of $\sigma^y$) for $T=[\pi/8,\pi/4]$.

For general non-commuting interaction Hamiltonians, one can pursue a similar analysis with the help of the Suzuki-Trotter expansion (see Appendix \ref{A_trotter} for a simple proof).
The evolution operator $U$ is now discretised into short-time interactions of $C$ with $A$ and then $B$ (or the reversed order) as
\begin{eqnarray}
U&=& \lim_{n\rightarrow \infty} \left(e^{-i H_{BC}\Delta{t}/\hbar}e^{-i H_{AC}\Delta{t}/\hbar}\right)^n ,
\end{eqnarray}
where $\Delta t = \tau /n \rightarrow 0$. 
Accordingly, Eq.~(\ref{ebound_finite}) holds with $\tau$ replaced by $\Delta t$.
It is now natural to ask if a scenario exists where entanglement could be increased via interactions with a classical $C$ \emph{at all times} by exploiting the discord in the instrumental state.
The example given above is not of this sort because, although we have $D_{AB|C}=0$ at $t=0$ and $\tau$, it is non-zero for $t \in (0,\tau)$.
It turns out that, for short evolution times, the discord of the instrumental state cannot be exploited as the theorem in the following section demonstrates.

\section{Quantum discord is a resource for entanglement gain}

This section focuses on showing that quantum discord $D_{AB|C}$ is necessary for entanglement gain in the partition $A:BC$. 
On the other hand, entanglement in the partition $A:B$ requires that the mediator is correlated with the principal objects (might be in the form of classical correlation).
We begin with the following theorem.

\begin{theorem}\label{TH_revealing}
Consider two quantum objects $A$ and $B$ that are interacting via a mediator $C$, but not with each other.
The entanglement follows 
\begin{equation}
E_{A:BC}(\tau)\le E_{A:BC}(0)
\end{equation}
if the state of the system is quantum--classical in the partition $AB:C$ at all times, i.e., $D_{AB|C}(t)=0$ for $t\in[0,\tau]$.
Note that all objects are allowed to interact with their own local environments.
\end{theorem}
\begin{proof}
For simplicity, the Hamiltonian of the system is taken as $H=H_{AC}+H_{BC}$, where $H_{AC}=H_{A}\otimes H_{C_1}$ and $H_{BC}=H_{B}\otimes H_{C_2}$.\footnote{Note that the term $H_{C_1}$ can be different from $H_{C_2}$ in general.}
We note that it is straightforward to extend the proof for more general Hamiltonians $H_{AC} = \sum_{\mu} H_A^{\mu} \otimes H_{C_1}^{\mu}$ and $H_{BC} = \sum_{\nu} H_B^{\nu} \otimes H_{C_2}^{\nu}$. 

As we assume, at all times, quantum--classical state in the partition $AB:C$, one writes the initial density matrix of the system as 
\begin{equation}
\rho_0 = \sum_{c} p_c \: \rho_{AB|c} \otimes \ket{c}\bra{c},
\end{equation}
and the state after a short time $\Delta t$ as
\begin{equation}
\rho_{\Delta t} = \sum_{c} p_c(\Delta t) \: \rho_{AB|c}(\Delta t) \otimes \ket{\phi_c}\bra{\phi_c},
\end{equation}
where $\{\ket{c}\}$ and $\{ \ket{\phi_c} \}$ both form orthonormal bases of the mediating system $C$.

Now we incorporate the local environment of all the objects under scrutiny. 
In particular, the dynamics of the whole system assumes the following coarse-grained master equation in Lindblad form
\begin{equation}\label{EQ_LE}
\frac{\rho_{\Delta t}-\rho_0}{\Delta t}=-\frac{i}{\hbar}[H,\rho_0]+\sum_{X=A,B,C}L_X \rho_0,
\end{equation}
where the first term on the RHS is the coherent evolution while the second describes incoherent interactions with the local environments. 
We recall that $L_X\rho_0\equiv \sum_k Q_X^k \rho_0 Q_{X}^{k\dag}-\frac{1}{2}\{Q_{X}^{k\dag}Q_X^k,\rho_0\}$, with the operator $Q_X^k$ only acting on object $X$.
Note that the strength of the interactions with environments have been absorbed in the operator $Q_X^k$.

By utilising the master equation (\ref{EQ_LE}), one obtains the following conditional state after a short time $\Delta t$
\begin{eqnarray}
p_j(\Delta t) \: \rho_{AB|j}(\Delta t) &=&\bra{\phi_j}  \rho_{\Delta t}\ket{\phi_j}\nonumber \\
&=& \bra{\phi_j} \rho_0 \ket{\phi_j} -i \frac{\Delta t}{\hbar} \bra{\phi_j} [H,\rho_0] \ket{\phi_j}\nonumber \\
&&+\Delta t  \sum_{X=A,B,C} \bra{\phi_j} L_X \rho_0\ket{\phi_j}.
\label{EQ_COMM}
\end{eqnarray}
We note that the continuous dynamics follows by taking the limit $\Delta t\rightarrow 0$ and applying the short time evolution successively. 
In this proof, terms of the order $\mathcal{O}(\Delta t^2)$ will not be taken into account, and we use ``$\simeq$" to denote equality where the $\mathcal{O}(\Delta t^2)$ terms are irrelevant and ignored.
In this short time, we write the change in the basis of $C$, up to $\Delta t$ order, as $\ket{c} \to \ket{\phi_c} = \alpha_c \ket{c} + \beta_c \ket{c_\perp}$, where $\ket{c_\perp}$ and $\ket{c}$ are orthogonal, and $\beta_c\propto \Delta t$.
As a result, one gets $|\langle c | \phi_j \rangle|^2 \simeq \delta_{cj}$, making $\bra{\phi_j} \rho_0 \ket{\phi_j} \simeq p_j\: \rho_{AB|j}$.

The coherent part of the dynamics in Eq.~(\ref{EQ_COMM}) is proportional to $\Delta t$, and therefore, we render $\Delta t$ order terms in $\bra{\phi_j} [H,\rho_0] \ket{\phi_j}$ irrelevant. 
As a result, we have
\begin{equation}
 \bra{\phi_j} [H,\rho_0] \ket{\phi_j} \simeq p_j [E_{C_1}^j H_A + E_{C_2}^j H_B,\rho_{AB|j}] ,
\end{equation}
where we define $E_{C_1(C_2)}^j\equiv \bra{j} H_{C_1(C_2)} \ket{j}$ as the expectation value of energy.\footnote{For clarity, we note that one can make these expectation values dimensionless such that the energy unit is absorbed in the terms $H_A$ and $H_B$ respectively.}
Note that, in this way, the coherent component is simply equivalent to an effective dynamics that are local in $A$ and $B$.

On the other hand, the incoherent part in Eq.~(\ref{EQ_COMM}) consists of three terms.
The first two in the summation are given by
\begin{equation}
\bra{\phi_j} L_A \rho_0+L_B \rho_0\ket{\phi_j}\simeq p_j(L_A+L_B)\rho_{AB|j},
\end{equation}
and they are local operations with respect to object $A$ and object $B$.
The last one reads
\begin{eqnarray}
\bra{\phi_j} L_C \rho_0\ket{\phi_j}&\simeq&\sum_c \sum_k p_c |\bra{j}Q_C^k\ket{c}|^2 \rho_{AB|c}\nonumber \\
&-&p_j\sum_k  \bra{j}Q^{k\dag}_{C}Q_C^k\ket{j} \rho_{AB|j}.
\end{eqnarray}

Plugging all these findings to Eq.~(\ref{EQ_COMM}), we write the conditional state explicitly as 
\begin{eqnarray}
\rho_{AB|j}(\Delta t)& \simeq &  \frac{p_j(1-\sum_k  \bra{j}Q^{k\dag}_{C}Q_C^k\ket{j}\Delta t )}{p_j(\Delta t)} \, \, \tilde \rho_{AB|j} \nonumber \\
  & + & \sum_c\frac{p_c}{p_j(\Delta t)}\sum_k |\bra{j}Q_C^k\ket{c}|^2\Delta t \: \rho_{AB|c},
\label{EQ_US}
\end{eqnarray}
where we define
\begin{eqnarray}
\tilde \rho_{AB|j} & \equiv & \rho_{AB|j} -i[E_{C_1}^j H_A + E_{C_2}^j H_B,\rho_{AB|j}] \frac{\Delta t}{\hbar} \nonumber \\
& + & (L_A+L_B)\rho_{AB|j} \Delta t.
\end{eqnarray}
Therefore, the density matrix $\tilde \rho_{AB|j}$ gives the evolution of $\rho_{AB|j}$ through an effective Lindblad master equation, where the interactions are local on the side of $A$ and $B$.
By using $\text{tr}(\rho_{AB|j}(\Delta t))=1$, one obtains the probability 
\begin{eqnarray}
p_j(\Delta t)&=& p_j \left( 1-\sum_k  \bra{j}Q^{k\dag}_{C}Q_C^k\ket{j}\Delta t \right)\nonumber \\
  &&+\sum_cp_c\sum_k |\bra{j}Q_C^k\ket{c}|^2\Delta t,
\end{eqnarray}
where we have utilised the cyclic property of trace. 
Note that the factors for the states $\tilde \rho_{AB|j}$ and $\rho_{AB|c}$ in Eq. (\ref{EQ_US}) are all real, non-negative, and sum up to unity.

Finally, the entanglement in the partition $A:BC$ reads
\begin{eqnarray}
E_{A:BC}(\Delta t)&=&\sum_{j} p_{j}(\Delta t) \: E_{A:B}(\rho_{AB|j}(\Delta t)) \\
&\le&\sum_j p_j E_{A:B}(\rho_{AB|j}) \nonumber \\
&-&\sum_jp_j\sum_k \bra{j}Q^{k\dag}_{C}Q_C^k\ket{j}\Delta t \:E_{A:B}(\rho_{AB|j})\nonumber \\
&+&\sum_j\sum_cp_c\sum_k|\bra{j}Q_C^k\ket{c}|^2\Delta t \:E_{A:B}(\rho_{AB|c}),\label{EQ_LL}  \\
&=& \sum_j p_j E_{A:B}(\rho_{AB|j}) = E_{A:BC}(0), \label{EQ_LAST}
\end{eqnarray}
where we explain the steps taken above as follows.
The flags condition \cite{flags} is applied in the first line to the quantum--classical state in the partition $AB:C$.
Next, the inequality is obtained by using the convexity property of entanglement on the state $\rho_{AB|j}(\Delta t)$ in Eq. (\ref{EQ_US}), and $E_{A:B}(\tilde \rho_{AB|j})\le E_{A:B}(\rho_{AB|j})$ from the monotonicity of entanglement under local interactions.
We note that the second and third lines in Eq. (\ref{EQ_LL}) cancel out. 
This is apparent by inserting $\sum_c \ket{c}\bra{c}=\openone$ between $Q^{k\dag}_{C}$ and $Q^k_C$ in the second line and exchanging the dummy indices $c \leftrightarrow j$.
The last equality is obtained from the flags condition applied to the initial state. 
The conclusion of the theorem follows by applying the arguments above successively from $t=0$ to $\tau$ with the condition $D_{AB|C}(t)=0$ at all times.
\end{proof} 

Theorem \ref{TH_revealing} above has an important special case, which is shown by the corollary below.

\begin{corollary}
For the premise of Theorem \ref{TH_revealing}, where all the objects are closed systems, we have 
\begin{equation}
E_{A:BC}(\tau) = E_{A:BC}(0).
\end{equation}
\end{corollary}
\begin{proof}
In this case, the dynamics is simply given by a unitary operator with Hamiltonian $H=H_{AC}+H_{BC}$.
The conditional state of Eq. (\ref{EQ_US}) simplifies to
\begin{equation}\label{EQ_SC}
\rho_{AB|j}(\Delta t)=\frac{p_j}{p_j(\Delta t)} \, \, \tilde \rho_{AB|j},
\end{equation}
where $\tilde \rho_{AB|j}$ now reads $\rho_{AB|j}-i[E_{C_1}^j H_A + E_{C_2}^j H_B,\rho_{AB|j}] \Delta t/\hbar$.
One arrives at $p_j(\Delta t)=p_j$ and $\rho_{AB}^j(\Delta t)=\tilde \rho_{AB}^j$ by taking the trace of Eq.~(\ref{EQ_SC}).
Therefore, $\rho_{AB|j}(\Delta t)$ is simply an evolution from the initial state $\rho_{AB|j}$ via an effective unitary operator that is local in both $A$ and $B$, making the corresponding entanglement in the partition $A:B$ invariant. 
Finally, from the flags condition applied to the initial state and the state after $\Delta t$ one gets $E_{A:BC}(\Delta t) = E_{A:BC}(0)$.
The Corollary follows by considering successive applications of the arguments above for a time $\tau$.
\end{proof}

We emphasise the generality of Theorem \ref{TH_revealing}, where both the mediator and probing objects are open to their own local environments. 
This matches a large number of experimentally relevant situations -- some of the important ones will be addressed later in this thesis.
It is also worth stressing that this Theorem extends the monotonicity of entanglement under local operations and classical communication (LOCC)~\cite{locc} to the case of continuous interactions.
In general, zero-discord states are good models for classical communication as they allow for continuous projective measurements on $C$ that do not disturb the whole multipartite state.

\subsection{The role of correlated mediator}

Now let us consider entanglement between the principal objects $A$ and $B$.
In particular, we have the following Lemma.

\begin{lemma}\label{LM_revealing1}
Consider the premise of Theorem \ref{TH_revealing}.
It follows that the gain of entanglement in the partition $A:B$ implies the existence of correlation in the partition $AB:C$.
\end{lemma}
\begin{proof}
This is proof by contradiction. 
Let us write the initial state with an uncorrelated mediator $C$ as $\rho(0)=\rho_{AB}\otimes \rho_C$.
This gives us $E_{A:BC}(0) = E_{A:B}(0)$ from the monotonicity of entanglement under tracing out the uncorrelated object $C$.
Furthermore, if one has quantum discord satisfying $D_{AB|C}(t)=0$ for $t\in [0,\tau]$, one can utilise the conclusion of Theorem \ref{TH_revealing} and obtain
\begin{equation}
E_{A:B}(\tau)\le E_{A:BC}(\tau)\le E_{A:BC}(0)=E_{A:B}(0),
\end{equation}
where the first inequality is also a consequence of the monotonicity property.
This shows that the observation of entanglement gain $E_{A:B}(\tau)>E_{A:B}(0)$ implies that either the initial state $\rho(0)\ne \rho_{AB}\otimes \rho_C$ or $D_{AB|C}(t)\ne 0$ at some time during the dynamics. 
In either case, the entanglement gain detects the presence of correlation between $AB$ and $C$.
\end{proof}

Let us note that in many experimental settings, one usually assumes decoupled objects as the initial condition, i.e., $\rho_0=\rho_A\otimes \rho_B\otimes \rho_C$.
As this type of state is a special case of $\rho_{AB}\otimes \rho_C$ considered in Lemma \ref{LM_revealing1}, one concludes that, for initial decoupled objects, quantum discord is a necessary resource for entanglement gain between the principal objects.

\section{Entanglement localisation via classical mediators}

Before we proceed with the non-classicality detection protocols, I will show in this section an exemplary dynamics where quantum entanglement \emph{can} be mediated using a classical mediator.

Let us again consider three qubits, where the Hamiltonian is given by Eq.~(\ref{xhamiltonian}), and choose the initial state
\begin{equation}
\rho_0 = \tfrac{1}{2} \ket{\psi_+}\bra{\psi_+} \otimes \ket{+}\bra{+} + \tfrac{1}{2} \ket{\phi_+}\bra{\phi_+} \otimes \ket{-}\bra{-},
\label{EQ_CL_EX}
\end{equation}
where $\sigma^x\ket{\pm}=\pm|\pm\rangle$, and $\ket{\psi_+}= \frac{1}{\sqrt{2}}(\ket{01}+\ket{10})$ and $\ket{\phi_+} = \frac{1}{\sqrt{2}}(\ket{00}+\ket{11})$ are two Bell states between subsystems $AB$.
As the initial state in Eq.~(\ref{EQ_CL_EX}) contains the eigenstates of $H_C$, the system remains classical, as measured on $C$, at all times.
Furthermore, the classical basis is the same at all times.
Yet, one can verify that the relative entropy of entanglement between the probes is given by $E_{A:B}(T) = 1-S_{AB}(T)$, where $S_{AB}(T)$ is the von Neumann entropy of the $AB$ state at time $T$, and oscillates between $0$ and $1$.
Note this means, in general, that entanglement gain in the partition $A:B$ does \emph{not} signify the non-classicality of $C$ (nonzero $D_{AB|C}$).
The corresponding dynamics is illustrated in Fig.~\ref{FIG_re_exp2}.
One can also see from Fig.~\ref{FIG_re_exp2} that there is entanglement already present initially in the bigger partition $A:BC$.
The subsequent evolution only localises such entanglement to the $A:B$ partition.
Let us therefore call this phenomenon \emph{entanglement localisation}.

\begin{figure}[h]
\centering
\includegraphics[scale=0.55]{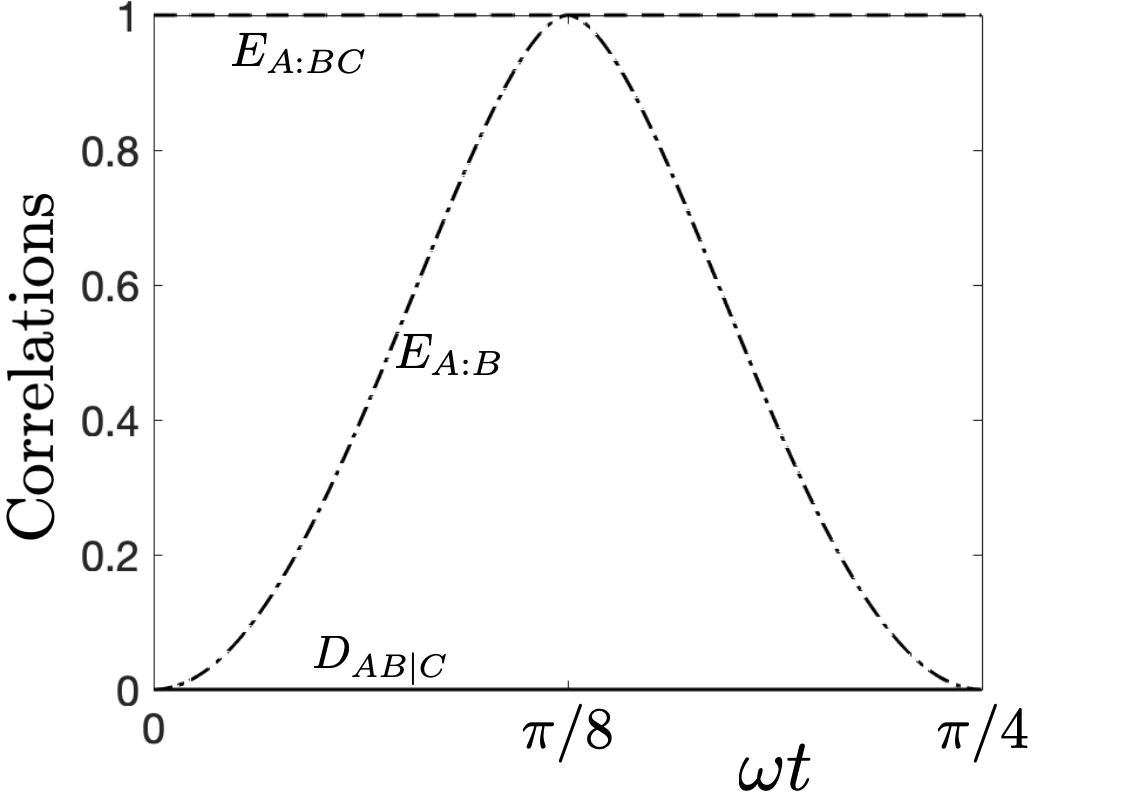} 
\caption{Entanglement distribution via classical ancillary system. 
The Hamiltonian is taken to be $H=(\sigma^x_A\otimes \openone \otimes \sigma^x_C+\openone \otimes \sigma^x_B \otimes \sigma^x_C) \: \hbar \omega$ with initial state $\tfrac{1}{2} \ket{\psi_+}\bra{\psi_+} \otimes \ket{+}\bra{+} + \tfrac{1}{2} \ket{\phi_+}\bra{\phi_+} \otimes \ket{-}\bra{-}$.
The dynamics shows that the entanglement between $A$ and $B$ is increasing while the quantum discord $D_{AB|C}$ is zero at all times.
}
\label{FIG_re_exp2} 
\end{figure}

For completeness I will provide the calculations for the dynamics in Fig.~\ref{FIG_re_exp2}.
Let us start with the state of $AB$ at time $T$
\begin{eqnarray}
\rho_{AB}(T)&=&\text{tr}_C\left(U_{BC}U_{AC}\rho_0U_{AC}^{\dagger}U_{BC}^{\dagger}\right)\nonumber \\
&=&\frac{1}{2}\rho_{\alpha}+\frac{1}{2}\rho_{\beta},
\end{eqnarray}
where we define $\rho_{\alpha}$ as the density matrix for a pure state $\ket{\alpha}=\cos(2T)\ket{\psi_+}-i\sin(2T)\ket{\phi_+}$ and similarly, $\rho_{\beta}$ for $\ket{\beta}=\cos(2T)\ket{\phi_+}+i\sin(2T)\ket{\psi_+}$.
As $\rho_{\alpha}$ and $\rho_{\beta}$ are both pure states, one can calculate the corresponding entanglement in the partition $A:B$ by von Neumann entropy of either $A$ or $B$. 
It is straightforward to obtain that $E_{A:B}(\rho_{\alpha})=E_{A:B}(\rho_{\beta})=1$, i.e., both states are maximally entangled at all times.
From the definition of REE: 
\begin{equation}
E_{A:B}(\rho_{\alpha})=\min_{\sigma_{\alpha}} \left( -\mbox{tr}\left(\rho_{\alpha}\log_2(\sigma_\alpha)\right)-S(\rho_{\alpha})\right).
\end{equation}
In this case, $E_{A:B}(\rho_{\alpha})=1$ and $S(\rho_{\alpha})=0$.
Therefore, we have $-\mbox{tr}\left(\rho_{\alpha}\log_2(\sigma_\alpha)\right) \ge 1$ (*), and similarly, $-\mbox{tr}\left(\rho_{\beta}\log_2(\sigma_\beta)\right) \ge 1$ (**).
Now we show that there exists a single separable state that achieves the minima in both (*) and (**).
By choosing 
\begin{equation}\label{EQ_css}
\sigma=\frac{1}{2}\ket{\psi_+}\bra{\psi_+}+\frac{1}{2}\ket{\phi_+}\bra{\phi_+},
\end{equation}
one directly verifies that $-\mbox{tr}\left(\rho_{\alpha}\log_2(\sigma)\right) =-\mbox{tr}\left(\rho_{\beta}\log_2(\sigma)\right)=1$.
Now, we calculate the entanglement as follows
\begin{eqnarray}
E_{A:B}(T)&=&\min_{\sigma} \left(-\text{tr}\left(\frac{1}{2}\rho_{\alpha}+\frac{1}{2}\rho_{\beta}||\sigma \right)\right)-S_{AB}(T)\nonumber \\
&=&\min_{\sigma}\left( \frac{1}{2}(-\text{tr}\left(\rho_{\alpha}||\sigma \right))+\frac{1}{2}(-\text{tr}\left(\rho_{\beta}||\sigma \right))\right)-S_{AB}(T)\nonumber \\
&=&1-S_{AB}(T),
\end{eqnarray}
where we have used (\ref{EQ_css}) as the minimising separable state.
See Appendix~\ref{A_eloc} for the dynamics of $E_{A:BC}$ and $D_{AB|C}$, and for a general prescription of entanglement localisation.
We note that our entanglement localisation proposal has been realised experimentally~\cite{elocexp}.

We note that similar considerations have been presented in Ref.~\cite{gyongyosi2014correlation} to provide a counter-example to the impossibility of entanglement gain via LOCC.
However, as mentioned above, the partition $A:BC$ is entangled already from the beginning (in our example we have $E_{A:BC}(0) = 1$).
Therefore, the entanglement localisation emphasises that the ancillary systems within the framework of LOCC, here $C$, are not allowed to be initially correlated with the principal objects $A$ and $B$, even if the correlations are classical.

\section{Non-classicality detection protocols}

Based on the acquired knowledge in previous sections, I will now present two methods for detecting quantum property of an inaccessible object (nonzero discord).
Note that, in this section, we will not assume the form of the initial state such that it is beneficial for general experimental purposes. 

\subsection{Via entanglement breaking channel}

We note that the only way of gaining entanglement in subsystem $AB$ via classical $C$ is to localise it from the already present entanglement in $A:BC$.
This is a consequence of Theorem \ref{TH_revealing} and reinforces its role as a proper generalisation of the monotonicity of entanglement to continuous interactions.
Namely,
\begin{equation}
E_{A:B}(\tau) \le E_{A:BC}(\tau) \le E_{A:BC}(0).
\label{E_CHAIN}
\end{equation}
Now, if we ensure by operating on the probes only that the initial entanglements coincide, i.e. $E_{A:BC}(0) = E_{A:B}(0)$, entanglement gain in system $AB$ is only possible due to nonzero discord $D_{AB|C}$.
As we are interested in observing entanglement gain, it is natural to start with as small entanglement as possible.
This leads us to propose the application of an entanglement-breaking channel to one of the principal systems, at time $t=0$.
Indeed, after application of the channel on $A$, we have $E_{A:B}(0)=E_{A:BC}(0)=0$.
In a more concrete example, the channel is a von Neumann measurement.
An arbitrary measurement is allowed and experimentalist should choose the one having potential for biggest entanglement gain.
Note that the measurement results need not be known.
Our method is symmetrical, i.e., the entanglement-breaking channel can be applied on $B$. 
This is a consequence of the symmetry in Theorem~\ref{TH_revealing}, leading to $E_{A:B}(\tau) \le E_{B:AC}(\tau) \le E_{B:AC}(0)$.
Together with the condition after the channel, $E_{A:B}(0)=E_{B:AC}(0)=0$, they give the same conclusion.
The non-classicality detection method is illustrated and summarised in Fig.~\ref{FIG_re_method}a.
Note that entanglement estimation in step (iii) can be realised with entanglement witnesses~\cite{RevModPhys.81.865,horodecki1996m}, rendering state tomography unnecessary.

\subsection{Via initial entropies}

Another method for revealing the non-classicality can be derived in terms of the purity (or entropy) of the probes.
For this purpose, we present the following Lemma.

\begin{lemma}
Provided the premise of Theorem \ref{TH_revealing}, the following inequality holds
\begin{equation}
E_{A:B}(\tau) \le S_A(0) + S_B(0),
\end{equation}
where the entanglement quantifier here is REE and $S_X$ denotes the von Neumann entropy of object $X$.
\end{lemma}
\begin{proof}
Let us begin here with Eq. (\ref{E_CHAIN}), that is $E_{A:B}(\tau) \le E_{A:BC}(\tau) \le E_{A:BC}(0)$. 
Now we will provide an upper bound to the initial entanglement $E_{A:BC}(0)$ in terms of initial von Neumann entropies of the accessible objects $A$ and $B$. 
First, note that REE is bounded by mutual information~\cite{modi2010unified}, which in our case reads $E_{A:BC} \le I_{A:BC}$.
We also have $I_{A:BC} \le S_A + S_B - S_{AB|C}$ from the sub-additivity of von Neumann entropy. 
As the state of object $C$ is classical, the conditional entropy follows $S_{AB|C} \ge 0$, and therefore
\begin{equation}\label{EQ_preie}
E_{A:B}(\tau) \le E_{A:BC}(0) \le S_A(0) + S_B(0).
\end{equation}
\end{proof}

Any violation of inequality (\ref{EQ_preie}) then reveals the non-classicality of $C$.
See Fig.~\ref{FIG_re_method}b for a summary of this detection scheme.

\begin{figure}[h]
\centering
\includegraphics[scale=0.45]{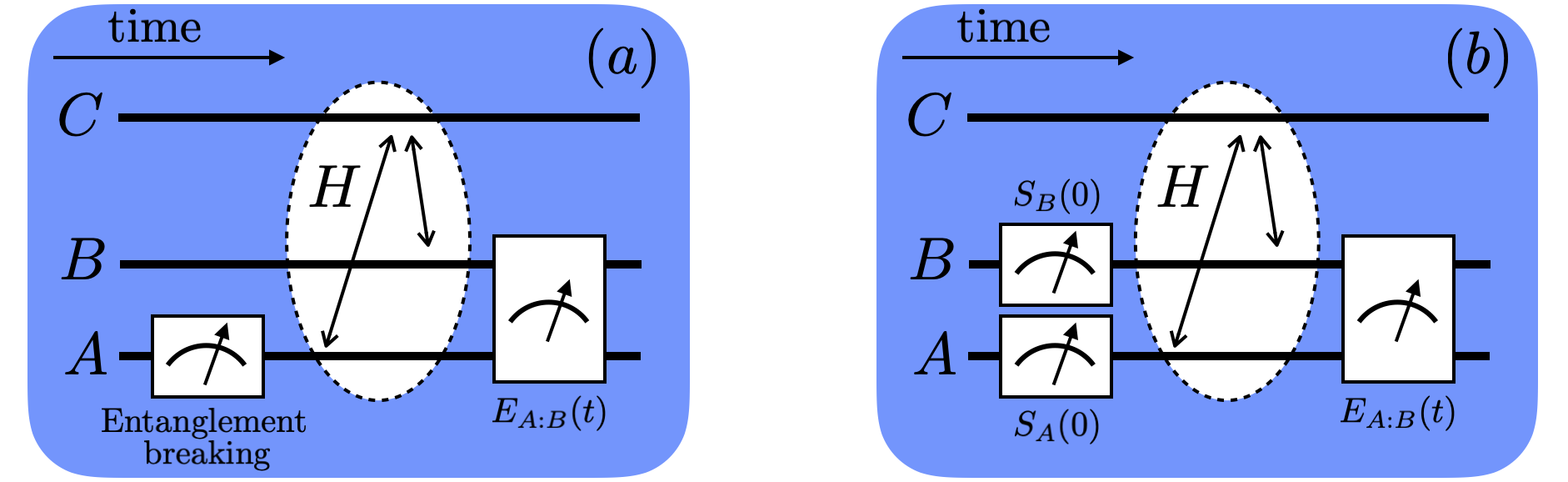} 
\caption{Summary of protocols for revelation of non-classicality of an inaccessible object $C$.
(a) Protocol via entanglement breaking channel involves the following steps:
(i) entanglement breaking channel on object $A$ (or object $B$), e.g., von Neumann measurement;
(ii) dynamics of the whole system;
(iii) estimation of entanglement between $A$ and $B$.
In this scheme, nonzero entanglement implies positive discord $D_{AB|C}$.
(b) Protocol via initial entropies has similar steps with the exception of measurement of initial entropies for objects $A$ and $B$. 
This scheme reveals positive $D_{AB|C}$ if the entanglement between $A$ and $B$ is larger than the sum of initial entropies, i.e., $E_{A:B}>S_A(0)+S_B(0)$. 
Note that for both protocols we have minimalistic assumptions. 
All objects can be open systems and no assumption is made regarding the initial state nor the explicit form of the Hamiltonian $H_{AC}+H_{BC}$.
}
\label{FIG_re_method} 
\end{figure}

\section{Minimalistic assumptions and applications}

We have proposed an entanglement-based criteria for the inference of quantumness of an inaccessible object.
Our protocols are fully non-disruptive of the state of the system to probe, and rely on only weak assumptions on the nature of the interactions involved.
They are also robust against decoherence.
These features make our proposal suitable to address non-classicality at many levels, 
from technological platforms such as quantum optomechanics (Chapter \ref{Chapter7}) to the problem of detecting quantumness in living organisms (Chapter \ref{Chapter_probing}) and the fundamental questions on the nature of gravity (Chapter \ref{Chapter_gravity}). 
For example, the gravity scheme puts $C$ as a gravitational field coupling massive objects $A$ and $B$, which are mutually non-interacting.
By determining experimentally the entanglement gain between $A$ and $B$ one would conclude, according to our scheme, the non-classical nature of the gravitational field between them.
That is, if we were to embed into the quantum formalism description of the masses and the field, the dynamics would be driving the field through non-orthogonal states as this is required for the quantum discord $D_{AB|C}$ to be non-zero.
A similar analysis can also be done for remote quantum dots in a solid-state substrate~\cite{pp} or spin-chain systems like in Ref.~\cite{NatPhys.11.255},
as their physics also naturally distinguishes a mediating object that is inaccessible, e.g., locations of unpaired spins are unknown in a sample.
We have performed detailed analysis of the latter in Ref.~\cite{kon2019non}.

\section{Summary}
In this study we considered two principal objects that are only indirectly interacting, via a mediator.
We have shown that quantum discord (a form of quantum property of the mediator) is a resource for entanglement gain.
This prompted us to propose schemes for assessing quantumness of an inaccessible object (the mediator).
The schemes are based on monitoring entanglement dynamics between the principal objects, which serve as probes over which we have control.
Our method is robust and experimentally friendly as it allows the controlled objects and the inaccessible one to be open systems, makes no assumptions about the initial state, dimensionality of involved Hilbert spaces, and details of the interaction Hamiltonian.


\chapter{Non-decomposability of evolution \& extreme gain of correlations} 

\label{Chapter_detecting} 

\lhead{Chapter 4. \emph{Non-decomposability of evolution \& extreme gain of correlations}} 

\emph{In this chapter, the amount of distributed correlations will be studied.\footnote{Parts of this chapter are reproduced from our published article of Ref.~\cite{nondecompos}, $\copyright$ [2019] American Physical Society. Where applicable, changes made will be indicated.}
I show that non-commutativity of interaction Hamiltonians (or, in a more general setting, non-decomposability of time evolution) is a necessary resource for having high gain of correlations between two objects.
From the extreme gain of correlations one can then detect the non-commutativity, which shows another level of quantum property that is present in the system.
This chapter starts with a motivation and a construction of a general setting that will be considered herein.
I will then proceed to show that for decomposable dynamics the correlations between two principal objects are bounded by a correlation capacity of a mediating system.
This applies for a plethora of correlation quantifiers, some of which will be presented in this chapter. 
Next, I will discuss a special scenario where all the objects involved are open to their local environments.
Finally, the origin of the possible extreme gain of correlations will be covered.
}

\clearpage 

\section{Motivation and objectives}

All classical observables are functions of positions and momenta.
Since there is no fundamental limit on the precision of position and momentum measurement in classical physics,
all classical observables are, in principle, measurable simultaneously.
Quite differently, the Heisenberg uncertainty principle forbids simultaneous exact knowledge of quantum observables corresponding to position and momentum.
The underlying non-classical feature is their non-commutativity:
Any pair of non-commuting observables cannot be simultaneously measured to arbitrary precision, as first demonstrated by Robertson in his famous uncertainty relation~\cite{robertson1929}.
Other examples of non-classical phenomena with underlying non-commutativity of observables include violations of Bell inequalities~\cite{landau1987,peres-book} or, more generally, non-contextual inequalities; e.g., see~\cite{thompson2016}.
Here we describe a method to detect non-commutativity of interaction Hamiltonians, and generally non-decomposability of temporal evolution, from the dynamics of correlations.

Consider the situation depicted in Fig.~\ref{FIG_de_setup}, where the probing systems $A$ and $B$ do not interact directly but only via the mediator $C$; i.e., there is no Hamiltonian term $H_{AB}$.
In general, we allow all objects to be open systems and study whether the evolution operator cannot be represented by a sequence of operations between each probe and the mediator, i.e., $\Lambda_{ABC}=\Lambda_{BC}\Lambda_{AC}$ or in reverse order.
For the special case in which all the systems are closed, non-decomposability implies non-commutativity of interaction Hamiltonians, i.e., $[H_{AC},H_{BC}]\ne0$.
Indeed, for commuting Hamiltonians, the unitary evolution operator is decomposable into $U_{BC}U_{AC}$, where, for example, $U_{AC} = \exp(- i t H_{AC}/\hbar)$.
We show that for decomposable evolution, correlations between $A$ and $B$ are bounded. 
We also show with concrete dynamics generated by non-commuting Hamiltonians that these bounds can be violated.
The bounds derived depend solely on the dimensionality of $C$ and not on the actual form of the evolution operators.
Hence, these operators can remain unknown throughout the assessment.
This is a desired feature, as experimenters usually do not reconstruct the evolution operators via process tomography.
It also allows applications of the method to situations where the physics is not understood to the extent that reasonable Hamiltonians or Lindblad operators can be written down.
Furthermore, the assessment does not depend on the initial state of the tripartite system and does not require any operations on the mediator.
It is therefore applicable to a variety of experimental situations; Refs.~\cite{rauschenbeutel2001,NatPhys.11.255,pp,hamsen} provide concrete examples.

\begin{figure}[h]
\centering
\includegraphics[scale=0.45]{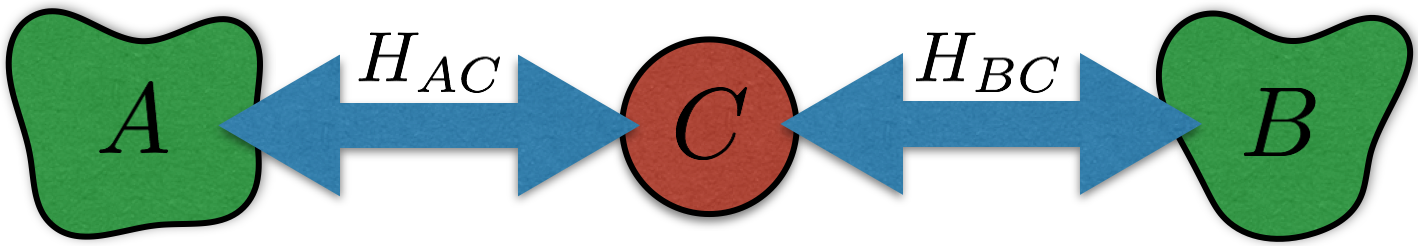}
\caption{General setting for detecting non-decomposability of time evolution.
High dimensional objects $A$ and $B$ are interacting with a low dimensional object $C$, but not with each other. 
In general, we allow all objects to be open to their local environments. 
The interaction Hamiltonians are given by $H_{AC}$ and $H_{BC}$, which describe the coherent part of the dynamics.
In the text, we show that large gain of correlation between $A$ and $B$, exceeding certain thresholds (functions of the dimensionality of $C$), implies non-decomposability of time evolution, which makes the latter crucial for substantial correlation distribution.
Note that a lot of correlation quantifiers are useful in our method. 
This setup is different from the one considered in Fig.~\ref{FIG_re_ABC} of Chapter~\ref{Chapter_revealing}.
For non-decomposability detection, here the dimension of $C$ is assumed to be known and, as we will show in the text, small value of such dimension is preferable.
}
\label{FIG_de_setup}
\end{figure}

Below I shall begin by presenting the general bounds on the amount of correlations one can establish if the evolution is decomposable.
It is shown that these bounds are generic and hold for a large number of correlation quantifiers.
We then calculate concrete bounds on exemplary quantifiers.
In Chapter~\ref{Chapter7}, I will show how they can be violated in a system of two fields coupled by a two-level atom.
We discuss the origin of the violation in terms of ``Trotterized" evolution, where a virtual particle is exchanged between $A$ and $B$ multiple times if the Hamiltonians do not commute but only once if they do commute.
Finally, we focus on immediate applications in quantum information and discuss the consequences of our findings for correlation distribution protocols and dimension witnesses.

\section{Correlation capacity bound}

Consider the setup illustrated in Fig. \ref{FIG_de_setup}.
System $C$, with finite dimension $d_C$, is mediating interactions between higher-dimensional systems $A$ and $B$.
For simplicity we take $d_A=d_B > d_C$.
We assume that there is no direct interaction between $A$ and $B$, such that the Hamiltonian of the whole tripartite system is of the form $H_{AC}+H_{BC}$ (local Hamiltonians $H_A$, $H_B$, and $H_C$ included).
Our bounds follow from a generalization of the following simple observation.
Consider, for the moment, the relative entropy of entanglement as the correlation quantifier \cite{vedral1997quantifying}.
If the evolution is decomposable, it can be written as $\Lambda_{BC} \Lambda_{AC}$, or in reverse order.
Therefore, it is as if particle $C$ interacted first with $A$ and then with $B$, a scenario similar to that in Refs.~\cite{cubitt,bounds1,bounds2,edssexp1,edssexp2,edssexp3}.
The first interaction can generate at most $\log_2(d_C)$ ebits of entanglement, whereas the second, in the best case, can swap all this entanglement. 
In the end, particles $A$ and $B$ gain at most $\log_2(d_C)$ ebits.
The bound is indeed independent of the form of interactions. 
Furthermore, it is intuitively clear, as this is just the ``quantum capacity'' of the mediator. 
Below we will generalise this bound for varied initial states and more correlation quantifiers.

\subsection{Arbitrary initial states}

Let us consider correlation quantifiers defined in the so-called ``distance" approach~\cite{vedral1997quantifying, modi2010unified}.
The idea is to quantify correlation $Q_{X:Y}$ in a state $\rho_{XY}$ as the shortest distance $D(\rho_{XY},\sigma_{XY})$ from $\rho_{XY}$ to a set of states $\sigma_{XY} \in \mathcal{S}$ without the desired correlation property,
i.e., $Q_{X:Y} \equiv \inf_{\sigma_{XY}\in \mathcal{S}} D(\rho_{XY},\sigma_{XY})$.
For example, the relative entropy of entanglement is given by the relative entropy of a state to the set of disentangled states~\cite{vedral1997quantifying}.
It turns out that most such quantifiers are useful for the task introduced here.
The conditions we require are that (i) $\mathcal{S}$ is closed under local operations $\Lambda_Y$ on $Y$,
(ii) $D(\Lambda[\rho],\Lambda[\sigma]) \leq D(\rho, \sigma)$ (monotonicity), and (iii) $D(\rho_0, \rho_1) \leq D(\rho_0, \rho_2) + D(\rho_2, \rho_1)$ (triangle inequality).
They are sufficient to prove Theorem~\ref{TH_cons} below.

For completeness let us begin with a useful lemma.

\begin{lemma}\label{LM_pre}
Consider correlation between two parties $X$ and $Y$, as measured by a correlation quantifier $Q_{X:Y}$ that is non-increasing under local operations on the side of $Y$.
It follows that $Q_{X:Y}$ stays constant under all reversible operations on $Y$, e.g., tracing-out uncorrelated objects in $Y$. 
\end{lemma}
\begin{proof}
Let us suppose that one performs an operation on party $Y$. 
Based on the monotonicity property, $Q_{X:Y}$ cannot increase in value. 
Furthermore, the value also cannot decrease for a reversible operation, because this means $Q_{X:Y}$ can increase by reversing the process. 
Therefore, $Q_{X:Y}$ can only be invariant. 
As tracing-out uncorrelated objects is a reversible process, the conclusion follows. 
\end{proof}

The main theorem in this chapter is proven as follows.

\begin{theorem}\label{TH_cons}
Consider a distance-based quantifier of correlation between party $X$ and party $Y$, defined as $Q_{X:Y} \equiv \inf_{\sigma_{XY}\in \mathcal{S}} D(\rho_{XY},\sigma_{XY})$ with properties
\begin{enumerate}
\item[{\rm (i)}] $\mathcal{S}$ is closed under local maps $\Lambda_Y$ on $Y$;
\item[{\rm (ii)}] $D(\Lambda[\rho],\Lambda[\sigma]) \leq D(\rho, \sigma)$; and
\item[{\rm (iii)}] $D(\rho_0, \rho_1) \leq D(\rho_0, \rho_2) + D(\rho_2, \rho_1)$.
\end{enumerate}
It follows that for a decomposable dynamical operator, i.e., $\Lambda_{ABC}=\Lambda_{BC}\Lambda_{AC}$, one has 
\begin{equation}\label{EQ_thebound}
	Q_{A:B} (t) \leq \mathcal{I}_{AC:B}(0) + \sup_{\ket{\psi}} Q_{A:C},
\end{equation}
where we define here $\mathcal{I}_{AC:B}(0) \equiv \inf_{\sigma_{AC} \otimes \sigma_B} D(\rho, \sigma_{AC} \otimes \sigma_B)$, use $\rho$ as the initial density matrix, and the supremum in the last term is over pure states $\{\ket{\psi}\}$ of $AC$.
\end{theorem}
\begin{proof}
First, we show that $Q_{X:Y}$ is monotone under local operations $\Lambda_Y$.
This is done as follows
\begin{eqnarray}\label{EQ_pre1}
Q_{X:Y}(\Lambda_Y[\rho_{XY}])&=&\inf_{\sigma_{XY}\in \mathcal{S}} D(\Lambda_Y[\rho_{XY}],\sigma_{XY}) \nonumber \\
&\le&D(\Lambda_Y[\rho_{XY}],\Lambda_Y[\sigma_{XY}^{0}])\nonumber \\
&\le&D(\rho_{XY},\sigma_{XY}^{0})\nonumber \\
&=&Q_{X:Y}(\rho_{XY}),
\end{eqnarray}
where $\sigma_{XY}^{0}$ is a state in $\mathcal{S}$ closest to $\rho_{XY}$.
The second line in (\ref{EQ_pre1}) is due to property (i) and that $\Lambda_Y[\sigma_{XY}^{0}]$ might not be the the closest to $\Lambda_Y[\rho_{XY}]$, while the third line follows from property (ii). 
This allows us to use the property in Lemma~\ref{LM_pre}.

Next, we have the following arguments
\begin{eqnarray}
	Q_{A:B} (t)
	&\leq& Q_{A:BC} \left(\Lambda_{BC} \Lambda_{AC} [\rho] \right) \label{APP_TH_S1} \\
	&\le& Q_{A:BC} \left( \Lambda_{AC} [\rho] \right) \label{APP_TH_S2} \\
	& \le & D \left(\Lambda_{AC} [\rho], \mu \right) \label{APP_TH_S3} \\
	& \leq & D \left( \Lambda_{AC} [\rho], \Lambda_{AC} [\sigma^0_{AC}]  \otimes \sigma^0_B \right) \nonumber \\
	& + & D\left( \Lambda_{AC} [\sigma^0_{AC}]  \otimes \sigma^0_{B}, \mu \right) \label{APP_TH_TRIAN}\\
	& \le &  D(\rho, \sigma^0_{AC} \otimes \sigma^0_B)	\nonumber \\
	&+ & D\left( \Lambda_{AC} [\sigma^0_{AC}]  \otimes \sigma^0_{B}, \mu \right) \label{APP_TH_S4} \\
	&=& \mathcal{I}_{AC:B}(0) + Q_{A:BC} (\Lambda_{AC} [\sigma^0_{AC}]  \otimes \sigma^0_{B} ) \label{APP_TH_S5} \\
	&=& \mathcal{I}_{AC:B}(0) + Q_{A:C} (\Lambda_{AC} [\sigma^0_{AC}] ) \label{APP_TH_S6} \\
	&\leq& \mathcal{I}_{AC:B}(0) +\sup_{\ket{\psi}} Q_{A:C}, \label{APP_TH_S7}
	\end{eqnarray}
where the steps are explained as follows.
We utilise that $Q_{X:Y}$ is monotone under local maps on $Y$ (in this case, tracing out $C$) in line (\ref{APP_TH_S1}).
Similarly, in (\ref{APP_TH_S2}), $Q_{A:BC}$ is monotone under the operation $\Lambda_{BC}$.
The inequality (\ref{APP_TH_S3}) utilises the definition of $Q_{A:BC}$ as a distance-based quantifier and holds for any state $\mu \in \mathcal{S}_{A:BC}$.
The triangle inequality (iii) confirms line (\ref{APP_TH_TRIAN}).
We note that the first term in (\ref{APP_TH_TRIAN}) is independent of $\mu$ and one can choose any state $\sigma^0_{AC}$ and $\sigma^0_B$ at this point.
Property (ii) is utilised in equality (\ref{APP_TH_S4}).
For line (\ref{APP_TH_S5}), the state $\sigma^0_{AC} \otimes \sigma^0_B$ is chosen as the closest product state to $\rho$, while $\mu \in \mathcal{S}_{A:BC}$ as the closest state to $\Lambda_{AC} [\sigma^0_{AC}]  \otimes \sigma^0_{B}$.
The equality (\ref{APP_TH_S6}) is obtained as $Q_{A:BC}$ is invariant under the partial trace of the decoupled state $\sigma^0_B$.
The final line follows from a property that a bipartite correlation quantifier is maximal on pure states if it is non-increasing under local operations on at least one party~\cite{streltsov2012general}.
\end{proof}

Note that although the relative entropy does not satisfy (iii) it still follows Theorem \ref{TH_cons}. This is the result of the following lemma. 

\begin{lemma}\label{LM_re}
For a distance measure based on the relative entropy the conclusion in Theorem \ref{TH_cons} follows.
\end{lemma}
\begin{proof}

First, we show the following identity
\begin{eqnarray}
S(\rho||\sigma_X\otimes \sigma_Y)&=&\mbox{tr}(\rho \log{\rho}-\rho\log{\sigma_X \otimes \sigma_Y})  \nonumber \\
&=& \mbox{tr}(\rho \log{\rho}-\rho \log{\rho_X \otimes \rho_Y})  \nonumber \\
&&+\mbox{tr}(\rho \log{\rho_X \otimes \rho_Y} - \rho \log{\sigma_X \otimes \sigma_Y}) \nonumber \\
&=&S(\rho||\rho_X\otimes \rho_Y)+S(\rho_X||\sigma_X) \nonumber \\
&&+S(\rho_Y||\sigma_Y), \label{GG0}
\end{eqnarray}
where $\rho_X$ and $\rho_Y$ are the reduced states of $\rho$.
We have also used, e.g., the relation $\mbox{tr}(\rho \log{\sigma_X \otimes \sigma_Y}) = \mbox{tr}(\rho_X \log{\sigma_X}) + \mbox{tr}(\rho_Y \log{\sigma_Y})$.

While relative entropy is known to follow property (ii) \cite{uhlmann1977relative}, it does not satisfy property (iii).
In this case, one starts with line (\ref{APP_TH_S2}) in Theorem \ref{TH_cons} and obtains
\begin{eqnarray}
Q_{A:BC} \left( \Lambda_{AC} [\rho] \right)&=&\inf_{\mu \in \mathcal{S}_{A:BC}} S \left(\Lambda_{AC}[\rho] || \mu \right) \label{GG01}\\
&\le& S(\Lambda_{AC}[\rho] || \mu_{AC}\otimes \mu_B) \label{GG1}\\
&=&  S(\Lambda_{AC}[\rho] || \rho_{AC}^{\prime}\otimes \rho_B^{\prime}) \nonumber \\
&&+S(\rho_{AC}^{\prime} || \mu_{AC})+S(\rho_{B}^{\prime} || \mu_{B}) \label{GG2} \\
&=&I_{AC:B}(\Lambda_{AC}[\rho] )+Q_{A:C}(\rho_{AC}^{\prime}) \label{GG3} \\
&\le&I_{AC:B}(0)+\sup_{\ket{\psi}} Q_{A:C},
\end{eqnarray}
where $\rho_{AC}^{\prime}$ and $\rho_B^{\prime}$ are the reduced states of $\Lambda_{AC}[\rho]$. 
We justify the steps above as follows.
In line (\ref{GG1}), we use a state of the form $\mu_{AC}\otimes \mu_B$ that is in the set of $\mathcal{S}_{A:BC}$.
The identity of (\ref{GG0}) confirms line (\ref{GG2}).
From the definition of mutual information as the distance from a state to its marginals \cite{modi2010unified}, one arrives at equality (\ref{GG3}).
Note that $\mu_{AC}\in \mathcal{S}_{A:C}$ has been chosen as the closest to $\rho_{AC}^{\prime}$ and $\mu_B=\rho_B^{\prime}$.
The final line is obtained from the monotonicity of mutual information under the local map $\Lambda_{AC}$ and the supremum of $Q_{A:C}$ is attained over pure states. 
\end{proof}

Correlations between probe $A$ and probe $B$ are therefore bounded by the maximal achievable correlation with the mediator, $\sup_{\ket{\psi}} Q_{A:C}$.
The additional term $\mathcal{I}_{AC:B}(0)$ in Eq.~(\ref{EQ_thebound}) reduces to the usual mutual information $I_{AC:B}(0)$ if $D(\rho_{XY},\sigma_{XY})$ is the relative entropy distance~\cite{modi2010unified} 
and characterizes the amount of total initial correlations between one of the probes and the rest of the system.
Note that the bound is independent of time.
This can be seen as a result of the effective description of such dynamics given by $\Lambda_{BC}\Lambda_{AC}$.
The particle $C$ is exchanged between $A$ and $B$ only once, independently of the duration of the dynamics. 

In a typical experimental situation the initial state can be prepared as completely uncorrelated $\rho = \rho_A \otimes \rho_B \otimes \rho_C$,
in which case Theorem~\ref{TH_cons} simplifies and the bound is given solely in terms of the ``correlation capacity'' of the mediator:
\begin{equation}
Q_{A:B} (t) \le \sup_{\ket{\psi}} Q_{A:C}.
\label{EQ_CC}
\end{equation}
Clearly, the same bound holds for initial states of the form $\rho = \rho_{AC} \otimes \rho_B$.
In Section~\ref{SS_ucs} we show that, with this initial state, Eq.~(\ref{EQ_CC}) holds for any correlation quantifier that is monotonic under local operations $\Lambda_{BC}$, 
not necessarily based on the distance approach, e.g., any entanglement monotone.

\subsection{Uncorrelated initial states}\label{SS_ucs}

Here we prove a bound on correlations for particular initial states where one of the probing objects is decoupled.

\begin{theorem}\label{TH_main}
Consider the initial density matrix of the system being in the form $\rho = \rho_{AC} \otimes \rho_B$.
For a decomposable dynamical operator $\Lambda_{ABC}=\Lambda_{BC}\Lambda_{AC}$, it follows that
\begin{equation}
Q_{A:B}(t)\le \sup_{\ket{\psi}} \:Q_{A:C},
\end{equation}
where $Q$ is any correlation quantifier that is monotone under local maps $\Lambda_{BC}$.
\end{theorem}
\begin{proof}
The proof is done with the following arguments 
\begin{eqnarray}
Q_{A:B}(t) & \le & Q_{A:BC}(t)\le Q_{A:BC}(\Lambda_{AC}[\rho]) \nonumber \\
&=& Q_{A:C}(\Lambda_{AC}[\rho]) \le \sup_{\ket{\psi}}  Q_{A:C},
\end{eqnarray}
where the steps are justified as follows.
The first inequality is obtained as tracing out party $C$ (which might be correlated with $AB$ in general) is a local operation on the side of $BC$.
The monotonicity of the quantifier under $\Lambda_{BC}$ confirms the second argument.
Party $B$ stays uncorrelated after the application of $\Lambda_{AC}$ on the initial density matrix $\rho_{AC}\otimes \rho_B$.
As a result, the equality follows by utilising Lemma \ref{LM_pre}.
The quantifier $Q_{A:C}$ is maximised on pure states, giving the final inequality.
\end{proof}

We note that for initial states that are close to $\rho = \rho_{AC} \otimes \rho_B$ one can utilize the continuity of the von Neumann entropy \cite{fannes1973continuity} and see that $\mathcal{I}_{AC:B}(0)$ in Eq. (\ref{EQ_thebound}) is indeed small. We can also ensure that the initial state is of the form $\rho = \rho_{AC} \otimes \rho_B$  by performing a correlation breaking channel on $B$ first.
One example of such a channel is a measurement in the computational basis followed by a measurement in some complementary (say Fourier) basis.
This implements the correlation breaking channel $(\identity_{AC} \otimes \Lambda_B) (\rho_{ABC}) = \rho_{AC} \otimes \frac{\identity}{d_B}$.\footnote{In this context one might ask whether there exists a channel such that $(\identity_{AC} \otimes \Lambda_B) (\rho_{ABC}) = \rho_{AC} \otimes \rho_B$, where $\rho_{AC}$ and $\rho_B$ are reduced density matrices. Such a channel does not exist, e.g., it would be nonlinear. See also Ref.~\cite{milegu}.}
In this way, our method does not require any knowledge of the initial state and any operations on the mediator, similar in spirit to the detection of quantum discord of inaccessible objects in Ref.~\cite{krisnanda2017}.
We now move to concrete correlation quantifiers and their correlation capacities.

\section{Correlation capacity of particular quantifiers}

We provide four correlation quantifiers (whose definitions will be repeated here for convenience), which capture different types of correlations between quantum particles.
They are mutual information, classical correlation, quantum discord, and negativity.
All of them are shown to be useful in detecting non-decomposability.

Mutual information is a measure of total correlations~\cite{groisman2005} and is defined as $I_{X:Y} = S_X + S_Y - S_{XY}$, where, e.g., $S_X$ is the von Neumann entropy of subsystem $X$.
It can also be seen as a distance-based measure with the relative entropy as the distance and a set of product states $\sigma_X \otimes \sigma_Y$ as $\mathcal{S}$~\cite{modi2010unified}.
The supremum in Eq. (\ref{EQ_CC}) is attained by the state (recall that $d_A > d_C$),
\begin{equation}
|\Psi \rangle = \frac{1}{\sqrt{d_C}} \sum_{j = 1}^{d_C} |a_j \rangle |c_j \rangle,
\label{EQ_MAX_ENT}
\end{equation}
where $| a_j \rangle$ and $| c_j \rangle$ form orthonormal bases. 
One finds $\sup_{\ket{\psi}} I_{A:C} = 2 \log_2(d_C)$.

An interesting quantifier in the context of non-classicality detection is the classical correlation in a quantum state.
It is defined as mutual information of the classical state obtained by performing the best local von Neumann measurements on the original state $\rho$~\cite{terhal2002},
i.e., $C_{X:Y} = \sup_{\Pi_X \otimes \Pi_Y} I_{X:Y}(\Pi_X \otimes \Pi_Y(\rho) )$, where $\Pi_X \otimes \Pi_Y(\rho) = \sum_{xy} \proj{xy} \rho \proj{xy}$, and $\ket{x}$, $\ket{y}$ form orthonormal bases.
The supremum of mutual information over classical states of $AC$ is $\log_2(d_C)$.

Quantum discord is a form of purely quantum correlations that contain quantum entanglement.
It can be phrased as a distance-based measure. 
In particular, we consider the relative entropy of discord~\cite{modi2010unified}, also known as the one-way deficit~\cite{deficit}.
It is an asymmetric quantity defined as $D_{X|Y} = \inf_{\Pi_Y} S(\Pi_Y(\rho)) - S(\rho)$, where $\Pi_Y$ is a von Neumann measurement conducted on subsystem $Y$.
The relative entropy of discord is maximized by the state (\ref{EQ_MAX_ENT}), for which we have $\sup_{\ket{\psi}} D_{A|C} = \log_2(d_C)$.

Our last example is negativity, a computable entanglement monotone~\cite{zyczkowski1998volume,lee2000entanglement,lee2000partial,negativity}.
For a bipartite system negativity is defined as $N_{X:Y} = (||\rho^{T_X}||_1-1)/2$,
where $||.||_1$ denotes the trace norm and $\rho^{T_X}$ is a matrix obtained by partial transposition of $\rho$ with respect to $X$.
Negativity is maximized by the state (\ref{EQ_MAX_ENT}), and the supremum reads $\sup_{\ket{\psi}} N_{A:C} = (d_C-1)/2$.

Clearly, many other correlation quantifiers are suitable for our detection method because the assumptions behind Eqs.~(\ref{EQ_thebound}) and (\ref{EQ_CC}) are not demanding.
For those presented here, see Table~\ref{TB_sum} for a summary.
In fact, one may wonder which correlations do not qualify for our method.
A concrete example is the geometric quantum discord based on $p$-Schatten norms with $p > 1$, as it may increase under local operations on $BC$~\cite{piani2012,paula2013}.

\begin{table}[h]
\begin{center}
\caption{Summary of particular correlation quantifiers and the corresponding capacities. 
Note that we have assumed $d_C<d_A,d_B$.
}\label{TB_sum}
\smallskip
\begin{tabular}{|l|c|}
\hline
Quantifier & Correlation capacity\\
\hline
\hline
Mutual information & $2\log_2(d_C)$\\
\hline
Classical correlation & $\log_2(d_C)$\\
\hline
Quantum discord & $\log_2(d_C)$\\
\hline
Negativity & $(d_C-1)/2$\\
\hline
\end{tabular}
\end{center}
\end{table}

\section{Local environment \& non-commutativity of Hamiltonians}

Finally, we wish to discuss a scenario where all the objects are open to their own local environments, as realised, e.g., in~\cite{hamsen}.
We take the evolution following the master equation in Lindblad form,
\begin{eqnarray}
\dot \rho & = & -\frac{i}{\hbar}[H_{AC}+H_{BC},\rho]+\sum_{X=A,B,C}L_X\rho, \label{EQ_de_open}\\
L_X\rho & = & \sum_k Q_X^k\rho Q_{X}^{k\dag}-\frac{1}{2}\{Q_{X}^{k\dag}Q_X^k,\rho\}, \nonumber
\end{eqnarray}
where the last term in (\ref{EQ_de_open}) is the incoherent part of the evolution and $L_X$ describes interactions of system $X$ with its local environment, i.e., the operators $Q_X^k$ act on system $X$ only.
We denote $\mathcal{L}_{AC}=-(i/\hbar)[H_{AC},\cdot]+L_A+L_C$ and $\mathcal{L}_{BC}=-(i/\hbar)[H_{BC},\cdot]+L_B$. 
One readily verifies that if $[H_{AC},H_{BC}]=0$ and $[L_C,H_{BC}/\hbar]=0$, we have commuting Lindblad operators $[\mathcal{L}_{AC},\mathcal{L}_{BC}]=0$. 
Note that, if one includes $L_C$ in $\mathcal{L}_{BC}$ instead, the second condition for commuting Lindblad operators now reads $[L_C,H_{AC}/\hbar]=0$.
For commuting Lindbladians, the corresponding evolution decomposes as $\Lambda_{BC} \Lambda_{AC}$, or in reverse order. 
Therefore, our bounds apply accordingly. 
Their violation implies that either the Hamiltonians do not commute or the operators describing dissipative channels on $C$ do not commute with $H_{AC}$ and $H_{BC}$.
In particular, if $C$ is kept isolated so that its noise can be ignored, the violation of our bounds is solely the result of the non-commutativity of the interaction Hamiltonians.

\section{Intuition behind excessive gain of correlations}

In Chapter~\ref{Chapter7}, I present high correlation gain (violating the capacity bounds) between two orthogonally-polarised cavity fields that are individually interacting with a two-level atom.
Here, the origin of the violation will be explained.

Since the total Hamiltonian is of the form $H_{AC} + H_{BC}$, the Suzuki-Trotter expansion of the resulting evolution is particularly illuminating,
\begin{equation}\label{EQ_trotterexp}
e^{i \frac{t}{\hbar} (H_{AC} + H_{BC})} = \lim_{n \to \infty} \left( e^{i\frac{\Delta t}{\hbar} H_{BC}} e^{i\frac{\Delta t}{\hbar} H_{AC}} \right)^n,
\end{equation}
where $\Delta t = t / n$.
If Hamiltonians do not commute, it is necessary to think about Eq. (\ref{EQ_trotterexp}) as $n$ sequences of pairwise interactions of $C$ with $A$ followed by $C$ with $B$, each for a time $\Delta t$.
Each pair of interactions can only increase correlations up to the correlation capacity of the mediator, but their multiple use allows the accumulation of correlations beyond what is possible with commuting Hamiltonians.
Recall that, in the latter case, we deal with only one exchange of system $C$, independently of the duration of dynamics.
We stress that Trotterization is just a mathematical tool and in the laboratory system $C$ is continuously coupled to $A$ and $B$.
It is rather as if a virtual particle $C$ were transmitted multiple times between $A$ and $B$, interacting with each of them for a time $\Delta t$.


\section{Summary}
One of the most elementary non-classical traits of quantum observables is their non-commutativity.
In this chapter, we have linked non-commutativity of interaction Hamiltonians (more generally, non-decomposability of time evolution) to the amount of correlations that can be created in the associated dynamics.
This led us to a method for detection of non-decomposability of evolution in a scenario where a mediator $C$ mediates interactions between two probing objects $A$ and $B$ (all these objects can interact with their local environments).
The Hamiltonians or Lindblad operators can remain unknown throughout the assessment, we only require knowledge of the dimension of the mediator.
Furthermore, no operation on $C$ is necessary at any time, which makes this strategy experimentally friendly.
Technically, under the assumption of decomposable evolution, we derived upper bounds on correlations between the probes.
Non-decomposability is detected by observing violation of certain bounds on $AB$ correlations.
A plethora of correlation quantifiers are helpful in our method, e.g., quantum entanglement, discord, mutual information, and even classical correlation.
We also provided an intuitive explanation in terms of multiple exchanges of a virtual particle which lead to the possible excessive accumulation of correlations.

\afterpage{\blankpage} 

\chapter{Speed of distribution of quantum entanglement} 

\label{Chapter_speedup} 

\lhead{Chapter 5. \emph{Speed of distribution of quantum entanglement}} 

\emph{This chapter investigates the rate at which quantum entanglement is generated.
I will begin with a motivation comprising the well-known quantum speed limit and general scenarios that will be considered here.
Entangling speed limit for the direct interaction setting will be examined. 
The following section then formulates the limit for its counterpart, the indirect interaction setting, with exemplary dynamics saturating the speed.
Next, I will revisit the entanglement localisation introduced in Chapter~\ref{Chapter_revealing} and provide an upper bound for entanglement gain in terms of initial correlations present in the initial state.
I will also present the implications of the results and a simple application to charging quantum batteries.
}

\clearpage
\section{Motivation and objectives}

A quantum state requires finite time to evolve into a state distinguishable from its initial form. 
This limitation on the shortest time is widely known as the quantum speed limit (QSL).
The lower bound on time required for the evolution was first obtained by Mandelstam and Tamm \cite{mandelstam}.
Since then, advances for the time bound include unitary evolution of pure states \cite{pure1,pure2,margolus} and for mixed states \cite{mix1,mix2,mix3,mix4}.
These results have many applications, e.g., to the investigation of the rate of entropy increase~\cite{deffner2010generalized}, the limit in quantum metrology~\cite{giovannetti2011advances}, the limit on computation \cite{lloyd2000ultimate,lloyd2002computational}, and the bound on charging power of quantum batteries \cite{binder2015quantacell,binderphdthesis,campaioli2017enhancing}.
Furthermore, previous studies have also shown the application of QSL to classical domain~\cite{shanahan2018quantum,okuyama2018quantum}.

We will use the time bound that is given by the \emph{unified} quantum speed limit \cite{unifiedbound,revbound}, which, for an evolution of a state $\rho$ to another state $\sigma$, reads
\begin{equation}\label{EQ_bound}
\tau(\rho,\sigma)=\hbar \frac{\Theta(\rho,\sigma)}{\min \{ \langle H\rangle,\Delta H\}},
\end{equation}
where $\Theta(\rho,\sigma)=\arccos( \mathcal{F}(\rho,\sigma))$ is the Bures angle, $\mathcal{F}(\rho,\sigma)=\mbox{tr}\left( \sqrt{\sqrt{\rho}\sigma\sqrt{\rho}} \right)$ is the Uhlmann root fidelity \cite{fidelity1,fidelity2}, $\langle H\rangle=\mbox{tr}(H\rho)-E_\text{g}$ is the average energy with respect to the ground state energy of $H$, and $\Delta H=\sqrt{\mbox{tr}[H^2\rho]-\mbox{tr}[H\rho]^2}$ is the standard deviation of energy (SDE).
Note that for pure states, the Bures angle is given by the Fubini-Study distance, i.e., $\Theta(\ket{\psi}\bra{\psi},\ket{\phi}\bra{\phi})=\arccos|\langle \psi |\phi \rangle|$ \cite{FS1,FS2,FS3}.

In this chapter, we study the quantum speed of system $AB$ with direct and indirect (through a mediator) continuous interactions, see Fig.~\ref{FIG_ch2_direct} and Fig.~\ref{FIG_ch2_indirect}b. 
The direct interaction case considers the objects $A$ and $B$ with Hamiltonian $H_{AB}$ (local Hamiltonians $H_A$ and $H_B$ included).
On the other hand, the interactions between $A$ and $B$ are mediated by an object $C$, such that the total Hamiltonian is of the form $H_{AC}+H_{BC}$ (again, local terms $H_A$, $H_B$, and $H_C$ included).
During the evolution, we study properties of the principle objects $A$ and $B$, such as quantum entanglement between them and charging power.
The latter has been defined as the speed to evolve the quantum state of $AB$ from $\ket{00}$ to $\ket{11}$ (fully charged). 
We compare the entanglement creation and the charging power between the direct and indirect interaction scenarios in terms of speed.
In particular, we investigate whether one setup has an advantage over the other and the factors that are relevant for optimal speed.

\section{Preliminaries}
It can be seen from the QSL of Eq.~(\ref{EQ_bound}) that the first relevant quantity is the fidelity between the initial state $\rho_{AB}(0)$ and the target state $\rho_{AB}(t)$. 
The second quantity is $\min \{ \langle H\rangle,\Delta H\}$, which involves the average energy or SDE. 
It is clear that the change in Hamiltonian $H\rightarrow kH$ (e.g., having more energy) results in $\langle H\rangle \rightarrow k\langle H\rangle$ and $\Delta H\rightarrow k \Delta H$ for a constant $k$. 
This means one can always speed up the evolution of the quantum state by, for example, providing a higher amount of energy.
Therefore, to put all processes on equal footing, we fix $\min \{ \langle H\rangle,\Delta H\}$ to be $\hbar \Omega$.
For ease of computation, let us rewrite Eq. (\ref{EQ_bound}) as
\begin{equation}\label{EQ_dbound}
\Gamma(\rho,\sigma)=\frac{\Theta(\rho,\sigma)}{\min \{ \langle M\rangle,\Delta M\}},
\end{equation}
where $\Gamma=\tau \Omega$ is a dimensionless time bound, $\langle M\rangle$ and $\Delta M$ are dimensionless quantities representing the average energy and SDE, respectively, per unit $\hbar \Omega$.
In this way, we fix the frequency $\Omega$ and set the resource on equal footing by having the condition $\min \{\langle M\rangle,\Delta M\}=1$.\footnote{Throughout this chapter, this will be referred to as \emph{resource equality}.}

We will use some correlation quantifiers (will be recalled briefly below) throughout this chapter.
For extensive definitions, see Chapter~\ref{Chapter1}.
As a quantifier for entanglement, we mostly use negativity, which is a computable entanglement monotone~\cite{zyczkowski1998volume,lee2000entanglement,lee2000partial,negativity}.
For two objects $X$ and $Y$, the negativity is denoted by $N_{X:Y}$ whose maximum is given by $(d-1)/2$, where $d=\min\{d_X,d_Y\}$ with $d_X (d_Y)$ being the dimension of the object $X (Y)$.
This maximum entanglement is given by pure states of the form
\begin{equation}
|\Psi_{\text{max}} \rangle = \frac{1}{\sqrt{d}} \sum_{j = 1}^{d} |x_j \rangle |y_j \rangle,
\label{EQ_sp_msent}
\end{equation}
where $| x_j \rangle$ and $| y_j \rangle$ form orthonormal bases. 
We will also use the relative entropy of entanglement (REE)~\cite{REEmutualinfo}, denoted as $E_{X:Y}$. 
The maximum of REE is given by $\log_2(d)$.
We consider the relative entropy of discord~\cite{modi2010unified}, which is also referred to as the one-way quantum deficit~\cite{deficit}.
It is an asymmetric quantity denoted as $\Delta_{X|Y}$.
Similar to REE, maximum discord is given by $\log_2(d)$.
To quantify the amount of total correlation, we use the mutual information $I_{X:Y}$ \cite{groisman2005}. 
Finally, on some occasions, we use conditional entropy $S_{X|Y}$.

\section{Entangling speed limit: Direct interactions}\label{SC_dents}

Let us now consider the quickest way to entangle our principal objects $A$ and $B$ with direct interaction. 
For an exemplary dynamics, assume for the moment, that all the objects are qubits. 
Starting with an initially disentangled states, e.g., $\ket{00}$, the optimal entangling speed is obtained with the following Hamiltonian
\begin{equation}\label{EQ_xx}
H_{AB}=\hbar \Omega\: \sigma_A^x\otimes \sigma_B^x,
\end{equation}
where $\sigma^x$ is the $x$-Pauli matrix and $\Omega$ denotes the frequency unit.
It is easy to see that the state at time $t$ takes the form $\ket{\psi(t)}=\cos(\Omega t)\ket{00}-i\sin(\Omega t)\ket{11}$, and therefore the state of $AB$ is oscillating between the disentangled states $\ket{00}$ or $\ket{11}$ and a maximally entangled state $(\ket{00}-i\ket{11})/\sqrt{2}$.

We note that in small time, the component $\ket{11}$ grows $\propto \Delta t$, which means that entanglement, as quantified by negativity, $N_{A:B}$ increases immediately.
In other words, the rate of entanglement follows $\dot N_{A:B}(0)>0$.
It is also easy to confirm that $\min \{ \langle M\rangle,\Delta M\}=1$ in this example, meaning that the time bound is simply given by the Bures angle, i.e., $\Gamma(\rho_{AB}(0),\rho_{AB}(T))=\Theta(\rho_{AB}(0),\rho_{AB}(T))$.\footnote{We use $T=\Omega t$ as a dimensionless time variable.}
The corresponding dynamics of entanglement and Bures angle is illustrated in Fig.~\ref{FIG_sp_exp1}. 
One can see that this dynamics saturates the time bound, i.e., the time equals the time bound $\Theta$ (dashed line in Fig.~\ref{FIG_sp_exp1}).

\begin{figure}[h]
\centering
\includegraphics[scale=0.45]{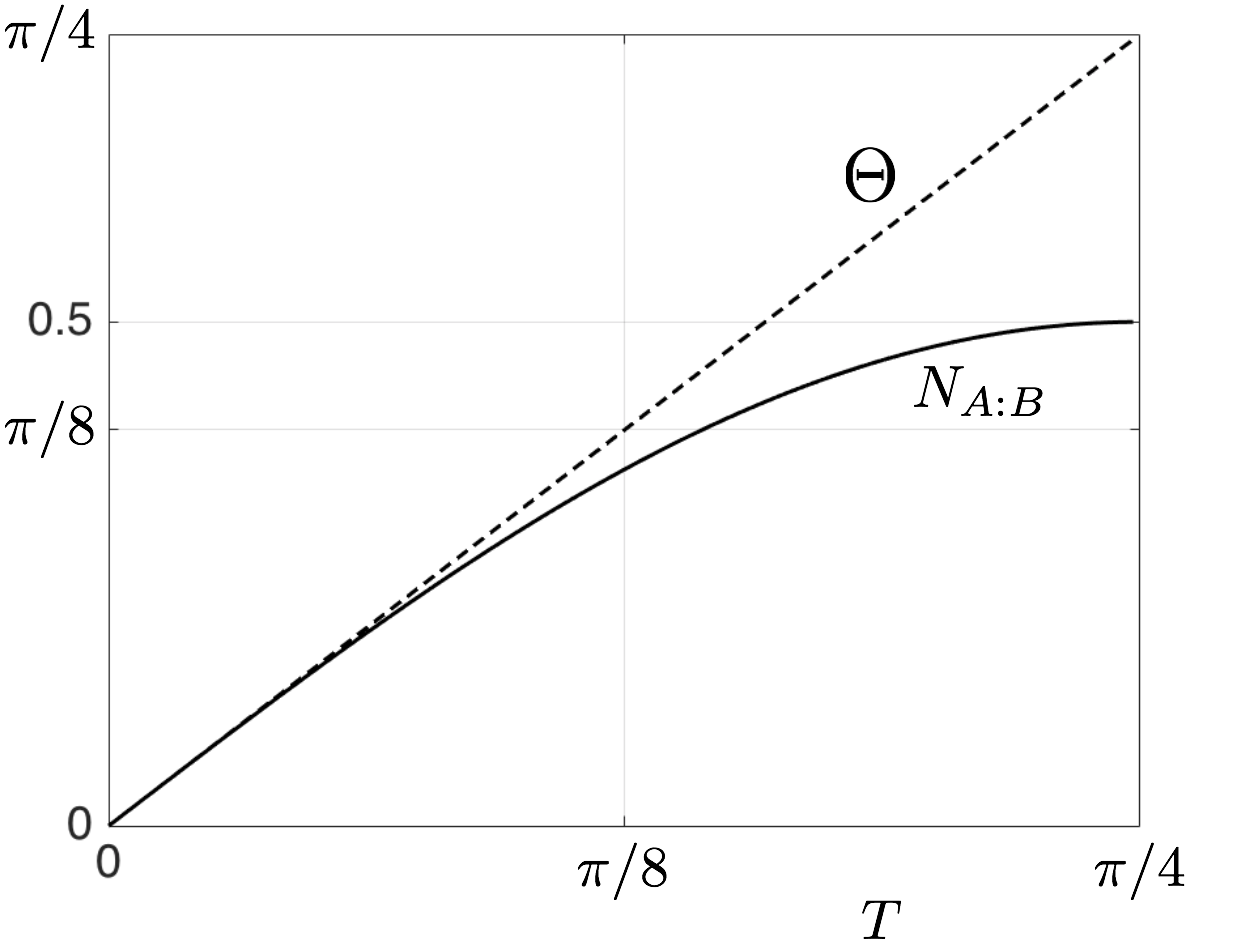}
\caption{Exemplary dynamics showing maximum entangling speed between two qubits. 
The Hamiltonian is given by $H_{AB}=\hbar \Omega\: \sigma_A^x\otimes \sigma_B^x$ with initial state $\ket{00}$.
Maximum negativity (solid curve) is achieved in $t=\pi/4\Omega$.
}
\label{FIG_sp_exp1}
\end{figure}

From the simple example above, it is apparent to see the generalisation of the bound for the case where the objects under consideration are $d$-dimensional.
In particular, by having the resource equality, the direct bound reads 
\begin{equation}
\Gamma_{\text{di}}=\arccos|\langle 00|\Psi_{\text{max}}\rangle|=\arccos\left(\frac{1}{\sqrt{d}}\right).
\end{equation}
For the actual dynamics, one has to figure out the Hamiltonian that both gives the resource equality and evolves the initial state to a state with given fidelity in the shortest possible time.

\section{Entangling speed limit: Indirect interactions}

It is now reasonable to ask whether the dynamics in Fig.~\ref{FIG_sp_exp1} can be sped up or at least replicated through indirect interactions. 
Although it is definitely possible to have entanglement growth from zero to maximum, we will show that it is impossible to beat the speed set by the direct interaction scenario with initial state $\ket{00}$ for system $AB$. 

Let us begin with the following theorem, which investigates the initial rate of entanglement gain in the partition $A:B$. 

\begin{theorem}\label{TH_rate}
Consider the indirect interaction scenario, where an ancilla $C$ is mediating interactions between objects $A$ and $B$. 
We allow all the objects to be open to their own local environments. 
If the initial state of the tripartite system is of the form $\rho_0=\rho_{AB}\otimes \rho_C$, then the rate of any entanglement monotone reads $\dot E_{A:B}(0)\le 0$. 
\end{theorem}

\begin{proof}
Let us start with the evolution of the tripartite state following the Lindblad form
\begin{eqnarray}
\frac{\rho_{\Delta t}-\rho_0}{\Delta t} & = & -\frac{i}{\hbar}[H,\rho_0]+\sum_{X=A,B,C}L_X\rho_0, \label{EQ_open}\\
L_X\rho_0 & = & \sum_k Q_X^k\rho_0 Q_{X}^{k\dag}-\frac{1}{2}\{Q_{X}^{k\dag}Q_X^k,\rho_0\}, \nonumber
\end{eqnarray}
where the first term in (\ref{EQ_open}) is the coherent part of the evolution, whereas the second term is the incoherent part and $L_X$ describes interactions of system $X$ with its local environment, i.e., the operators $Q_X^k$ act on system $X$ only. Without loss of generality we assume the Hamiltonians of the form $H=H_A\otimes H_{C_1}+H_B\otimes H_{C_2}$, where $H_{C_2}$ can be different from $H_{C_1}$. Note that the generalisation of the proof to Hamiltonians $H=\sum_{\mu} H_A^{\mu}\otimes H_{C_1}^{\mu}+\sum_{\nu} H_B^{\nu}\otimes H_{C_2}^{\nu}$ is straightforward.

The state of $AB$ at time $\Delta t$ is calculated as follows
\begin{eqnarray}
\rho_{AB}(\Delta t)&=&\mbox{tr}_C(\rho_{\Delta t})\nonumber \\
&=&\mbox{tr}_C( \rho_0-i\frac{\Delta t}{\hbar} [H,\rho_0]+\Delta t \sum_{X=A,B,C}L_X\rho_0 )\nonumber \\
&=&\rho_{AB}-i\frac{\Delta t}{\hbar}[H_AE_{C_1}+H_BE_{C_2},\rho_{AB}]\nonumber \\
&&+\Delta t(L_A+L_B)\rho_{AB}, \label{EQ_rabdt}
\end{eqnarray}
where we have used $\rho_0=\rho_{AB}\otimes \rho_C$, and $E_{C_1}=\mbox{tr}(H_{C_1}\rho_C)$ and $E_{C_2}=\mbox{tr}(H_{C_2}\rho_C)$ are the initial mean energies.\footnote{We will assume that these constants are dimensionless and the energy unit is transferred entirely to the terms $H_A$ and $H_B$, respectively.} We have also utilised the cylic property of trace such that $\mbox{tr}_C(Q_C^k\rho_C Q_{C}^{k\dag}-\frac{1}{2}\{Q_{C}^{k\dag}Q_C^k,\rho_C\})=0$. 

One can see from Eq. (\ref{EQ_rabdt}) that the state of $AB$ at time $\Delta t$ is simply as a result of effective local Hamiltonians $H_AE_{C_1}+H_BE_{C_2}$ and local interactions with environments. 
For any entanglement monotone, i.e., non-increasing under local operations and classical communication (LOCC), we have $E_{A:B}(\Delta t)\le E_{A:B}(0)$ and therefore $\dot E_{A:B}(0)\le 0$.
\end{proof}

For example, if we consider unitary dynamics and the initial pure state $\ket{00}$ of $AB$, it requires $C$ to be in a decoupled form, i.e., the whole system assumes $\ket{00}\bra{00}\otimes \rho_C$.
From Eq.~(\ref{EQ_rabdt}) of Theorem~\ref{TH_rate}, it is easy to see that $E_{A:B}(\Delta t)= E_{A:B}(0)$, i.e., $\dot E_{A:B}(0)= 0$. 
This is due to the fact that unitary operations resulting from local Hamiltonians leave the entanglement invariant. 
This implies that, in order to have positive (or negative) initial entanglement rate, there has to be initial correlation between the mediator $C$ and the principal objects $A$ and $B$.

Note also that if $C$ is uncorrelated with $AB$ at all times, i.e., $I_{AB:C}(t)=0$, one can apply the method in Theorem \ref{TH_rate} successively and conclude that entanglement distribution is impossible as $E_{A:B}(t)\le E_{A:B}(0)$. 
This has been proven similarly in Chapter~\ref{Chapter_revealing}.
This way, entanglement gain between the principal objects can be used as a witness of correlation with the mediator $C$, i.e., $I_{AB:C}>0$ at some time during the evolution.

Theorem~\ref{TH_rate} is not enough to conclude the impossibility of beating the direct bound as it does not explain what is in play during the dynamics.
It might happen that although the initial rate of production is non-positive, the rate during the dynamics is high such that after some finite time $T<\Gamma_{\text{di}}$ the maximum entanglement is reached.
We will show that this scheme is impossible by utilising the following ultimate bound.

\begin{theorem}\label{TH_untimate}
Consider any coherent dynamics of three objects $A$, $B$, and $C$ with the total Hamiltonian given by $H$. This may include the direct interaction scenario between $A$ and $B$ as well as the indirect interaction where $C$ is the mediator. Starting with initial density matrix of the form $\rho(0)=\rho_{ABC}$, where the reduced state $\rho_{AB}$ is separable, we show that the time it takes to maximally entangle $A$ and $B$ is lower bounded as
\begin{equation}
T\ge \arccos\left(\frac{1}{\sqrt{d}}\right),
\end{equation}
where $T=\Omega t$ and $\min \{\langle H\rangle,\Delta H\}=\hbar \Omega$.
\end{theorem}
\begin{proof}
First let us note that maximum entanglement between $A$ and $B$ implies that the state of $AB$ is pure, and that the state of $C$ is decoupled, i.e., $\rho(T)=\ket{\Psi_{\text{max}}}\bra{\Psi_{\text{max}}} \otimes \rho_C^{\prime}$.
Taking the fidelity of the initial and final states gives
\begin{eqnarray}\label{EQ_utb}
\mathcal{F}(\rho(0),\rho(T))&=&\mathcal{F}(\rho_{ABC},\ket{\Psi_{\text{max}}}\bra{\Psi_{\text{max}}}\otimes \rho_C^{\prime})\nonumber \\
&\le&\mathcal{F}(\rho_{AB},\ket{\Psi_{\text{max}}}\bra{\Psi_{\text{max}}}) \nonumber \\
&=&\sqrt{\langle \Psi_{\text{max}} |\rho_{AB}|\Psi_{\text{max}}}\rangle \nonumber \\
&\le&\max_{p_j,\ket{a_jb_j}} \sqrt{\sum_j p_j |\langle a_jb_j|\Psi_{\text{max}}\rangle|^2} \nonumber \\
&\le&\max_{\ket{a_jb_j}}  |\langle a_jb_j|\Psi_{\text{max}}\rangle| \nonumber \\
&=&\frac{1}{\sqrt{d}},
\end{eqnarray}
where we have used the non-decreasing property of fidelity under trace-preserving completely positive maps (in our case, tracing-out object $C$) in the first line~\cite{nielsen1996entanglement}, expressed the separable state as $\rho_{AB}=\sum_j p_j\:\ket{a_jb_j}\bra{a_jb_j}$, and used convexity in the last inequality.

Finally, by utilising the QSL of Eq.~(\ref{EQ_dbound}) and noting that $\min \{\langle H\rangle,\Delta H\}=\hbar \Omega$, one obtains
\begin{equation}
T\ge \arccos\left({\mathcal{F}(\rho(0),\rho(T))}\right) \ge \arccos{\left(\frac{1}{\sqrt{d}}\right)}=\Gamma_{\text{di}}.
\end{equation}
\end{proof}

Let us now discuss a special case of Theorem~\ref{TH_untimate} above. 
In particular, consider the initially decoupled mediator, i.e., $\rho(0)=\rho_{AB}\otimes \rho_C$.
Consequently, we have 
\begin{eqnarray}\label{EQ_tbspecial}
\mathcal{F}(\rho(0),\rho(T))&=&\text{tr} \left( \sqrt{\sqrt{\rho_{AB}\rho_C} \ket{\Psi_{\text{max}}}\bra{\Psi_{\text{max}}} \rho_C^{\prime} \sqrt{\rho_{AB}\rho_C}  }\right)\nonumber \\
&=&\text{tr} \left( \sqrt{\sqrt{\rho_{AB}} \ket{\Psi_{\text{max}}}\bra{\Psi_{\text{max}}}\sqrt{\rho_{AB}}  \sqrt{\rho_C} \rho_C^{\prime}\sqrt{\rho_C}}\right)\nonumber \\
&=&\text{tr} \left( \sqrt{\sqrt{\rho_{AB}} \ket{\Psi_{\text{max}}}\bra{\Psi_{\text{max}}}\sqrt{\rho_{AB}} }\right) \text{tr} \left( \sqrt{\sqrt{\rho_C} \rho_C^{\prime}\sqrt{\rho_C}}\right)\nonumber \\
&=&\mathcal{F}(\rho_{AB},\ket{\Psi_{\text{max}}}\bra{\Psi_{\text{max}}})\:\mathcal{F}(\rho_{C},\rho_C^{\prime}).
\end{eqnarray}
Note that as the fidelity $0\le \mathcal{F}(\rho_{C},\rho_C^{\prime})\le 1$, we recover the second line of Eq.~(\ref{EQ_utb}).
Now let us discuss the consequence of the initial state $\rho(0)=\rho_{AB}\otimes \rho_C$.
From Eq.~(\ref{EQ_tbspecial}), we notice that if $\rho_{AB}$ is strictly mixed, then for a unitary dynamics (purity-preserving), $\rho_C^{\prime}$ will be more mixed than $\rho_C$, i.e., $\rho_C\ne \rho_C^{\prime}$.
This means $\mathcal{F}(\rho_C,\rho_C^{\prime})<1$, making the corresponding first inequality in (\ref{EQ_utb}) strict, and resulting in a strict time bound $T>\Gamma_{\text{di}}$.
Therefore, in order to saturate the direct bound, the state $\rho_{AB}$ has to be pure and $\mathcal{F}(\rho_C,\rho_C^{\prime})=1$.
A trivial dynamics (via direct interactions in a tripartite setting) is given by the example presented in Section~\ref{SC_dents} with the addition of a decoupled mediator $\rho_C$ and a local Hamiltonian $H_C$.

\subsection{Saturating the limit}\label{SC_satlimit}

Now we present examples of indirect interactions that saturate the direct time bound.
As initial states of the form $\rho_{AB}\otimes \rho_C$ cannot be used for maximum entanglement speed, one would normally think of utilising entanglement in the partition $AB:C$.
In particular, consider the following initial state and Hamiltonian:
\begin{eqnarray}
\ket{\psi(0)}&=&\frac{1}{\sqrt{2}}(\ket{000}+\ket{111}),\nonumber \\
H&=&\frac{\hbar \Omega}{2\sqrt{2}}(\sigma^z_A\otimes H_{C_1}+\sigma^z_B\otimes H_{C_2}),
\end{eqnarray}
where $H_{C_1}=-(\openone +\sigma^x_C+\sigma^y_C+\sigma^z_C)$ and $H_{C_2}=\openone -\sigma^x_C-\sigma^y_C+\sigma^z_C$.
One finds that this example has $\min \{ \langle M\rangle,\Delta M\}=1$ and therefore the bound is also given by the Bures angle $\Theta(\rho_{AB}(0),\rho_{AB}(T))$, where now $\rho_{AB}(T)=\mbox{tr}_C\left(\rho_{ABC}(T)\right)$.
Furthermore, the resulting dynamics is the same as that in Fig.~\ref{FIG_sp_exp1}. 
From this example we note that the initial state has some purely quantum properties: entanglement $N_{AB:C}(0)=0.5$ and mutual information $I_{AB:C}(0)=2$.

However, is it necessary to have entanglement with the mediator $C$ to saturate the entangling speed of direct interactions?
We now show that this is \emph{not} the case.
In particular, it is enough to have classical correlation with $C$ initially.
For this purpose, we take the initial state and Hamiltonian
\begin{eqnarray}\label{EQ_exp3}
\rho(0)&=&\frac{1}{2}\ket{\psi_+}\bra{\psi_+}\otimes \ket{+}\bra{+}+\frac{1}{2}\ket{\phi_+}\bra{\phi_+}\otimes \ket{-}\bra{-},\nonumber \\
H&=&\frac{\hbar \Omega}{2}(\sigma^x_A\otimes \sigma^x_C+\sigma^x_B \otimes \sigma^x_C),
\end{eqnarray}
where $\sigma^x \ket{\pm}=\pm \ket{\pm}$, and $\ket{\psi_+}=(1/\sqrt{2})(\ket{01}+\ket{10})$ and $\ket{\phi_+}=(1/\sqrt{2})(\ket{00}+\ket{11})$ are two Bell states of party $AB$.\footnote{Note that this example is similar to that already presented in Chapter~\ref{Chapter_revealing} to demonstrate entanglement localisation.}
One can see that the initial state in Eq. (\ref{EQ_exp3}) is disentangled in the partition $AB:C$ and is correlated with mutual information $I_{AB:C}(0)=1$.
This example also has the resource equality satisfied and the dynamics of that in Fig.~\ref{FIG_sp_exp1}.
Furthermore, the dynamics from Eq. (\ref{EQ_exp3}) has been shown to have zero quantum discord $D_{AB|C}$ at all times (see Chapter~\ref{Chapter_revealing}), showing that the correlation between $C$ and $AB$ is purely classical.
One can extend further and see that initial states of the form $p\ket{\psi_+}\bra{\psi_+}\otimes \ket{+}\bra{+}+(1-p)\ket{\phi_+}\bra{\phi_+}\otimes \ket{-}\bra{-}$, where $p$ stands for probability, can be made maximally entangled by Hamiltonian in Eq.~(\ref{EQ_exp3}).
This maximum entanglement is achieved in $T=\pi/4$ for $0<p<1$, see Fig.~\ref{FIG_sp_exp3}.
Independent of $p$, the Bures angle dynamics is given by the dashed line in Fig.~\ref{FIG_sp_exp1}, which indicates saturation.

\begin{figure}[h]
\centering
\includegraphics[scale=0.37]{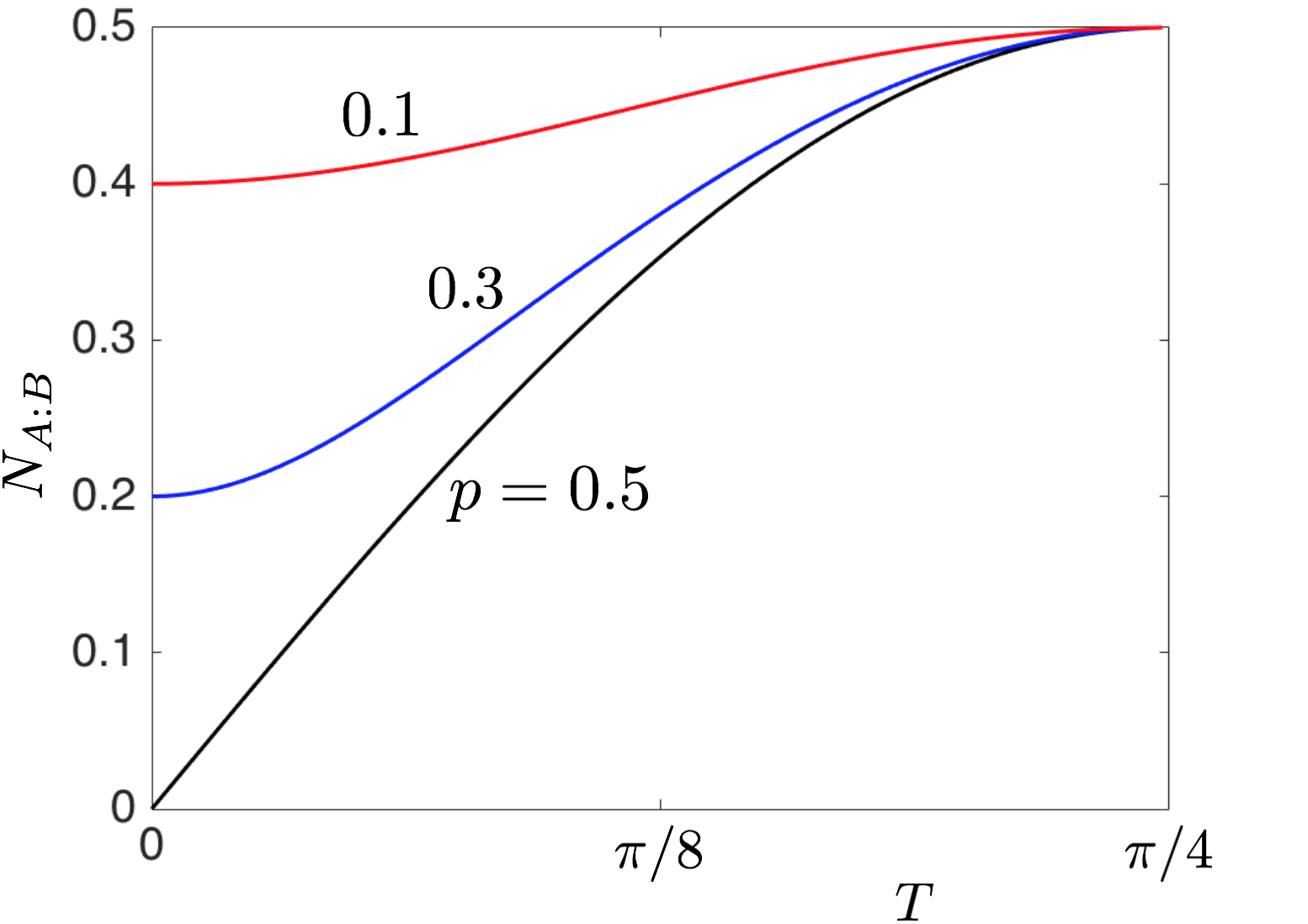}
\caption{Maximum entangling speed for various initial states.
The Hamiltonian is taken to be $H=\frac{\hbar \Omega}{2}(\sigma^x_A\otimes \sigma^x_C+\sigma^x_B \otimes \sigma^x_C)$ with initial state $p\ket{\psi_+}\bra{\psi_+}\otimes \ket{+}\bra{+}+(1-p)\ket{\phi_+}\bra{\phi_+}\otimes \ket{-}\bra{-}$.
Entanglement is plotted for $p=0.1$ (red curve), $0.3$ (blue curve), and $0.5$ (black curve).
}
\label{FIG_sp_exp3}
\end{figure}

Let us compare the three examples shown above. 
First of all, it is clear that starting with pure states, e.g., $\ket{00}$, it is possible to get maximum entanglement via direct interactions.
Also note that during the evolution, the purity of the state is preserved.
This is because unitary operations do not change purity.
Therefore, if one were to start with a mixed initial state $\rho_{AB}$, the maximum entanglement can never be achieved. 
This is true as maximum entanglement is given by pure states (as measured by any entanglement monotone) \cite{streltsov2012general}. 
In the indirect interaction scenario, the evolution of the state of $AB$ is not unitary. 
Therefore, the dynamics can change the purity of the state $\rho_{AB}$.
For example, in the last two examples, $\rho_{AB}$ changes from having zero entanglement with purity $0.5$ to maximum entanglment with purity $1$. 
Hence, for mixed initial states, one would want to utilise mediators to distribute maximum entanglement between the principal objects $A$ and $B$.

Finally, one might wonder, for an initial mixed state, is it possible to achieve maximum entanglement with direct interaction setup, this time allowing the principal objects to be open to their local environments?
This question is sensible as interactions with environments can change the purity of the principal objects, and therefore allowing them to reach maximum entanglement.
To illustrate this, let us assume a scenario where one of the principal objects, say $A$, is interacting with a single-qubit environment $C$.
Furthermore, take the initial state of Eq. (\ref{EQ_exp3}) and Hamiltonian of the form $H=\hbar \Omega \:\sigma^x_A\otimes \sigma^x_C$. 
One can confirm that the resulting $A:B$ entanglement dynamics is identical as that in Fig.~\ref{FIG_sp_exp1}.

\section{Initial correlation with mediators bounds localised entanglement}

Note that the example in (\ref{EQ_exp3}), the entanglement localisation, realises maximum entanglement with dynamics saturating the direct bound by having initial state with only classical correlation in the partition $AB:C$.
It is then natural to ask what bounds the entanglement gain.
In particular, whether the gain is related to the initial mutual information $I_{AB:C}(0)$.
In order to answer this question, let us begin by presenting the following Lemma.

\begin{lemma}\label{LM_initial}
For a three-party system $ABC$, where the subsystem $C$ is only classically correlated, i.e., $\rho=\sum_cp_c\:\rho_{AB|c}\otimes \ket{c}\bra{c}$ with $\{\ket{c}\}$ forming an orthonormal basis, the relative entropy of entanglement follows the bound
\begin{equation}
E_{A:BC}-E_{A:B}\le I_{AB:C}.
\end{equation}
\end{lemma}
\begin{proof}
By definition of the REE we have $E_{A:BC}(\rho)= -\mbox{tr}(\rho \log_2 \sigma)-S(\rho)$, where $\sigma$ is the closest separable state to $\rho$ \cite{REEmutualinfo}. 
According to the flags condition $E_{A:BC}=\sum_c p_c\: E_{A:B}(\rho_{AB|c})$, where $E_{A:B}(\rho_{AB|c})=-\mbox{tr}(\rho_{AB|c}\log_2 \sigma_{AB|c})-S(\rho_{AB|c})$ \cite{flags}.
Also, it has been shown that $\sigma=\sum_cp_c\:\sigma_{AB|c}\otimes \ket{c}\bra{c}$, where $\sigma_{AB|c}$ is a separable state closest to $\rho_{AB|c}$~\cite{flags}. 
Next, we note for the first term in the definition of $E_{A:BC}$:
\begin{eqnarray}
&&-\mbox{tr}\left(\sum_cp_c\:\rho_{AB|c}\otimes \ket{c}\bra{c}\log_2(\sum_jp_j\:\sigma_{AB|j}\otimes \ket{j}\bra{j})\right)\nonumber \\
&=&-\mbox{tr}\left(\sum_cp_c\:\rho_{AB|c}\otimes \ket{c}\bra{c}\sum_j\log_2(p_j\:\sigma_{AB|j})\otimes \ket{j}\bra{j}\right)\nonumber \\
&=&-\mbox{tr}\left(\sum_cp_c\:\rho_{AB|c}\log_2(p_c\:\sigma_{AB|c})\right)\nonumber \\
&=&-\sum_cp_c (\log_2(p_c)+\mbox{tr}(\:\rho_{AB|c}\log_2(\sigma_{AB|c})))\nonumber \\
&=&S_C+\sum_cp_c (-\mbox{tr}(\:\rho_{AB|c}\log_2(\sigma_{AB|c}))),
\end{eqnarray}
where each term in the sum is minimised as $\sigma_{AB|c}$ is closest to $\rho_{AB|c}$.
On the other hand, we have $E_{A:B}=-\mbox{tr}(\sum_c p_c\: \rho_{AB|c}\log_2 \sigma_{AB})-S_{AB}$, where $\sigma_{AB}$ is the closest separable state to $\sum_c p_c\: \rho_{AB|c}$.
Putting things together we have $E_{A:BC}-E_{A:B}=I_{AB:C}+\eta$,
where $I_{AB:C}=S_{AB}+S_C-S_{ABC}$ and 
\begin{eqnarray}
\eta&=&\sum_cp_c\:[-\mbox{tr}(\:\rho_{AB|c}\log_2(\sigma_{AB|c}))]\nonumber \\
&&-\sum_cp_c\:[-\mbox{tr}(\:\rho_{AB|c}\log_2(\sigma_{AB}))].
\end{eqnarray}
One can see that $\eta$ is negative, since $-\mbox{tr}(\:\rho_{AB|c}\log_2(\sigma_{AB|c}))$ is minimised, confirming the Lemma.
\end{proof}

Now, for the dynamics, let us consider the indirect interaction scenario where all our objects are allowed to be open to their own local environment. 
The entanglement bound is proven in the following Theorem.

\begin{theorem}\label{TH_ebound}
Consider the indirect interaction scenario where all the objects $A$, $B$, and the mediator $C$ are open to their local environment.
If $C$ is classical at all times, i.e., $D_{AB|C}(t)=0$ for $t\in [0,\tau]$, we have
\begin{equation}\label{EQ_bb}
E_{A:B}(\tau)-E_{A:B}(0)\le I_{AB:C}(0).
\end{equation} 
\end{theorem}
\begin{proof}
The proof is done with the following chain of equations:
\begin{eqnarray}
E_{A:B}(\tau)-E_{A:B}(0)&\le&E_{A:BC}(\tau)-E_{A:B}(0) \label{EQ_11}\\
&\le &E_{A:BC}(0)-E_{A:B}(0)\label{EQ_22}\\
&\le&I_{AB:C}(0),\label{EQ_33}
\end{eqnarray}
where the steps are justified as follows.
The inequality (\ref{EQ_11}) is due to the monotonicity of entanglement under local operations (tracing out the mediator $C$), i.e., $E_{A:B}(\tau)\le E_{A:BC}(\tau)$.
Line (\ref{EQ_22}) follows from Theorem \ref{TH_revealing}, stating that entanglement in the partition $A:BC$ cannot grow via classical $C$, that is $E_{A:BC}(\tau)\le E_{A:BC}(0)$.
Finally, Lemma \ref{LM_initial} confirms line (\ref{EQ_33}).
\end{proof}

The bound in Eq.~(\ref{EQ_bb}) can be made simpler at the expense of tightness.
Let us note that for states separable in the partition $AB:C$ (including states where $C$ is classical, i.e., $\rho=\sum_cp_c\:\rho_{AB|c}\otimes \ket{c}\bra{c}$), the conditional entropy $\{S_{AB|C},S_{C|AB}\}\ge 0$.
This means that the mutual information can be bounded as $I_{AB:C}=S_{AB}-S_{AB|C}\le S_{AB}$ or $I_{AB:C}=S_{C}-S_{C|AB}\le S_{C}$.
Therefore, entanglement gain via classical mediators cannot be larger than $\log_2(d_C)$ as it is the maximum entropy for system $C$.
We would also like to mention that if one observes $E_{A:B}(t)-E_{A:B}(0)> S_{AB}(0)$ then there had to be non-zero quantum discord $D_{AB|C}$ during the evolution. 
This is a witness of non-classicality of $C$ by observing only the objects $A$ and $B$, similar to the detection of quantum discord in Chapter~\ref{Chapter_revealing}.

As mentioned in Preliminaries, the evolution of a quantum state can be sped up simply by providing more energy to the system.
From the results presented in this chapter, one can also speed up the creation of entanglement between two objects via a mediating system that is initially correlated with them.
Indeed, for initially uncorrelated systems, it would take longer time to reach maximum entanglement.
In a way, this puts the two quantities, energy and correlations, on the same footing. 
Both can be seen as a resource for fast distribution of quantum entanglement.

\section{Charging power}

In this section we study the quickest way to flip a quantum bit, i.e., to evolve it from the ground state $\ket{0}$ to the excited state $\ket{1}$.
This is known as charging.
Also for this task, we normalise the resource by having $\min \{\langle M\rangle,\Delta M\}=1$.
For simplicity, let us first consider a system of two qubits. 
One would say that even free Hamiltonian of the form $H=\hbar \Omega(\sigma^x_A+\sigma^x_B)$ will suffice.
While this will charge the system, it does not provide the quickest process as we will show below.
Starting with the state $\ket{00}$, the resource equality requirement sets the free Hamiltonian to be $H_{\text{free}}=\hbar \Omega (\sigma^x_A+\sigma^x_B)/\sqrt{2}$.
As can be seen from the blue dashed line in Fig.~\ref{FIG_chrg}, this dynamics does not saturate the time bound.

We consider another Hamiltonian, $H_{\text{int}}=\hbar \Omega(\sigma^x_A\otimes \sigma^x_B)$, which involves interactions between $A$ and $B$.
Indeed, with this Hamiltonian, charging is done fastest as illustrated by the black dashed line in Fig.~\ref{FIG_chrg}.
The generalisation of the fastest charging to $N$ qubits simply reads
\begin{eqnarray}
\ket{\psi(0)}&=&\ket{0}^{\otimes N},\nonumber \\
H&=&\hbar \Omega\: (\sigma^x_1\otimes \sigma^x_2 \otimes \cdots \otimes \sigma^x_N).
\end{eqnarray}
Independent of $N$, fully charged state $\ket{1}^{\otimes N}$ is achieved in $T=\pi/2$.

\subsection{Charge quantifier}
Let us consider two quantum objects $A$ and $B$. 
The goal is to evolve the state from $\ket{00}$ to the fully charged state $\ket{11}$. 
As a way of measuring how close the state $\rho_{AB}(t)$ is to the target state, one may define the charge using the root fidelity as $\Xi_{AB}(t)=\mathcal{F}(\ket{11}\bra{11},\rho_{AB}(t))$.
The value would go from $0$ to $1$ (fully charged).

For direct interactions, where the state is pure, the charge simply reduces to $\Xi_{AB}(t)=|\langle11|\psi_{AB}(t)\rangle|$.
This would mean that for immediate or fast charging, the state $|\psi_{AB}(\Delta t)\rangle$ has to have a component $\ket{11}$ in the time of order $\Delta t$ such that $\Xi_{AB}(\Delta t)>0$, i.e., $\dot {\Xi}_{AB}(0)>0$.
This scenario is nicely illustrated in Fig.~\ref{FIG_chrg} by using the $H_{\text{int}}$.
On the other hand, free Hamiltonians of the form $H_A+H_B$ would evolve the subsystems separately, such that the state at time $\Delta t$ reads $\ket{\psi_{AB}(\Delta t)}=(\openone-i\Delta t H_A/\hbar)\ket{0}\otimes(\openone-i\Delta t H_B/\hbar)\ket{0}$.
One can easily see that, in this case, $\Xi_{AB}(\Delta t)=0$.
From these observations, we note that positive charging rate is linked with the presence of entanglement between objects $A$ and $B$.
This is because entanglement is nonzero for the state $\ket{\psi_{AB}(\Delta t)}$ with $\ket{11}$ component.

\begin{figure}[!h]
\centering
\includegraphics[scale=0.55]{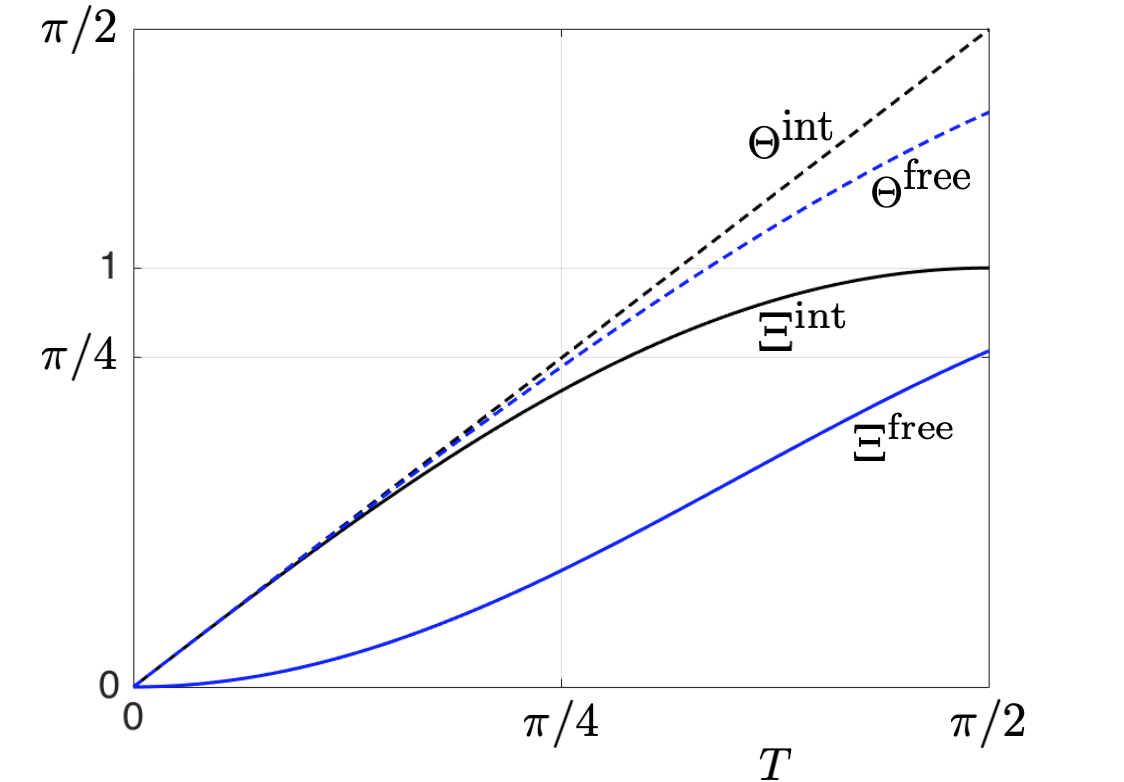}
\caption{Dynamics of charge $\Xi$ (solid curves) and Bures angle (dashed lines) for interacting (black) and non-interacting (blue) qubits. 
In both cases, the resource has been fixed to the same amount, i.e., $\min \{\langle M\rangle,\Delta M\}=1$.
}
\label{FIG_chrg}
\end{figure}

For indirect interactions, with system $C$ as the mediator, the initial state is of the form $\ket{00}\bra{00}\otimes \rho_C$.
We show with the following Lemma that an initially correlated mediator is also a factor for having positive charging rate.
This is similar to the entanglement distribution task.

\begin{lemma}
For indirect interactions, where all objects can be open to their local environment, the charging rate follows $\dot {\Xi}_{AB}(0)=0$.
\end{lemma}
\begin{proof}
First, let us express the charge as 
\begin{equation}
\Xi_{AB}(t)=\mathcal{F}(\ket{11}\bra{11},\rho_{AB}(t))=\sqrt{\langle 11|\rho_{AB}(t)|11\rangle}.
\end{equation} 
Next, from Eq. (\ref{EQ_rabdt}), the state of $AB$ at time $\Delta t$ is given by 
\begin{eqnarray}
\rho_{AB}(\Delta t)&=&\ket{00}\bra{00}-i\frac{\Delta t}{\hbar}[H_AE_C+H_BE_{\gamma},\ket{00}\bra{00}]\nonumber \\
&&+\Delta t(L_A+L_B)\ket{00}\bra{00},
\end{eqnarray}
where we have used $\rho_0=\ket{00}\bra{00}\otimes \rho_C$.
One can see that $\Xi_{AB}(\Delta t)=0$ from the observation that $\langle 11|\rho_{AB}(\Delta t)|11\rangle=0$, which proves the Lemma.
\end{proof}

\section{Summary}
Quantum speed limit sets the minimum time required for evolving a quantum state and therefore also its properties. 
In this chapter we have presented the time bound for the aim of distributing quantum entanglement between two objects both via direct as well as indirect interactions. 
We have shown that the indirect interaction setting cannot beat the direct interaction setting in terms of entangling speed. 
Furthermore, the correlation between the mediator and the principal objects is required for optimal distribution.
From the perspective of speeding up the creation of correlations, we discussed that both energy and correlated mediators play similar roles.
We also presented briefly a simple application of the quantum speed limit to charging of quantum batteries and compared the direct and indirect interaction settings.


\chapter{Observable quantum entanglement due to gravity} 

\label{Chapter_gravity} 

\lhead{Chapter 6. \emph{Observable quantum entanglement due to gravity}} 

\emph{This chapter presents experimental setups where quantum entanglement is generated through gravitational coupling.\footnote{Parts of this chapter are reproduced from our published article of Ref.~\cite{backtoback}, which is licensed under the Creative Commons Attribution 4.0 International License (http://creativecommons.org/licenses/by/4.0/). Where applicable, changes made will be indicated.}
Under certain natural assumptions, this in turn shows the quantum nature of gravity, which is one of the most important quests in modern physics.
Firstly, a motivation is presented with a focus on the need for experimental evidence of quantum gravity. 
Some of past experiments where gravity affected quantum matter are briefly reviewed.
I will then proceed to introduce our proposal that involves two masses coupled gravitationally.
Two cases will be considered in which the masses are either trapped in harmonic potentials or released from the traps.
I will present details of the dynamics including the calculations of analytic figures of merit that characterise the amount of generated entanglement and the accumulation time.
An analysis of other effects including environmental noises, decoherence mechanisms, and Casimir interactions will also be presented.
Finally, I discuss the conclusion one can draw from our proposal.
}

\clearpage

\section{Motivation and objectives}
The successful unification of electromagnetic, weak, and strong interactions within the quantum framework strongly suggests that gravity should also be quantised.
Up to date, however, there is no experimental evidence of quantum features of gravity.
In numerous experiments gravity is key to the interpretation of the observed data, 
but it is often sufficient to use Newtonian theory (quantum particle moving in a background classical field) or general relativity (quantum particle moving in a fixed spacetime) to gather a meaningful understanding of such data.
Milestone experiments described within Newtonian framework include gravity-induced quantum phase shift in a vertical neutron interferometer~\cite{gphase}, 
precise measurement of gravitational acceleration by dropping atoms~\cite{gravimeter}, 
or quantum bound states of neutrons in a confining potential created by the gravitational field and a horizontal mirror~\cite{nfall}.
Quantum experiments that require general relativity include gravitational redshift of electromagnetic radiation~\cite{gshift1} or time dilation of atomic clocks at different heights~\cite{clock1}.

A number of theoretical proposals discussed scenarios capable of revealing quantumness of gravity.
For example Refs.~\cite{massg1,massg2,massg22,massg3,massg4,massg5,massg6,qgdis,cavenexp} proposed the observation of a probe mass interacting with the gravitational field generated by another mass.
More recent proposals put gravity in a role of mediator of quantum correlations and are based on the fact that quantum entanglement between otherwise non-interacting objects can only increase via a quantum mediator \cite{krisnanda2017,gravity1,gravity2}, see Chapter~\ref{Chapter_revealing}.
Motivated by these proposals and by advances in optomechanics \cite{RMP.86.1391}, in particular the cooling of massive mechanical (macroscopic) oscillators close to their quantum ground state \cite{ligomirror,nm,nanomirror} 
and the measurement of quantum entanglement of a two-mode system \cite{evalue1,evalue2,evalue3}, we study two nearby cooled masses interacting gravitationally.

We propose two scenarios capable of increasing gravitational entanglement between the masses.
In the first scenario, we consider the masses trapped at all times in 1D harmonic potentials (optomechanics). In the second one, the masses are released from the optical traps.
For both settings, we derive an analytic figure of merit characterising the amount of gravitationally induced entanglement and the time it takes to observe it.
The derivation includes various initial states and shows that the objects have to be cooled down very close to their ground states and that squeezing of their initial state significantly enhances the amount of generated entanglement.
We then formulate a numerical approach, which accounts for all the relevant sources of noise affecting the settings that we propose, to identify a set of parameters required for the observation of such entanglement. 
Finally, we discuss the conclusions that can be drawn from this experiment with emphasis on the need for independent verification that the gravitational interaction between nearby objects is indeed mediated.

\section{Proposed experimental setups}
Consider two particles, separated by a distance $L$, as depicted in Fig. \ref{FIG_gr_setup}. 
In what follows, we study the setting where the massive particles are either held or released from unidimensional harmonic traps.
In the former case one can treat the particles as identical harmonic oscillators, with the same shape, mass $m$, and vibrational frequency $\omega$. 
The two oscillators and the gravitational interaction between them give rise to the total Hamiltonian $H=H_0+H_{\text g}$, where
\begin{equation}\label{EQ_hamiltonian}
H_0=\frac{p_A^2}{2m}+\frac{1}{2}m\omega^2 x_A^2+\frac{p_B^2}{2m}+\frac{1}{2}m\omega^2 x_B^2
\end{equation}
and $H_{\text g}$ describes the gravitational term.
If the harmonic traps are removed the corresponding Hamiltonian simplifies to $H_0=(p_A^2+p_B^2)/2m$.
Before we proceed with detailed calculations, we shall discuss generic features of the gravitational term and the conditions required for the creation of entanglement.

\begin{figure}[h]
\centering
\includegraphics[scale=0.3]{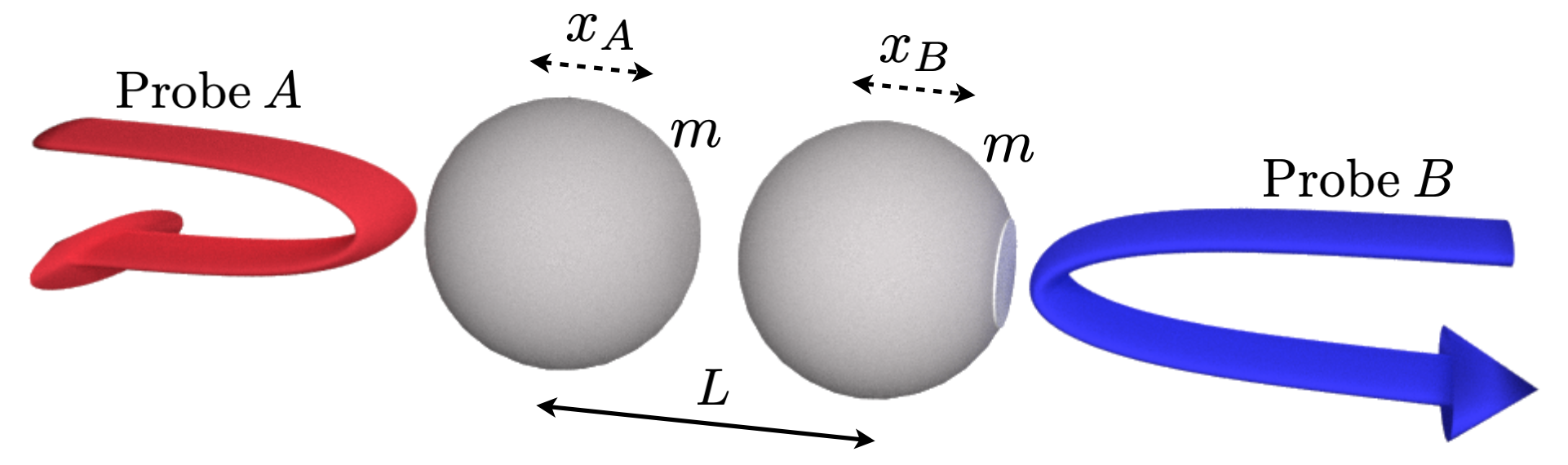}
\caption{Thought experiment for the generation of entanglement via gravitational interactions.
Two particles are a distance $L$ apart and trapped with unidimensional harmonic potentials.
It is assumed that both particles are cooled down near their ground state. 
We consider two scenarios in which the particles are either trapped at all times or released from the traps. 
In the text, we investigate entanglement dynamics for both scenarios.
Note that the generated entanglement can be probed with weak lights. 
Apart from the dominant gravitational coupling, our analysis takes noise, damping, decoherence, and Casimir forces into account. 
}
\label{FIG_gr_setup}
\end{figure}

In general, the gravitational term $H_{\text g}$ depends on the geometry of the objects. 
Various configurations have been analysed (will be presented later in Section~\ref{SC_osm}). 
The results of such analysis suggest that spherical masses give rise to the highest amount of generated entanglement.
The Newtonian gravitational energy of this setting is the same as if the two objects were point-like masses, that is 
$H_{\text g}=-Gm^2/ (L + x_B - x_A)$, where $L$ is the distance between the objects at equilibrium and $x_A$ ($x_B$) is the displacement of mass $A$ ($B$) from equilibrium.
By expanding the energy in the limit $x_A-x_B \ll L$, which is well justified for oscillators that are cooled down close to their ground state, one gets
\begin{equation}\label{EQ_hspheresexp}
H_{\text g}=-\frac{Gm^2}{L}\left(1+\frac{(x_A - x_B)}{L}+\frac{(x_A - x_B)^2}{L^2}+\cdots\right).
\end{equation}
The first term is a rigid energy offset, while the second is a bi-local term and cannot thus give rise to quantum entanglement. 
The third term, which is proportional to $(x_A-x_B)^2$, is the first that couples the masses. 
When written in second quantisation, it becomes apparent that this term includes contributions responsible for the correlated creation of excitations in both oscillators. 
In the quantum optics language, this is commonly referred to as a ``two-mode squeezing" operation, which can in principle  entangle the masses provided a sufficient strength of their mutual coupling.

\section{Dynamics of gravitationally induced entanglement: Oscillators}

We first provide an intuitive argument setting the scales of experimentally relevant parameters, which will then be proven rigorously.
In order to achieve considerable entanglement, we should ensure that the coupling (third term) in Eq.~(\ref{EQ_hspheresexp}) is comparable to the energy $\hbar \omega$ of each oscillator, that is $Gm^2(x_A - x_B)^2/L^3 \sim \hbar \omega$.
As we assume that the oscillators are near their ground state, we estimate their displacements by the ground state extension, $(x_A-x_B)^2 \sim 2 \hbar / m \omega$. 
We thus introduce the (dimensionless) figure of merit
\begin{equation}
\eta_{\text g} \equiv \frac{2 G m}{\omega^2L^3}.
\label{EQ_OSC_MERIT}
\end{equation}
We should have $\eta_{\text g}\sim1$ in order for the oscillators to be significantly entangled. 
This sets the requested values of the experimentally relevant parameters $m$, $\omega$, and $L$.

Below we will demonstrate the following results, which embody the key findings of our investigation:
(i) Starting from the ground state of each oscillator and assuming (for the sake of argument) only negligible environmental noise, the maximum entanglement (as quantified by the logarithmic negativity~\cite{negativity,adesso2004extremal}) generated during the dynamics is given by $E_{\mbox{\scriptsize th}}^{\max}\approx \eta_{\text g}/\ln{2}$.\footnote{Note that in this chapter, we shall use $E$ to denote the logarithmic negativity for simplicity.} 
Moreover, the time taken for entanglement to reach such maximum value is $t_{\mbox{\scriptsize th}}^{\max}=\pi/2(1-\eta_{\text g})\omega$;
(ii) Single-mode squeezing of the initial ground state of each oscillator substantially enhances the gravity-induced entanglement. 
The corresponding maximum entanglement becomes $E_{\mbox{\scriptsize sq}}^{\max}\approx |s_A+s_B|/\ln{2}$, 
where $s_j~(j=A,B)$ is the degree of squeezing of the $j$th oscillator, and we assume $\eta_{\text g} \ll s_A, s_B$.
In this case, the maximum entanglement is reached in a time $t_{\mbox{\scriptsize sq}}^{\max}=\pi/2\eta_{\text g}\omega$;
(iii) Weaker entanglement is generated with increasing temperature of the masses or coupling to the environment.

We will begin by constructing Langevin equations of the masses in Heisenberg picture. 
As we will be dealing with Gaussian states, we provide the dynamics of the covariance matrix, which completely describes the state of the whole system. 
From this dynamics, we will calculate important quantities such as quantum entanglement between the masses, both analytically for special cases and numerically for general cases.

\subsection{Langevin equations}\label{SC_leosc}

As the third term in Eq.~(\ref{EQ_hspheresexp}) is already very small under usual experimental conditions,\footnote{Note that the ratio between any two consecutive terms in Eq.~(\ref{EQ_hspheresexp}) is given by $(x_A - x_B) / L \sim \sqrt{\hbar / m \omega L^2}$. 
For instance, taking $m=100$ $\mu$g, $\omega=100$ kHz, and $L=0.1\:$ mm gives this ratio $\sim 10^{-12}$, and for macroscopic values $m = 1$ kg, $\omega = 0.1$ Hz, and $L = 1$ cm the ratio is $\sim10^{-15}$.}
we neglect all terms of order higher than the second in the displacement from equilibrium. 
We note Ref. \cite{linearp} for similar treatment of linearised central-potential interactions.
By taking the total Hamiltonian with a suitably truncated gravitational term $H_{\text g}$, one gets a set of Langevin equations in Heisenberg picture
\begin{equation}\label{EQ_langevins}
\begin{aligned}
\dot { \bm{X}}_j &=\omega \:  \bm{P}_j\qquad (j=A,B),\\ 
\dot { \bm{P}}_A&=-\omega\left(1-\eta_{\text g}\right) \bm{X}_A-\omega\eta_{\text g}\, \bm{X}_B-\gamma \,  \bm{P}_A+\xi_A+ \nu,\\
\dot { \bm{P}}_B&=-\omega\left(1-\eta_{\text g}\right) \bm{X}_B-\omega\eta_{\text g}\, \bm{X}_A-\gamma \,  \bm{P}_B+\xi_B-\nu,\\
\end{aligned}
\end{equation}
where we have introduced the constant frequency $\nu={Gm^2}/\sqrt{\hbar m\omega L^4}$ and the dimensionless quadratures $\bm{X}_j=\sqrt{m\omega/\hbar}\:x_j$ and $\bm{P}_j=p_j/\sqrt{\hbar m\omega}$. 
These equations incorporate Brownian-like noise -- described by the noise operators $\xi_j$ -- and damping (at rate $\gamma$) affecting the dynamics of the mechanical oscillators, due to their interactions with their respective environment.
We assume the (high mechanical quality) conditions $\mathcal{Q}=\omega/\gamma \gg 1$, as it is the case experimentally, so that the Brownian noise operators can {\it de facto} be treated as uncolored noise and we can write $\langle \xi_j(t)\xi_j(t^{\prime})+\xi_j(t^{\prime})\xi_j(t)\rangle/2 \simeq \gamma(2 \bar n+1)\delta(t-t^{\prime})$ for $j=A,B$~\cite{benguria1981quantum,giovannetti2001phase}. Here, $\bar n=(e^\beta-1)^{-1}$ is the thermal phonon number with $\beta=\hbar \omega/k_B T$ and $T$ the temperature of the environment with which the oscillators are in contact.  

The linearity of Eqs.~(\ref{EQ_langevins}) and the Gaussian nature of the noise make the theory of continuous variable Gaussian systems very well suited to the description of the dynamics and properties of the oscillators under scrutiny. In this respect, the key tool to use is embodied by the covariance matrix $V(t)$ associated with the state of the system, whose elements $V_{ij}(t) = \langle u_i(t)u_j(t)+u_j(t)u_i(t)\rangle/2-\langle u_i(t)\rangle \langle u_j(t)\rangle$ encompass the variances and correlations of the elements of the quadrature vector $u(t)=(\bm{X}_A(t), \bm{P}_A(t), \bm{X}_B(t), \bm{P}_B(t))^T$.
The temporal behaviour of physically relevant quantities for our system of mechanical oscillators can be drawn from $V(t)$ by making use of the approach for the solution of the dynamics that is illustrated below.

\subsection{Dynamics of covariance matrix}\label{SC_gr_cmrm}

In this section we provide the solution of the Langevin equations, and consequently the covariance matrix. 
One can rewrite the equations in (\ref{EQ_langevins}) as a single matrix equation $\dot u(t)=Ku(t)+l(t)$, with the vector $u(t)=(\bm{X}_A(t), \bm{P}_A(t), \bm{X}_B(t), \bm{P}_B(t))^T$ and a drift matrix
\begin{equation}\label{EQ_drift}
K=\left( \begin{array}{cccc} 
0&\omega&0&0\\ \eta_{\text g}
-\omega(1-\eta_{\text g})&-\gamma&-\omega\eta_{\text g}&0\\
0&0&0&\omega\\
-\omega\eta_{\text g}&0&-\omega(1-\eta_{\text g})&-\gamma
\end{array}\right).
\end{equation}
We split the last term in the matrix equation into two parts, representing the noise and constant term respectively, i.e., $l(t)=\upsilon(t)+\kappa_{\nu}$, where $\upsilon(t)=(0,\xi_A(t),0,\xi_B(t))^T$ and the constant $\kappa_{\nu}=\nu(0,1,0,-1)^T$ with $\nu={Gm^2}/\sqrt{\hbar m\omega L^4}$. 

The solution to the Langevin equations is given by 
\begin{eqnarray}\label{AEQ_gr_Lsol}
u(t)&=&W_+(t)u(0)+W_+(t)\int_0^t dt^{\prime}  W_-(t^{\prime})l(t^{\prime}),
\end{eqnarray}
where $W_{\pm}(t)=\exp{(\pm Kt)}$. 
This allows one to calculate the expectation value of the $i$th quadrature $\langle u_i(t)\rangle$ numerically, which is given by the $i$th element of 
\begin{eqnarray}\label{AEQ_quad}
&&W_+(t)\langle u(0)\rangle+W_+(t)\int_0^t dt^{\prime}  W_-(t^{\prime})\kappa_{\nu},
\end{eqnarray}
where we have used the fact that the noises have zero mean, i.e., $\langle \upsilon_i(t) \rangle=0$ and that $\langle \kappa_{\nu} \rangle=\mbox{tr}(\kappa_{\nu} \rho)=\kappa_{\nu}$.
From Eq.~(\ref{AEQ_gr_Lsol}), one can also calculate other important quantities via the covariance matrix as shown below.

Covariance matrix of our system is defined as $V_{ij}(t)\equiv \langle \{ \Delta u_i(t),\Delta u_j(t)\}\rangle/2=\langle u_i(t)u_j(t)+u_j(t)u_i(t)\rangle/2-\langle u_i(t)\rangle \langle u_j(t)\rangle$ where we have used $\Delta u_i(t)=u_i(t)-\langle u_i(t)\rangle$. 
This means that $\kappa_{\nu}$ does not contribute to $\Delta u_i(t)$ (and hence the covariance matrix) since $\langle \kappa_{\nu} \rangle=\kappa_{\nu}$. 
We can then construct the covariance matrix at time $t$ from Eq. (\ref{AEQ_gr_Lsol}) without considering $\kappa_{\nu}$ as follows
\begin{eqnarray}
V_{ij}(t)&=&\langle u_i(t)u_j(t)+u_j(t)u_i(t)\rangle/2-\langle u_i(t)\rangle \langle u_j(t)\rangle \nonumber \\
\label{EQ_gr_l3}
V(t)&=&W_+(t)V(0)W_+^T(t) \nonumber \\
&&+W_+(t)\int_0^t dt^{\prime} W_-(t^{\prime})DW_-^T(t^{\prime}) \:W_+^T(t) ,
\end{eqnarray}
where $D=\mbox{Diag}[0,\gamma(2\bar n+1),0,\gamma(2\bar n+1)]$ and we have assumed that the initial quadratures are not correlated with the noise quadratures such that the mean values of the cross terms are zero. 
A more explicit solution of the covariance matrix, after integration in Eq. (\ref{EQ_gr_l3}), is given by 
\begin{eqnarray}\label{EQ_gr_Ct}
KV(t)+V(t)K^T&=&-D+KW_+(t)V(0)W_+^T(t) \nonumber \\
&&+W_+(t)V(0)W_+^T(t)K^T \nonumber \\
&&+W_+(t)DW_+^T(t),
\end{eqnarray}
which is linear and can be solved numerically. 

Consider a special case, in which the damping term $\gamma$ is negligible, giving $D=\bm{0}$.
In this case, Eq. (\ref{EQ_gr_Ct}) simplifies to 
\begin{equation}\label{EQ_nd}
V(t)=W_+(t)V(0)W_+^T(t).
\end{equation}
In this regime we will obtain analytical results of Section~\ref{SC_gr_oscana}.

\subsubsection{Entanglement from covariance matrix}
We use logarithmic negativity to quantify the amount of entanglement between the coupled masses.
One can calculate this from the covariance matrix by using the method in Chapter~\ref{Chapter1}, which will be briefly repeated here for convenience.
The covariance matrix $V(t)$ describing our two-mode system can be written in a block form 
\begin{equation}
V(t)=\left( \begin{array}{cc} I_{A}&L\\ L^T&I_{B} \end{array}\right),
\end{equation}
where the component $I_{A}$ ($I_{B}$) is a $2\times 2$ matrix describing local mode correlation for $A$ ($B$) while $L$ is a $2\times 2$ matrix characterising the intermodal correlation. 
A two-mode covariance matrix has two symplectic eigenvalues $\{\nu_1,\nu_2\}$.
A physical system has $\nu_1,\nu_2\ge 1/2$ \cite{weedbrook2012gaussian}. 

For entangled modes, the covariance matrix will not be physical after partial transposition with respect to mode $B$ (this is equivalent to flipping the sign of the oscillator's momentum operator $\bm{P}_B$ in $V(t)$). 
This unphysical $V(t)^{T_B}$ is shown by the minimum symplectic eigenvalue $\tilde \nu_{\min}<1/2$. 
The explicit expression is given by $\tilde \nu_{\min}=(\Sigma-\sqrt{\Sigma^2-4\:\mbox{det}V} )^{1/2}/\sqrt{2}$, 
where $\Sigma=\mbox{det}I_A+\mbox{det}I_B-2\:\mbox{det}L$.
Entanglement between mode $A$ and mode $B$ is then quantified by logarithmic negativity as $E=\max \big \{0,-\log_2{(2\tilde \nu_{\min})}\big \}$~\cite{negativity, adesso2004extremal}. 
Note that the separability condition, when $V(t)^{T_B}$ has $\tilde \nu_{\min}\ge1/2$, is sufficient and necessary for two-mode systems \cite{werner2001bound}. \\

\subsection{Noiseless dynamics: Analytical solution}\label{SC_gr_oscana}

\subsubsection{Thermal initial state}
Due to weakness of the gravitational coupling, we have $\eta_{\text g} \ll 1$ in practically any realistic experimental situation, and we thus assume such condition here.\footnote{We use $\eta_{\text g}\ll1$ for all analytical derivations in this chapter.}
In the case of no damping (i.e., $\gamma = 0$) and assuming an initial (uncorrelated) thermal state of the oscillators, 
a tedious but otherwise straightforward analytical derivation shows that the entanglement between the mechanical systems, as quantified by the logarithmic negativity, oscillates in time with an amplitude of $\eta_{\text g}/\ln{2}-\log_2(2\bar n+1)$. At low operating temperature, a condition achieved through a combination of passive and radiation-pressure cooling~\cite{RMP.86.1391}, $\bar n\approx 0$ and the maximum entanglement between the oscillators is $E_{\mbox{\scriptsize th}}^{\max}\approx \eta_{\text g}/\ln{2}$, a value reached at a time $t^{\max}_{\mbox{\scriptsize th}}$ such that $\omega t^{\max}_{\mbox{\scriptsize th}}=\pi/2(1-\eta_{\text g})$. 
We present the dynamics of entanglement for varying values of $\eta_{\text g}$ in Fig.~\ref{FIG_eground}.

\begin{figure}[h!]
\centering
\includegraphics[scale=0.45]{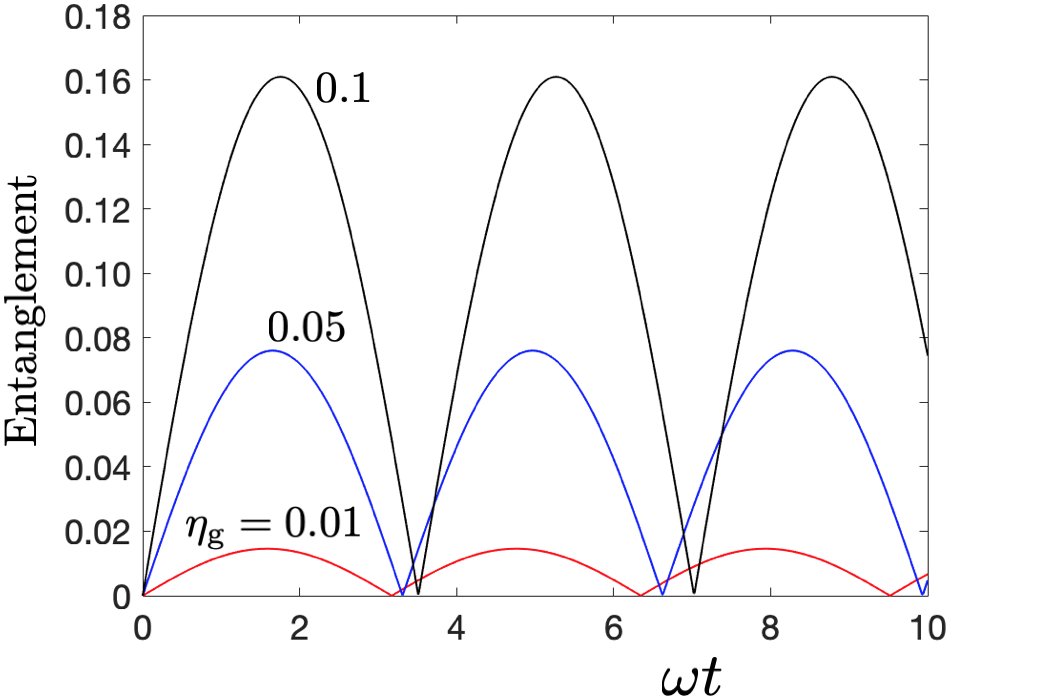}
\caption{Entanglement dynamics with initial ground state for each mass.
The figure of merit is varied as $\eta_{\text g}=0.01$ (red curve), $0.05$ (blue curve), and $0.1$ (black curve).
The maximum entanglement approximately follows $\eta_{\text g}/\ln{2}$.
}
\label{FIG_eground}
\end{figure}

\subsubsection{Squeezed thermal initial state}

An analytic solution is also possible for the case of mechanical systems initially prepared in squeezed thermal states, a situation that can be arranged by suitable optical driving~\cite{vanner2013cooling,rashid2016experimental}.
Each mass is prepared in a state $S\rho_{\text {th}}S^{\dagger}$, where $\rho_{\text {th}}$ is a thermal state and $S=\exp{(-i\:s(\bm{X}^2-\bm{P}^2)/2)}$ is the squeezing operator with strength $s$ (assumed to be real). 
This operator corresponds to anti-squeezing (squeezing) the position quadrature for $s>0$ ($s<0$).
By writing individual-oscillator squeezing as $s_j$ and assuming $s_j \gg \eta_{\text g}$, the entanglement is again observed to oscillate, but with amplitude $|s_A + s_B|/\ln{2}-\log_2(2\bar n+1)$. 
We present the entanglement dynamics in Fig.~\ref{FIG_sqsq} for varying values of the squeezing strength.
Note that it is irrelevant whether quadratures of both masses are squeezed or anti-squeezed.
We provide explanation in the Details of entanglement dynamics (Section~\ref{SC_doedyn}).
Therefore, only the degree of pre-available single-oscillator squeezing and the environmental temperature set a limit to the amount of entanglement that can be generated between the mechanical systems through the gravitational interaction.  
In the low temperature limit, where $E_{\mbox{\scriptsize sq}}^{\max}\approx |s_A+s_B|/\ln{2}$, which is in principle arbitrarily larger than the case without squeezing, 
a time $t^{\max}_{\mbox{\scriptsize sq}}=\pi/(2\eta_{\text g}\omega)\gg t^{\max}_{\mbox{\scriptsize th}}$ would be required for such entanglement to accumulate. 
Needless to say, long accumulation times are far from the  possibilities offered by state-of-the-art optomechamical experiments, which prompts an assessment that includes ab initio the effects of environmental interactions.

\begin{figure}[h!]
\centering
\includegraphics[scale=0.4]{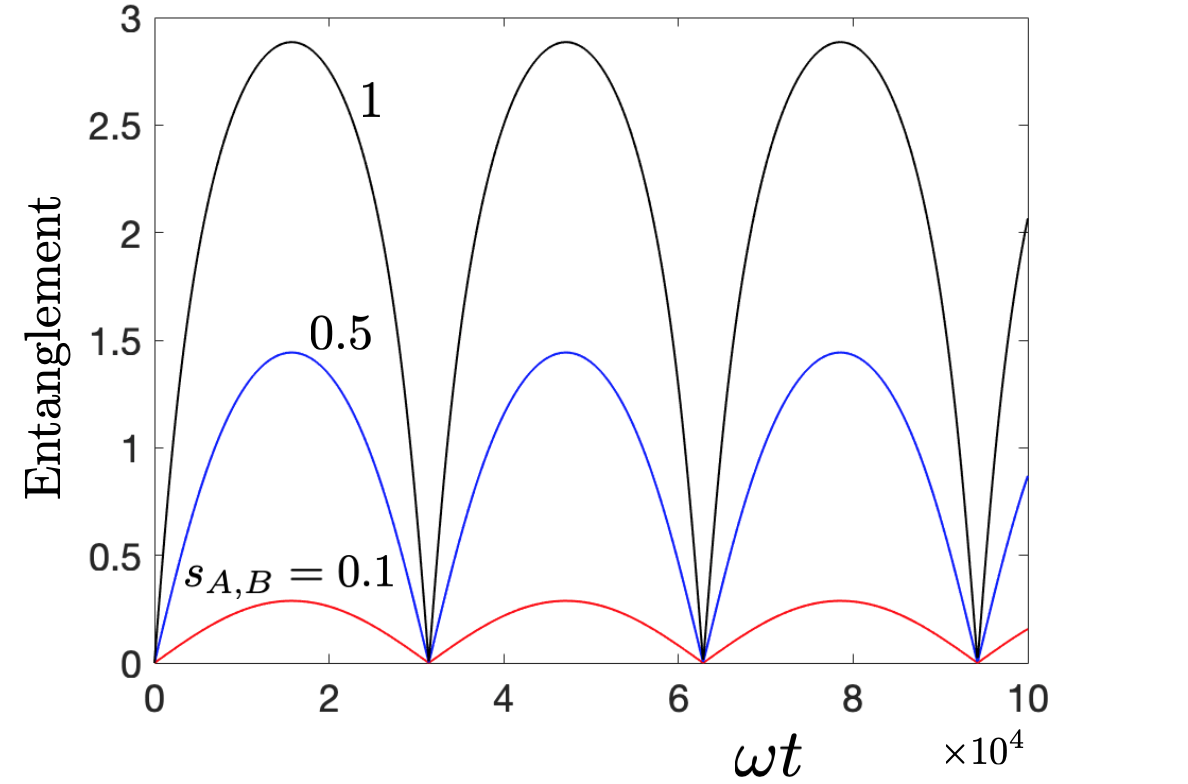}
\caption{Entanglement dynamics with squeezed initial ground state for each mass.
The figure of merit is taken as $\eta_{\text g}=10^{-4}$ and the squeezing strength is varied as $s_{A,B}=0.1$ (red curve), $0.5$ (blue curve), and $1$ (black curve).
The maximum entanglement approximately follows $2s_{A,B}/\ln{2}$.
}
\label{FIG_sqsq}
\end{figure}

\subsection{Noisy dynamics: Numerical simulation}

In the case of noisy dynamics, however, an analytical solution is no longer available and we have to resort to a numerical analysis.
Let us therefore consider the figure of merit $\eta_{\text g}$ in order to set the parameters for numerical investigation.
We consider two oscillators of spherical shape with uniform density $\rho$ and radius $R$, which are separated by a distance $L = 2.1 R$. This might be a situation matching current experiments in levitated optomechanics~\cite{kiesel2013cavity,vovrosh2017parametric}, which are rapidly evolving towards the possibility of trapping multiple dielectric nano-spheres in common optical traps and controlling their relative positions~\cite{Kiesel}. 
However, low-frequency oscillators, which are favourable for the figure of merit and typically associated with large masses, are unsuited to such platforms and would require a different arrangement, such as LIGO-like ones~\cite{ligomirror}.

In terms of the density $\rho$, we have $\eta_{\text g}=8\pi G\rho/3(2.1)^3 \omega^2$, which does not depend on the dimensions of the oscillators nor their mass.
As the density of materials currently available for such experiments varies within a range of only two orders of magnitude, the linear dependence on $\rho$ sets a considerable restriction on the values that $\eta_{\text g}$ can take. 
The densest naturally available material is Osmium, which has $\rho=22.59$~g/$\mbox{cm}^3$ and, in order to provide an upper bound to the generated entanglement that would be attainable using other materials, we shall use this density in our numerical simulations.
Accordingly, $\eta_{\text g}=1.36 \times 10^{-6}/\omega^2$, where $\omega$ is in Hz.

\begin{figure}[h!]
\centering
\includegraphics[scale=0.55]{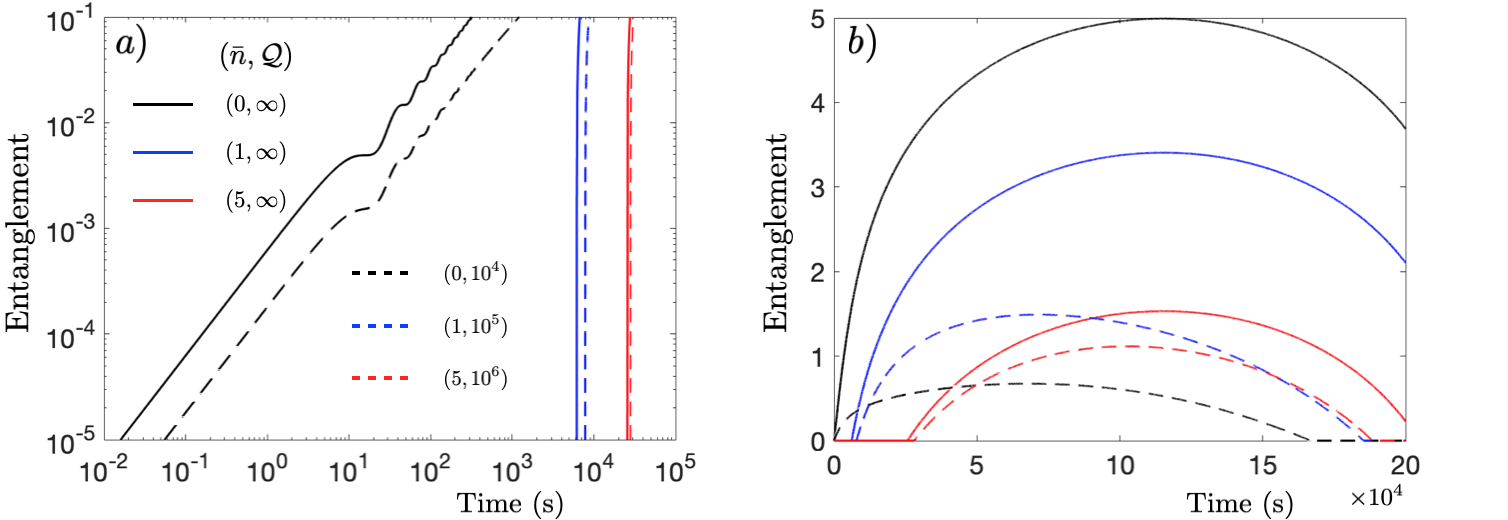}
\caption{Entanglement dynamics between two spherical oscillators. 
Dynamics for shorter times (a), showing entanglement resolution, and for longer times (b), showing high accumulation.
The varying parameters are the temperature (mean phonon number $\bar n$) and the damping of the oscillators (mechanical quality factor $\mathcal{Q}$). 
We have fixed both the frequency of the mirrors $\omega=0.1$~Hz and the amount of squeezing $s_{A,B}=1.73$.
Note that entanglement is quantified using the logarithmic negativity.
}
\label{FIG_Edyndam}
\end{figure}

Fig.~\ref{FIG_Edyndam} shows exemplary entanglement dynamics for different values of the thermal phonon number $\bar n$ and mechanical quality factor $\mathcal{Q}$.\footnote{The oscillations of entanglement for unsqueezed initial state are still present in this dynamics, showing repeating pattern with a period of $\pi/[(1-\eta_{\text g})\omega ] \approx 31$ s.}
The frequency has been fixed to $\omega=0.1$ Hz (cf. Section~\ref{SC_cwrexp}).
As expected, higher damping (lower $\mathcal{Q}$) results in the decay of entanglement, and the higher the temperature of the mirror (higher $\bar n$) the higher the mechanical quality factor needed to maintain entanglement. 
The setup allows for high entanglement, even with low coupling strength $\eta_{\text g} \sim 10^{-4}$.
However, this comes at the expense of the time for which the dynamics of the oscillators should be kept coherent.
In the same figure one can see the time required for accumulation of entanglement.
It shows that cooling down the masses close to their ground state, $\bar n\approx 0$, is crucial for the reduction of the required coherence time.

\subsection{Details of entanglement dynamics}\label{SC_doedyn}

In this section we show that entanglement gain is linked to the evolution of the position variance of each mass.
This is intuitive because bigger variance means stronger gravitational coupling for parts of the wave functions which are closer.
In order to illustrate this, we take, as an example, the oscillators setup with squeezed initial ground state for each mass.

\begin{figure}[h]
\centering
\includegraphics[scale=0.47]{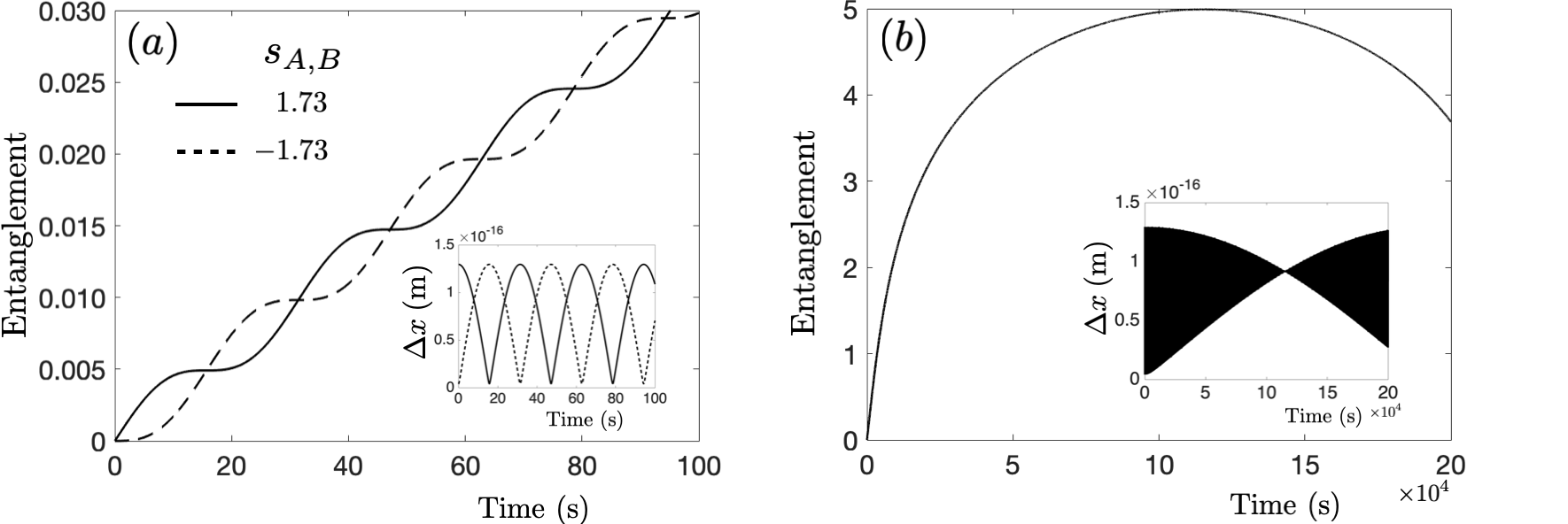}
\caption{Exemplary dynamics of entanglement and width of two coupled oscillators.
Each oscillator is made of Osmium and has a spherical shape with mass $m=1$~kg.
The equilibrium distance is $L=2.1R$ and the frequency of the trapping potential is $\omega=0.1$~Hz. 
(a) Short dynamics within $100$~s for squeezing parameters $s_{A,B}=\pm 1.73$.
(b) Higher gain of entanglement requires longer time. 
Note that the dynamics in (b) is approximately the same for positive and negative squeezing parameters. 
}
\label{FIG_sqtogether}
\end{figure}

Fig. \ref{FIG_sqtogether}a shows the entanglement dynamics with both $s_{A,B}=1.73$ and $-1.73$, corresponding to initial states of the masses that are both anti-squeezed or squeezed, respectively, in position quadrature.
During the evolution, the width of each mass oscillates as illustrated in the inset of Fig. \ref{FIG_sqtogether}a. 
It is clear that the oscillation of the width matches the oscillation of entanglement, for both squeezing parameters.
For longer time, i.e., $t\sim 10^4$~s, the oscillation in width leads to an accumulation of high entanglement as can be seen in Fig. \ref{FIG_sqtogether}b.
Note that the oscillation of the position variance is very rapid on this timescale and the envelope of these oscillations is the same for both positive and negative $s$, see the inset of Fig. \ref{FIG_sqtogether}b.
As a result, the entanglement dynamics is approximately equal for squeezed and anti-squeezed cases. 
This confirms our analytical result regarding the maximum entanglement gain being $|s_A + s_B|/\ln{2}$.

\section{Dynamics of gravitationally induced entanglement: Released masses}

As seen previously, the experimental parameters required for detectable gravitational entanglement of masses in harmonic traps are demanding.
We therefore study one more feasible system, where the traps are switched off after cooling the masses.
All solutions are analytical under the assumed weak gravitational coupling.

\subsection{Langevin equations and covariance matrix}
Similar to the treatment of two oscillators, one starts with the total Hamiltonian for free masses and truncated gravitational term, and obtains the following equations of motion
\begin{equation}\label{EQ_langevinsfm}
\begin{aligned}
\dot { \bm{X}}_j &=\omega \:  \bm{P}_j\qquad (j=A,B),\\ 
\dot { \bm{P}}_A&=\omega \eta_{\text g} \bm{X}_A-\omega\eta_{\text g}\, \bm{X}_B+ \nu,\\
\dot { \bm{P}}_B&=\omega \eta_{\text g} \bm{X}_B-\omega\eta_{\text g}\, \bm{X}_A-\nu.\\
\end{aligned}
\end{equation}
Note that $\omega$ here just sets the conversion between $x_j,p_j$ and their dimensionless parameters $\bm{X}_j,\bm{P}_j$. 
In what follows, we will consider starting the dynamics with thermal state for each mass.
For example, the ground state is a Gaussian state with width $\Delta x(0)=\sqrt{\hbar/2m\omega}$.
This way, one can think of $\omega$ as a parameter characterising the initial spread of the wave function.

For the solution of the dynamics, one can follow similar treatments as in Section~\ref{SC_gr_cmrm}, keeping in mind $\gamma=0$ and $\upsilon(t)=(0,0,0,0)^T$ such that the quadrature dynamics and covariance matrix is given by Eq. (\ref{AEQ_quad}) and Eq. (\ref{EQ_nd}) respectively with a new drift matrix
\begin{equation}\label{EQ_driftfm}
K=\left( \begin{array}{cccc} 
0&\omega&0&0\\ 
\omega \eta_{\text g}&0&-\omega\eta_{\text g}&0\\
0&0&0&\omega\\
-\omega\eta_{\text g}&0&\omega \eta_{\text g}&0
\end{array}\right).
\end{equation}

\subsection{Noiseless dynamics: Analytical solution}

One can obtain the covariance matrix $V(t)$ from Eqs.~(\ref{EQ_langevinsfm}) and consequently derive the entanglement dynamics using the approach described above. 
After imposing the limits $\eta_{\text g} \ll 1$ and $\sqrt{\eta_{\text g}}\: \omega t \ll 1$, which apply in typical experimental situations, 
one obtains the analytical expression for the entanglement dynamics as follows
\begin{eqnarray}\label{EQ_fment}
E_{\text {th}}(t)&=&\max \big \{0,E_{\text {gnd}}(t)-\log_2(2\bar n+1)\big \},\\
E_{\text {gnd}}(t)&=&-\log_2\Big(\sqrt{1+2\bm{\sigma}(t)-2\sqrt{\bm{\sigma}(t)^2+\bm{\sigma}(t)}}\Big),\nonumber
\end{eqnarray}
where $E_{\text {gnd}}(t)$ is the entanglement with initial ground state for each mass and $\bm{\sigma}(t)=4G^2m^2\omega^2t^6/9L^6$.
Since entanglement is an increasing function of $\bm{\sigma}(t)$, the latter is a figure of merit for entanglement gain relevant in the case of released masses. 
We present exemplary entanglement dynamics in Fig.~\ref{FIG_fsdynamics} for which entanglement $\sim10^{-2}$ is achieved within seconds. 
The parameters used here are $m=100~\mu$g, $\omega=100$~kHz, and $L=3R$.
We will show later that with these values gravity is the dominant interaction and coherence times are much longer than $1$~s.
Note that this setup does not require any squeezing.

\begin{figure}[h]
\centering
\includegraphics[scale=0.35]{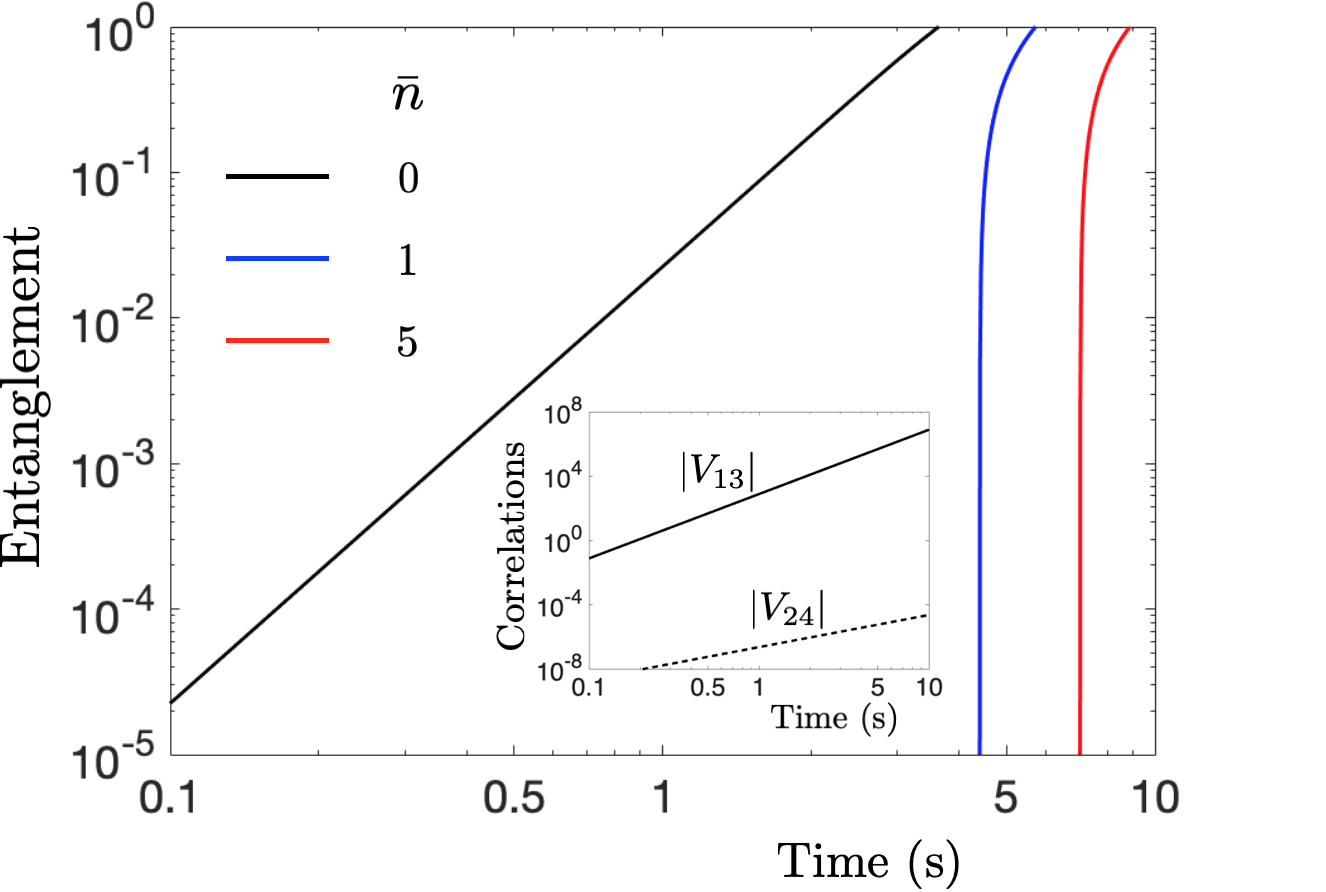}
\caption{Entanglement dynamics between two released spherical particles. 
Each particle with mass $100$~$\mu$g is initially trapped in 1D harmonic potential ($\omega=100$~kHz) and having phonon number $\bar n$. 
This corresponds to the initial state of each mass being Gaussian thermal state.
We take the equilibrium distance to be $L=3R\approx 0.1$~mm, making gravitational coupling more dominant than the Casimir interactions. 
We note that detectable entanglement, $E\sim 0.01$, is generated in seconds: $0.8$, $4.5$, and $7.5$~s for $\bar n=0$, $1$, and $5$ respectively.
It is shown in the text that this accumulation is within the coherence times of the particles. 
}
\label{FIG_fsdynamics}
\end{figure}

These improvements over the scheme with trapped masses are the result of unlimited expansion of the wave functions. 
For example, for initial ground state, the evolution of the width of each sphere closely follows $\Delta x(t)\approx \sqrt{\hbar/2m\omega} \sqrt{1+\omega^2 t^2}$, which is an exact solution to a free non-interacting mass.
The effect of gravity is stronger attraction of parts of the spatial superposition that are closer, and hence generation of position and momentum correlations, leading to the growth of quantum entanglement, see inset in Fig.~\ref{FIG_fsdynamics}.

In order to understand the effect of squeezing in this setup, let us suppose, for simplicity, the squeezing strengths $s_{A,B}=s$.
It is as if one initially prepared each mass in a Gaussian state with a new initial spread $\Delta x^{\prime}(0)=\Delta x(0)\exp(s)$. 
One can then calculate the entanglement dynamics using Eq. (\ref{EQ_fment}) with a new frequency $\omega^{\prime}=\omega \exp(-2s)$.
This means that anti-squeezing the initial position quadrature ($s>0$) would decrease entanglement gain, a situation opposite to the oscillators setup.
This is because a Gaussian state with smaller $\Delta x(0)$ spreads faster, such that during the majority of the evolution, the width is larger than that if one started with larger $\Delta x(0)$.
In principle, one obtains higher entanglement gain by squeezing the position quadrature ($s<0$).
However, this will result in higher final width, making it more susceptible to decoherence by environmental particles (see calculations in Section~\ref{SC_decdet}).

From numerical simulations, one confirms that, within $t=[0,10]$s, the displacements of the two masses follow $x_A-x_B \ll L$. 
Furthermore, the trajectories coincide for both quantum treatment with truncated gravitational energy and classical treatment with full $H_{\text g}$ (see Section~\ref{SC_qnct}).
This justifies the approximations used.
A summary of comparison between our two proposed setups can be seen in Table~\ref{TB_gravitypara}.

\begin{table}[h]
\begin{center}
\caption{Summary of parameters for the proposed experimental setups: Oscillators and released masses. Note that for the dynamics (shaded), we assume initial ground state for each mass.}\label{TB_gravitypara}
\smallskip
\begin{tabular}{|l|c|c|}
\hline
Parameters & Oscillators&Released masses\\
\hline
\hline
Mass $m_{A,B}$ & $1~\text{kg}$ &$100~\mu\text{g}$\\
\hline
Radius $R$ & $2.2$~cm &$0.1$~mm \\
\hline
Equilibrium distance $L$ & $4.6$~cm &$0.3$~mm \\
\hline
Trap frequency $\omega$ & $0.1$~Hz&$100$~kHz \\
\hline
Squeezing strength $s_{A,B}$ & $1.73$&$0$ \\
\hline
\cellcolor{trolleygrey}Average width $\Delta x$ &\cellcolor{trolleygrey} $8.3\times 10^{-17}$~m&\cellcolor{trolleygrey}$2.8\times 10^{-12}$~m \\
\hline
\cellcolor{trolleygrey}Entangling time $\tau_{\text e}$ &\cellcolor{trolleygrey} $31.7$~s &\cellcolor{trolleygrey}$0.8$~s \\
\hline
\end{tabular}
\end{center}
\end{table}

\subsection{Quantum and classical trajectories}\label{SC_qnct}

From the quadrature dynamics of Eq. (\ref{AEQ_quad}) above, one can calculate the expectation value of position for both masses. 
They are presented in Fig. \ref{FIG_traject}, where we have assumed initial conditions $\langle u(0)\rangle=(0,0,0,0)^T$.
On the other hand, without truncating the gravitational interaction, one can easily solve the classical dynamics and obtain the following equation for $x_t$:
\begin{equation}
t\sqrt{\frac{2Gm}{L}}=\sqrt{x_t(L-2x_t)}+\frac{L}{2\sqrt{2}}\left(\frac{\pi}{2}-\tan^{-1}(\theta(x_t))\right),
\end{equation}
where $\theta(x_t)=(L-4x_t)/\sqrt{8x_t(L-2x_t)}$ and $t$ is the time taken for the left mass to move a distance $x_t$. 
The trajectory of the right mass is simply $-x_t$.
One can confirm that the classical trajectories indeed coincide with the quantum ones.
This justifies the truncation of $H_{\text g}$ in our calculations. 
Note also that, within $10$~s, the displacement $(x_A-x_B)\sim 10^{-9}$~m, which is much smaller than the initial distance between the masses $L\approx 0.3$~mm.
This validates the use of the limit $x_A-x_B \ll L$.

\begin{figure}[!h]
\centering
\includegraphics[scale=0.25]{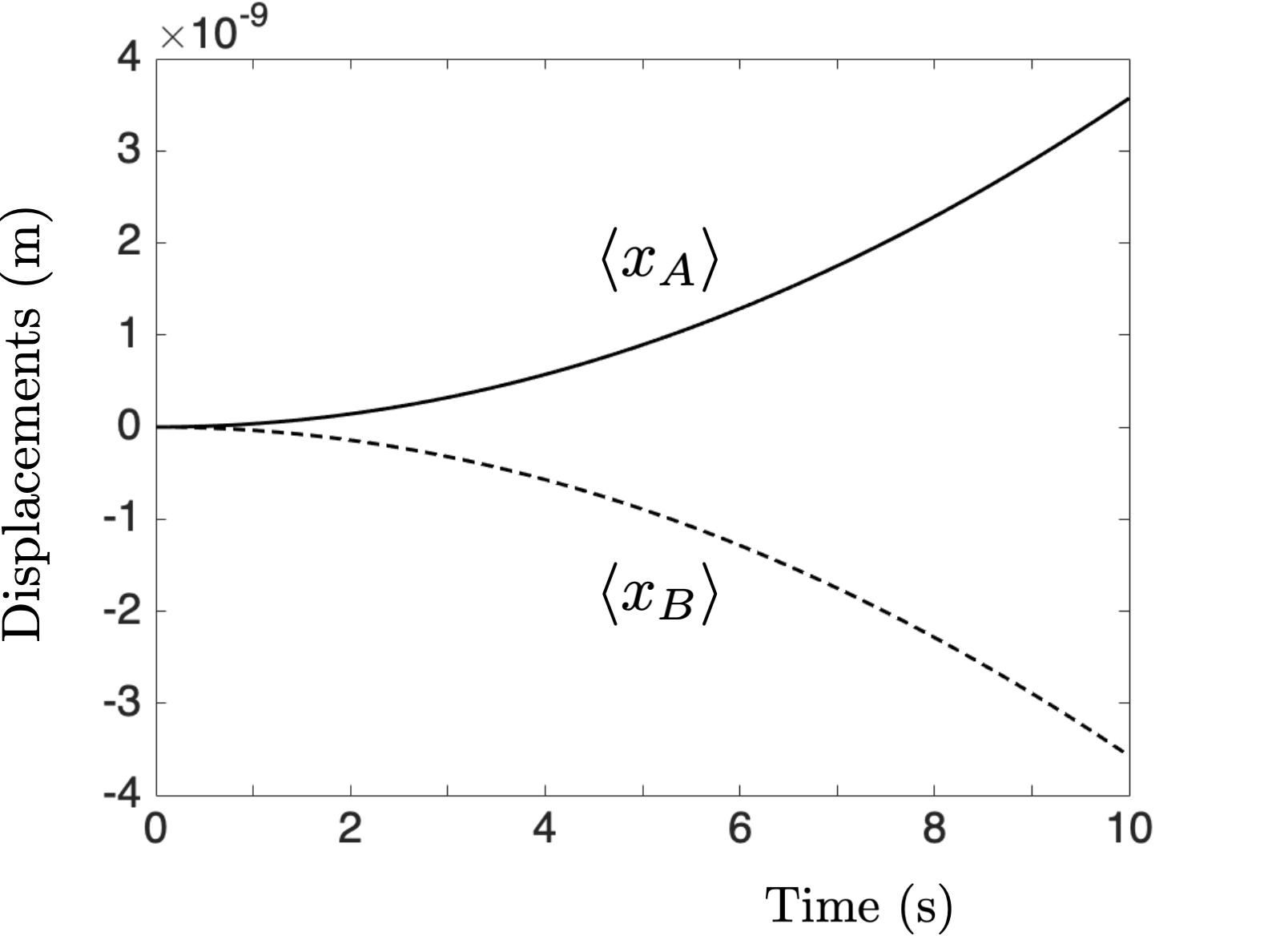}
\caption{Expectation value of the displacement of the masses.
Two Osmium spheres, each with mass $100$~$\mu$g, are separated by $L=3R$.
The dynamics shown is independent of the temperature and the frequency of the initial trapping potential $\omega$.
}
\label{FIG_traject}
\end{figure}

\section{Other shapes of masses}\label{SC_osm}

First, let us note, from Eqs. (\ref{EQ_langevins}) and (\ref{EQ_langevinsfm}), that the magnitude of the interaction rate between different modes (couples the momentum of one mass to the position of the other) in the case of identical spheres is given by $r_1= 2Gm/\omega L^3$. 
One can alternatively see this from an element of the drift matrix, i.e., $r_1=|K_{23}|$.
Assuming, e.g., the material is Osmium and $L=2.1R$, one gets $r_1=1.36\times 10^{-6}/\omega$.
Below, we will compare the interaction rate for different shapes of the objects to this reference.

Consider a particle of spherical shape, mass $m_A$, and frequency $\omega_A$ that interacts with a second sphere with mass $m_B$ and frequency $\omega_B$, see Fig. \ref{FIG_setupothers}a
(the same treatment applies to released masses with initial trap frequencies $\omega_A$ and $\omega_B$).
Moreover, we assume both spheres are made of Osmium and take $R_B=\alpha R_A$, where $\alpha=[0,\infty )$.
After taking similar steps as those in Section~\ref{SC_leosc} one obtains a new intermodal interaction rate $r_2= 1.36\times 10^{-6}f(\alpha)/\omega_A$, where $f(\alpha)=\alpha^{9/4}$ and we have used $L=2.1R_A$.
We have also assumed spring-like scaling for the frequency, i.e., $\omega_B=\omega_A\sqrt{m_A/m_B}$.
This result is intuitive as one expects stronger gravitational interaction by making the second sphere larger, i.e., larger $\alpha$.

More intriguing is the setting in Fig. \ref{FIG_setupothers}b, in which a rod with length $d$, mass $\lambda_B d$, and radius $R_B$ is interacting with a sphere of mass $m_A$ and radius $R_A$.
For simplicity, we assume both objects have a single mechanical frequency $\omega_A$ and $\omega_B$ respectively, and that the rod is thin such that its radius is much smaller than $L$. 
The gravitational interaction reads
\begin{equation}\label{EQ_hrodsphere}
H_{\text g}=-2Gm_A\lambda_B \ln{\left( \frac{d/2+\sqrt{(d/2)^2+(L^{\prime})^2}}{L^{\prime}}\right)},
\end{equation}
where $L^{\prime}=L-(x_A - x_B)$ with $x_A$ ($x_B$) being the displacement of mass $A$ ($B$). 
Note that this expression is the same as if we had a one-dimensional rod and a point mass.
By expanding Eq.~(\ref{EQ_hrodsphere}) in the limit $x_A - x_B \ll L$ and keeping only up to the quadratic term in displacements, one can show that the intermodal interaction rate is given by $r_3=2.18\times 10^{-7}f(\varsigma)/\omega_A$, where 
\begin{equation}
f(\varsigma)= (\varsigma)^{1/4}\left( 1-\frac{\varsigma^2((\varsigma^2-1)\sqrt{1+\varsigma^2}-1)}{(1+\sqrt{1+\varsigma^2})^2(1+\varsigma^2)^{3/2}}\right),
\end{equation}
and we define $\varsigma=2L/d$.
We have also taken $L=1.1R_A$, $R_B=0.1R_A$, and the same spring-like scaling as in the two-sphere configuration.
Although the strength of the gravitational energy (\ref{EQ_hrodsphere}) increases monotonically with $d$, it is not the case for $f(\varsigma)$, and hence $r_3$, 
which peaks at $\varsigma \approx 1.14$, i.e., $d\approx 1.75 L\approx 1.93R_A$.
The maximum $f_{\mbox{max}}=1.07$ gives maximum interaction rate $2.33\times 10^{-7}/\omega_A$.
For higher $d$ the rate $r_3$ decreases, implying that $r_3$ is always weaker than $r_1$.

We also note that both gravitational field and field gradient are necessary (but not sufficient) for producing entanglement.
This is clear from consideration of yet another configuration -- an infinite plane and a point mass separated by a distance $L$. 
One immediately observes that there is no field gradient here and that the gravitational energy of this configuration is proportional to $L-(x_A - x_B)$, which does not couple the masses and therefore does not contribute to entanglement.

\begin{figure}[h]
\centering
\includegraphics[width=0.7\textwidth]{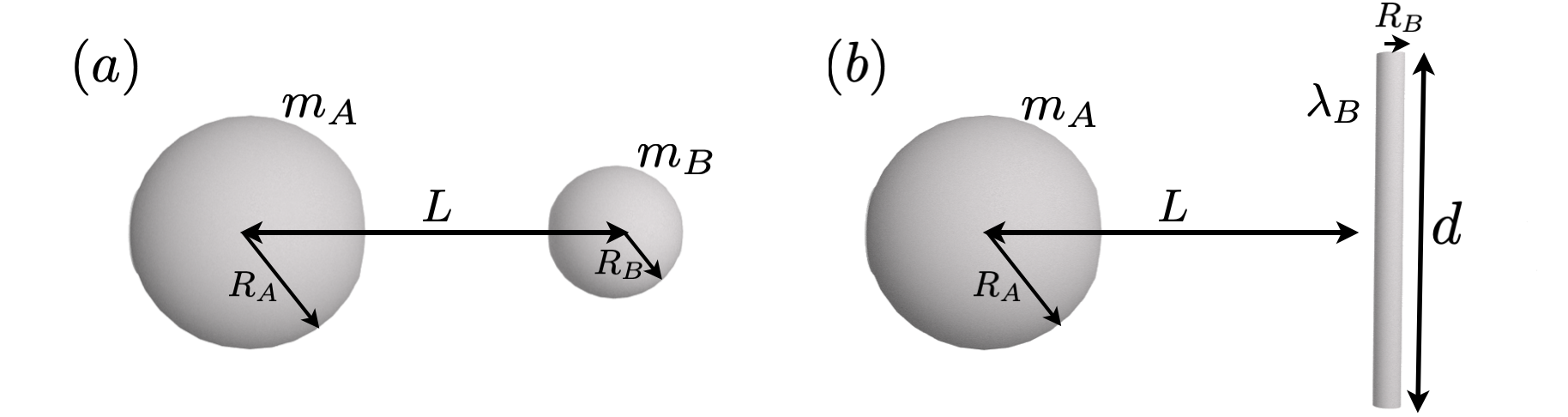}
\caption{Other configurations and notation for the two interacting masses. A mass with spherical shape is situated next to another spherical mass with a different radius (a) or a thin rod (b).
}
\label{FIG_setupothers}
\end{figure}

\section{Casimir interactions}

In light of their proximity, apart from gravitational interaction, the two masses can also interact via Casimir force.
It has been shown that the Casimir energy between two nearby spheres is given by a fraction of the ``proximity force approximation" $\mathcal{E}_{\text{c}}=-f_0 (\pi^3/1440)\hbar c R/(L - 2R + x_B - x_A)^2$, with the factor $0\le f_0\le1$ \cite{casimir1,casimir2}.
As typically $x_A - x_B \ll L - 2R$, we can expand such expression to find a quadratic term in $x_A - x_B$ that can produce entanglement between the masses.
The strength of this term, however, is much weaker than the strength of the corresponding entangling term of gravitational origin:
for Osmium oscillators with mass $\sim 1$~kg separated by $L=2.1R$, the ratio between the Casimir and gravitational term is $r_{\text{cg}}\sim 10^{-12}$.
Similar calculations made for released masses of the same material with $m=100~\mu$g and $L=3R$, give $r_{\text{cg}}\sim 10^{-2}$.
It is thus legitimate to ignore Casimir interaction in both schemes.

\section{Standard decoherence}\label{SC_decdet}

Let us also discuss common decoherence mechanisms, i.e., due to interactions with thermal photons and air molecules \cite{decoherence}.
All the situations we consider follow the limit $\Delta x \ll \lambda$, i.e., the ``size'' of superposition is much smaller than the wavelength of the particles causing the decoherence. 
In this regime, the coherence time due to interactions with thermal photons is given by $\tau_{\text {ph}}=1/\bm{\Lambda}_{\text {ph}} (\Delta x)^2$ with
\begin{equation}
\bm{\Lambda}_{\text {ph}}=10^{36} R^6\:T^9\: [1/\text m^2 \text s],
\end{equation}
where $R$ is the radius of the sphere and $T$ is the temperature of the environment. 
Note that all variables are in SI units.

The decoherence due to interactions with other scattering particles, e.g., air molecules, gives
\begin{equation}
\bm{\Lambda}_{\text {am}}=\frac{8}{3\hbar^2}\frac{N}{V}\sqrt{2\pi m_{\text {air}}} \:R^2 (k_{\text B}T)^{3/2},
\end{equation}
where $N/V$ is the density of air molecules with mass $m_{\text {air}}$.
We take $m_{\text {air}}\approx 0.5\times 10^{-25}$~kg.
For ultrahigh vacuum, pressure $\sim 10^{-10}$~Pa, the density is $\sim 10^{12}$ particles/m$^3$.
One could also consider performing these experiments in space.
In this case, by taking the pressure $\sim 10^{-15}$~Pa, the density can be as low as $\sim 10^7$~particles/m$^3$.

We take the average width of the wave function as an estimate for the superposition that is subjected to decoherence.
For oscillators made of Osmium, we use $m\sim 1$ kg and frequency $\omega \sim 0.1$~Hz.
Taking $L=2.1R$ and starting with ground state give $\Delta x \approx 8 \times 10^{-17}$m.
From interactions with thermal photons at environmental temperature of $4$~K (liquid Helium), the coherence time for the oscillators is $\tau_{\text {ph}}\sim 5$s.
The coherence time due to collisions with air molecules can be improved by evacuating the chamber with the oscillators -- for about $10^{12}$ molecules/$\mbox{m}^3$ (ultrahigh vacuum), the coherence time is also $\tau_{\text {am}} \sim 5$s.
For the space experiment, by taking the temperature as $2.7$~K (cosmic microwave background) and assuming $10^{7}$ molecules/$\mbox{m}^3$, we obtain $\tau_{\text {ph}}\sim 170$s and $\tau_{\text {am}} \sim 10^6$s.

By making similar calculations for released masses, with parameters considered in Fig.~\ref{FIG_fsdynamics}, one obtains $\tau_{\text {ph}}\sim 10^5$s and $\tau_{\text {am}} \sim 10^{-4}$s for the experiment on Earth with liquid Helium temperature and ultrahigh vacuum.
For the space experiment the coherence times improve to $\tau_{\text {ph}}\sim 10^7$s and $\tau_{\text {am}} \sim 41$s respectively.
We present in Table~\ref{TB_gravitydecoherence} a summary of coherence times.

\begin{table}[h]
\begin{center}
\caption{Coherence times for oscillators ($m=1$~kg, $\omega=0.1$~Hz) and released masses ($m=100$~$\mu$g, $\omega=100$~kHz) both for experiments on Earth and in space.
We use the pressure $P\sim 10^{-10}$~Pa and the temperature $T=4$~K for Earth experiment, while $P\sim 10^{-15}$~Pa and $T=2.7$~K for space experiment.
}
\label{TB_gravitydecoherence}
\smallskip
\begin{tabular}{|c|c|c|c|c|}
\hline
Coherence& \multicolumn{2}{c|}{Oscillators}&\multicolumn{2}{c|}{Released masses}\\
\cline{2-5}
time&Earth&Space&Earth&Space\\
\hline
\hline
$\tau_{\text {ph}}$ &$5$~s& $170$~s &$4\times 10^5$~s& $2\times 10^7$~s \\
\hline
$\tau_{\text {am}}$ &$5$~s& $10^6$~s &$2\times 10^{-4}$~s& $41$~s \\
\hline

\end{tabular}
\end{center}
\end{table}

Other schemes have been proposed for gravitationally induced entanglement~\cite{gravity1,gravity2}.
They are based on Newtonian interaction between spatially superposed microspheres with embedded spins.
\textcolor{black}{In those proposals, entanglement is reached faster and the small size of the experiment is the main advantage.
However, in order} to separate gravitational and Casimir contributions in that setup, each diamond sphere with mass $m=10^{-14}$~kg has to be superposed across 250~$\mu$m.
Decoherence due to scattering of molecules then becomes the main limiting factor.
The schemes we discussed here are complementary in a sense that vibrations of each oscillator are minute (no macroscopic superposition) but larger mass, $100$~$\mu$g, is required for observable entanglement.

\section{Comparison with recent experimental achievements}\label{SC_cwrexp}

We have shown that two nearby masses -- both trapped and released -- can become entangled via gravitational interaction.
Let us now discuss the experimental conditions required to observe this entanglement in light of recent experimental achievements.

Logarithmic negativity in the order $10^{-2}$ has already been observed experimentally between mechanical motion and microwave cavity field~\cite{evalue1}. 
Extrapolating the same entanglement resolution to the case of two massive oscillators sets the required frequency to $\omega \sim 10^{-2}$~Hz, see Eq. (\ref{EQ_OSC_MERIT}) and its expression in terms of $\omega$.
Interestingly, kilogram-scale mirrors of similar frequency ($\omega \sim 10^{-1}$~Hz) were recently cooled down near their quantum ground state~\cite{ligomirror}.
Furthermore, recent experiments on squeezed light have reported high squeezing strength \cite{sqpara1,sqpara2} (see also a review in this context \cite{sqpara3}), up to 15 dB, which corresponds to $s \approx 1.73$.
Advances in the state transfer between light and optomechanical mirrors~\cite{RMP.86.1391} make this high squeezing promising also for mechanical systems.

For released masses, the experimental requirements are more relaxed. 
Their mass can be considerably smaller while the frequency for initial trapping considerably higher, which is close to common experimental parameters used for optomechanical system \cite{nanomirror,nm,RMP.86.1391}.
Note that higher frequency (lower $\Delta x(0)$) actually improves entanglement gain, unlike in the oscillators case where small $\omega$ is preferable.
However, one has to be cautious of decoherence mechanisms as a result of faster spreading rate of the wave functions. 
For future experiments, an improvement in the sensitivity of entanglement detection will also be beneficial.

\section{Gravity: Direct or mediated interaction?}\label{SC_directmediated}

We conclude with analysis of implications on the nature of the gravitational coupling that one can draw from such experiments and future research directions.
In laboratory, we deal with two nearby masses which are experimentally shown to become entangled.
These setups can be theoretically treated in different ways depending on the assumptions one makes about gravity.
In a ``conservative approach'' the two masses are coupled via Newtonian potential.
As seen from our calculations and those in Refs.~\cite{gravity1,gravity2} this indeed leads to gravitationally-induced entanglement.
In this picture gravity is a direct interaction and hence it is difficult to draw conclusions about the form of quantumness of the gravitational field.
We note that even in this conservative approach such an experiment has considerable value as it would show the necessity of at least the quadratic term in the expansion of the Newtonian potential for generating entanglement.

The objection to the conservative approach is instantaneity of gravitational interaction: Newtonian potential directly couples masses independently of their separation.
\textcolor{black}{On the other hand, it has been shown that gravitational waves travel with finite speed~\cite{GravWaves}.}
For nearby masses this retardation is hardly measurable and Newtonian potential is \textcolor{black}{dominant and} expected to correctly describe the amount of generated entanglement.
A more consistent option in our opinion, motivated by quantum formalism and comparison with other fundamental interactions, is to treat gravitational field as a separate physical object.
In this picture the masses are not directly coupled, but each of them individually interacts with the field.
It has been argued within this mass-field-mass setting that entanglement gain between the masses implies non-classical features of the field~\cite{krisnanda2017,gravity1,gravity2}, see Chapter~\ref{Chapter_revealing}.

This discussion shows that it would be useful to provide methods for independent verification of the presence or absence of a physical object mediating the interaction.
We finish with a toy example of a condition capable of revealing that there was \emph{no} mediator.
To this end we consider two scenarios: (i) evolving a bipartite system described at time $t$ by a density matrix $\rho_{12}(t)$;
(ii) two objects interacting via a mediator $M$, i.e., with Hamiltonian $H_{1M} + H_{M2}$, described by a tripartite state $\rho_{1M2}(t)$.
We ask whether there exists bipartite quantum dynamics $\rho_{12}(t)$ that cannot be obtained by tracing out the mediator in scenario (ii).
Indeed, if $\rho_{12}(t)$ is a pure state at all times and entanglement increases, the dynamics could not have been mediated.
The purity assumption requires the mediator to be uncorrelated from $\rho_{12}(t)$, and uncorrelated mediator is not capable of entangling the principal system, composed of particles $1$ and $2$~\cite{krisnanda2017}.
It would be valuable to generalise this argument to mixed states measured at finite number of time instances.

\textcolor{black}{For example, in an experimental situation, the state of particles $1$ and $2$ might only be close to a pure state (with purity $1-\epsilon$, where $\epsilon$ is a small positive parameter) and therefore they could be weakly entangled to the mediator $M$ (with entanglement $\sim \mathcal{O}(\epsilon)$).
One would naturally expect that entanglement gain between particles $1$ and $2$ is bounded, e.g., as a function of $\epsilon$.
An observation of higher entanglement gain therefore excludes the possibility of mediated dynamics.}

\section{Summary}
No experiment to date verified that gravity possesses quantum nature.
A possible way towards providing such evidence requires the generation of quantum entanglement between massive objects.
In this chapter, we have provided systematic study of entanglement dynamics between two masses that are coupled gravitationally.
We put forward two proposals where the masses are either trapped in 1D harmonic potentials or released from such traps. 
We derived figures of merit that are useful for estimating experimental parameters needed to produce entanglement.
Our schemes have the advantage that no macroscopic superpositions develop during the dynamics, resulting in observable entanglement generated within the coherence times of the masses. 
Finally, we concluded with a discussion on the nature of the gravitational interactions that can be infered from our proposals.


\afterpage{\blankpage} 

\chapter{Probing quantum features of photosynthetic organisms} 

\label{Chapter_probing} 

\lhead{Chapter 6. \emph{Probing quantum features of photosynthetic organisms}} 

\emph{This chapter studies the presence of quantum features in photosynthetic organisms, particularly the green sulphur bacteria.\footnote{Parts of this chapter are reproduced from our published article of Ref.~\cite{bacteria}, which is licensed under the Creative Commons Attribution 4.0 International License (http://creativecommons.org/licenses/by/4.0/). Where applicable, changes made will be indicated.}
Our aim is to reveal that the organisms can possess quantumness without directly measuring them.
For this purpose, I will start by presenting our experimental proposal that involves placing the bacteria inside a Fabry-Perot cavity.
In order to show that our proposal is viable, I proceed with the modelling of the bacterial interactions with cavity fields and environment. 
The dynamics of interesting features, such as quantum entanglement, will be presented.  
It is confirmed with our simulations that quantum entanglement between the cavity fields indeed shows that parts of the bacteria are non-classical.
}

\clearpage

\section{Motivation and objectives}

There is no a priori limit on the complexity, size or mass of objects to which quantum theory is applicable.
Yet, whether or not the physical configuration of macroscopic systems could showcase quantum coherences has been the subject of a long-standing debate. 
The pioneers of quantum theory, such as Schr\"odinger \cite{SCH} and Bohr \cite{BOH}, wondered whether there might be limitations to living systems obeying the laws of quantum theory. 
Wigner even claimed that their behaviour violates unitarity~\cite{WIG}. 

A striking way to counter such claims on the implausibility of macroscopic quantum coherence would be the successful preparation of quantum superposition states of living objects. 
A direct route towards such goal is provided by matter-wave interferometers, which have already been instrumental in observing quantum interference from complex molecules \cite{arndt}, and are believed to hold the potential to successfully show similar results for objects as large as viruses in the near future.

However, other possibilities exist that do not make use of interferometric approaches. 
An instance of such alternatives is to interact a living object with a quantum system in order to generate quantum correlations.
Should such correlations be as strong as entanglement, measuring the quantum system in a suitable basis could project the living object into a quantum superposition. 
Furthermore, requesting the establishment of entanglement is, in general, not necessary as the presence of quantum discord, that is a weaker form of quantum correlations, would already provide evidence that the Hilbert space spanned by the living object must contain quantum superposition states~\cite{discord1,discord2,discord3,discord4,discord5}.
For example, by operating on the quantum system alone one could remotely prepare quantum coherence in the living object~\cite{coherence-distillation}.

A promising step in this direction, demonstrating strong coupling between living bacteria and optical fields and suggesting the existence of entanglement between them~\cite{bacteria-th}, has recently been realised~\cite{bacteria-exp}.
See also Refs. \cite{bio1,bio2,bio3,bio4,bio5,bio6,bio7,bio8,bio9,bio10,bio11} for a broader picture of quantum effects in photosynthetic organisms.
However, the experimental results reported in Ref.~\cite{bacteria-exp} can as well be explained by a fully classical model~\cite{ZHU, bacteria-th,bacteria-exp,rabienergy}, which calls loud for the design of a protocol with more conclusive interpretation.

In this chapter we make a proposal in such a direction by designing a thought experiment in which the bacteria are mediating interactions between otherwise uncoupled light modes.
This scheme fits into the general framework of Ref.~\cite{krisnanda2017}, which shows in the present context that quantum entanglement between the light modes can only be created if the bacteria are non-classically correlated with them during the process.
It is important to realise that in this way we bypass the need of exact modelling of the living organisms and their interactions with external world.
Indeed, experimenters are never asked to directly operate on the bacteria, it is solely sufficient to observe the light modes.
A positive result of this experiment, i.e. observation of quantum entanglement between the light modes, provides an unambiguous witness of quantum correlations, in the form of quantum discord, between the light and bacteria.

In order to demonstrate that there should be observable entanglement in the experiment we then propose a plausible model of light-bacteria interactions and noises in the experiment.
We focus on the optical response of the bacteria and model their light-sensitive part by a collection of two-level atoms with transition frequencies matching observed bacterial spectrum~\cite{bacteria-exp}.
All processes responsible for keeping the organisms alive are thus effectively put into the environment of these atoms.
We argue that standard Langevin approach gives a sensible treatment of this environment due to its quasi-thermal character, low energies compared to optical transitions and no evidence for finite-size effects.
Within this model we find scenarios with non-zero steady-state entanglement between the light modes which is always accompanied by light-bacteria entanglement (in addition to quantum discord), which is in turn empowered by the ultra-strong coupling\footnote{Note that we refer to the ultra-strong coupling as the regime where the so-called counter-rotating terms in the Hamiltonian have considerable strength. These terms are responsible for simultaneous creation (and annihilation) of excitations both in the light and bacterial modes.} between such systems.

\section{Proposed setup: Bacteria in a Fabry-Perot cavity}

Our idea is to design a setup which, on one hand, is close to what has already been realised with bacteria and light, in order to utilise their strong coupling,
and whose description, on the other hand, can be phrased within the framework of Chapter~\ref{Chapter_revealing}.
It was shown there that two physical systems, $A$ and $B$, coupled via a mediator $C$, i.e., described by a total Hamiltonian of the form $H_{AC} + H_{BC}$,
can become entangled only if quantum discord $D_{AB|C}$ is generated during the evolution.
This also holds if each system is allowed to interact with its own local environment.
Therefore, observation of quantum entanglement between $A$ and $B$ is a witness of quantum discord $D_{AB|C}$ during the evolution if one can ensure the following conditions:
\begin{itemize}
\item[(i)] $A$ and $B$ do not interact directly, i.e., there is no term $H_{AB}$ in the total Hamiltonian.
\item[(ii)] All environments are local, i.e., they do not interact with each other.
\item[(iii)] The initial state is completely unentangled (otherwise entanglement between $A$ and $B$ can grow via classical $C$~\cite{krisnanda2017}).
\end{itemize}

We now propose a concrete scheme for revealing non-classicality of the bacteria and argue how it meets the conditions above.
Consider the arrangement in Fig.~\ref{FIG_pr_setup}.
The bacteria are inside a driven single-sided multimode Fabry-Perot cavity where they interact independently with a few cavity modes.
The cavity modes are divided into two sets which play the role of systems $A$ and $B$ in the general framework.
The bacteria are mediating the interaction between the modes and hence they represent system $C$.
Condition (i) above can be realised in practice in at least two ways.
An experimenter could utilise the polarisation of electromagnetic waves and group optical modes polarised along one direction to system $A$ and those polarised orthogonally to system $B$.
Another option, which we will study in detail via a concrete model below, is to choose different frequency modes and arbitrarily group them into systems $A$ and $B$.
Condition (ii) holds under typical experimental circumstances where the environment of the cavity modes is outside the cavity
whereas that of the bacteria is inside the cavity or even part of bacteria themselves.
The electromagnetic environment outside the cavity is a large system giving rise to the decay of cavity modes but having no back-action on them.
Therefore each cavity mode decays independently and cannot get entangled via interactions with the electromagnetic environment.
Finally, condition (iii) is satisfied right before placing the bacteria into the cavity, because at this time all three systems $A$, $B$, and $C$ are in a completely uncorrelated state $\rho_A \otimes \rho_B \otimes \rho_C$.

\begin{figure}[h]
\centering
\includegraphics[scale=0.4]{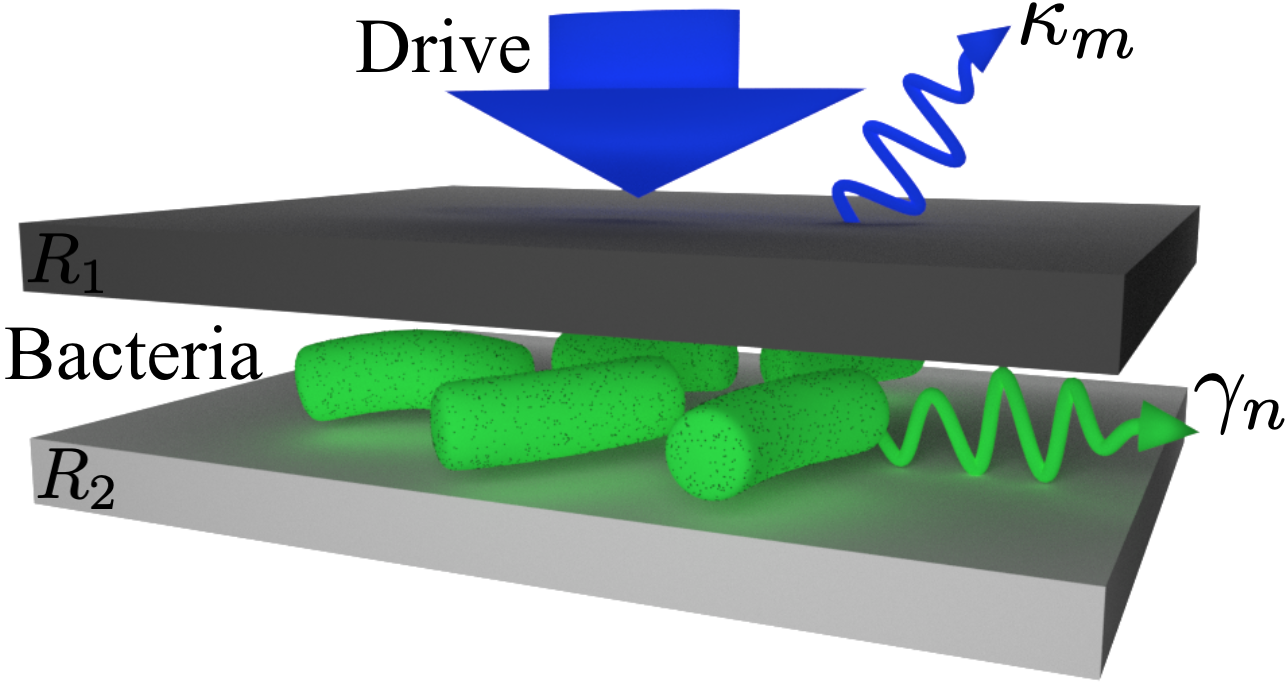}
\caption{Proposed experimental setup for probing quantum features of photosynthetic organisms.
In particular, green sulphur bacteria are mediating interactions between a few cavity modes in a driven multimode Fabry-Perot cavity, which are otherwise non-interacting. 
The cavity has an input mirror with reflectivity $R_1$ and an end mirror with perfect reflectivity, i.e., $R_2\approx 1$.
As a result of interactions with the environment, the $m\text{th}$ cavity mode has a decay rate $\kappa_m$, while the bacterial system dissipates at a rate $\gamma_n$.
In the text, we treat cavity fields and bacterial modes as open systems.
}
\label{FIG_pr_setup} 
\end{figure}

We note again that this discussion is generic with almost no modelling of the involved systems.
In particular, nothing has been assumed regarding the physics of the bacteria and their interactions with light and the external world.
This makes our proposal experimentally attractive.

In order to make concrete predictions about the amount of intermodal entanglement $E_{A:B}$ we now study a specific model for the energy of the discussed system.
This additional assumption about the overall Hamiltonian will allow us to demonstrate that the entanglement $E_{A:B}$ is accompanied by light-bacteria entanglement $E_{AB:C}$.
This independently confirms the presence of light-bacteria discord as entanglement is a stronger form of quantum correlations than discord~\cite{discord3,discord4,discord5}.
In the remainder of this chapter we will therefore only calculate entanglement.

\section{Model of optical coupling}

We consider a photosynthetic bacterium, \emph{Chlorobaculum tepidum}, that is able to survive in extreme environments with almost no light~\cite{blankenship1995antenna}. Each bacterium, which is approximately $2\mu\mbox{m}\times500\mbox{nm}$ in size, contains $200-250$ chlorosomes, each having $200,000$ bacteriochlorophyll $c$ (BChl $c$) molecules. Such pigment molecules serve as excitons that can be coupled to light \cite{bacteria-exp,chlorosomes}. 
The extinction spectrum of the bacteria (BChl $c$ molecules) in water shows two pronounced peaks, at wavelengths $\lambda_{\scriptsize\mbox{I}} = 750$nm and $\lambda_{\scriptsize\mbox{II}} = 460$nm (see Fig.~1b of Ref.~\cite{bacteria-exp}).
We therefore model the light-sensitive part of the bacteria by two collections of $N$ two-level atoms with transition frequencies $(\Omega_{\scriptsize\mbox{I}},\Omega_{\scriptsize\mbox{II}})=(2.5,4.1)\times 10^{15}$ Hz.
Simplification of this model to atoms with a single transition frequency was already shown to be able to explain the results of recent experiments~\cite{bacteria-exp, chlorosomes}.
This simplification was adequate because only one cavity mode was relevant in the previous experiments.
In contrast, several cavity modes are required for the observation of intermodal entanglement and it is correspondingly more accurate to include also all relevant transitions of BChl $c$ molecules.
We assume that the molecules (two-level atoms in our model) are coupled through a dipole-like mechanism to each light mode. 
For $N\gg1$, such collections of two-level systems can be approximated to spin $N/2$ angular momenta. In the low-excitation approximation (which we will justify later), such angular momentum can be mapped into an effective harmonic oscillator through the use of the Holstein-Primakoff transformation~\cite{HP}.
This allows us to cast the energy of the overall system as
\begin{eqnarray}\label{EQ_H}
H&=&\sum_m \hbar \omega_m a^{\dagger}_m a_m + \sum_{n}\hbar \Omega_n b^{\dagger}_n b_n +\sum_{m,n} \hbar G_{mn}(a_m+a_m^{\dagger})(b_n+b^{\dagger}_n) \nonumber \\ 
&&+\sum_{m} i\hbar \bm{E}_m (a_m^{\dagger}e^{-i\Lambda_{m} t}-a_m e^{i\Lambda_{m}t}).
\end{eqnarray}
Here, $m=1,\dots,M$ is the label for the $m\text{th}$ cavity mode, whose annihilation (creation) operator is denoted by $a_m$ ($a_m^\dag$) and having frequency $\omega_m$. 
Moreover, both harmonic oscillators describing the bacteria are labelled by $n=\scriptsize\mbox{I},\scriptsize\mbox{II}$ with $b_n$ ($b^\dag_n$) denoting the corresponding bosonic annihilation (creation) operator.
Each oscillator is coupled to the $m\text{th}$ cavity field at a rate $G_{mn}$. 
The collective form of the coupling allows us to write $G_{mn}=g_{mn}\sqrt{N}$ with $g_{mn}=\mu_n \sqrt{\omega_m/2\hbar \epsilon_r\epsilon_0 V_m}$, where $\mu_n$ is the dipole moment of the $n\text{th}$ two-level transition, $\epsilon_r$ relative permitivity of medium, and $V_m$ the $m\text{th}$ mode volume~\cite{rabienergy}, see also \cite{optoatoms,bai2016robust} for similar treatments. 
The cavity is driven by a multimode laser, each mode having frequency $\Lambda_{m}$, amplitude $\bm{E}_m=\sqrt{2\bm{P}_m\kappa_m/\hbar \Lambda_{m}}$, power $\bm{P}_m$, and amplitude decay rate of the corresponding cavity mode $\kappa_m$. 
It is important to notice that in Eq.~(\ref{EQ_H}) we have not invoked the rotating-wave approximation but actually retained the counter-rotating terms $a_mb_n$ and $a_m^{\dagger}b_n^{\dagger}$. These cannot be ignored in the regime of ultra-strong coupling and we will show that they actually play a crucial role in our proposal.

\section{Dynamics of light-bacteria system}

\subsection{Langevin equations}

We assume the local environment of the light-sensitive part of the bacteria to give rise to Markovian open-system dynamics, which is modelled as decay of the two-level systems.
For justification we note that in actual experiments the bacteria are surrounded by water which can be treated as a standard heat bath and although the environment of interest cannot be in a thermal state (because the bacteria are alive) its state is expected to be quasi-thermal.
Given that the bacterial environment of the BChl $c$ molecules is of finite size we should also justify the Markovianity assumption. 
To the best of our knowledge there is no experimental evidence against this assumption. 
Likely this is due to the fact that all excitations arriving at this environment are further rapidly dissipated to the large thermal environment of water, whose energy is small compared to the optical transitions.

We treat the environment of the cavity modes as the usual electromagnetic environment outside the cavity~\cite{noise,walls2007quantum}.
This results in independent decay rates of each mode.
Taken all together, the dynamics of the optical modes and bacteria can be written using the standard Langevin formulation in Heisenberg picture. 
This gives the following equations of motion, taking into account noise and damping terms coming from interactions with the local environments
\begin{equation}
\label{Eq_lgvn}
\begin{aligned}
\dot {{a}}_m&=-(\kappa_m+i\omega_m)a_m - i\sum_{n}G_{mn}(b_n+b_n^{\dagger})+\bm{E}_m e^{-i\Lambda_{m}t}+\sqrt{2\kappa_m}\: F_m,\\
\dot {{b}}_n&=-(\gamma_n+i\Omega_n)b_n-i\sum_{m}G_{mn}(a_m+a_m^{\dagger})+\sqrt{2\gamma_n}\:Q_n, 
\end{aligned}
\end{equation}
where $\gamma_n$ is the amplitude decay rate of the bacterial system. 
$F_m$ and $Q_n$ are operators describing independent zero-mean Gaussian noise affecting the $m\text{th}$ cavity field and the $n\text{th}$ bacterial mode respectively.
The only nonzero correlation functions between these noises are $\langle F_m(t)F_{m'}^{\dagger}(t^{\prime}) \rangle=\delta_{mm'}\delta(t-t^{\prime})$ 
and $\langle Q_n(t) Q_{n'}^{\dagger}(t^{\prime}) \rangle=\delta_{nn'}\delta(t-t^{\prime})$~\cite{noise,walls2007quantum}.
We note that in this model the light-sensitive part of the bacteria is treated collectively, i.e. all its two-level atoms are indistinguishable.
This assumption is standardly made in present-day literature, see e.g.~\cite{bacteria-exp,chlorosomes} where modelling of the bacteria / chlorosomes as a harmonic oscillator fits observed experimental results.
But it should be stressed that this assumption deserves an in-depth experimental assessment.

We express the Langevin equations in terms of mode quadratures. 
In particular, we use $ \bm{X}_m\equiv ( a_m+ a_m^{\dagger})/\sqrt{2}$ and $ \bm{Y}_m\equiv ( a_m-a_m^{\dagger})/i\sqrt{2}$.
This allows us to write the Langevin equations for the quadratures in a simple matrix equation $\dot u(t)=Ku(t)+l(t)$, with the vector $u=( \bm{X}_{1}, \bm{Y}_{1},\cdots, \bm{X}_{M}, \bm{Y}_{M}, \bm{X}_{\scriptsize\mbox{I}}, \bm{Y}_{\scriptsize\mbox{I}},\bm{X}_{\scriptsize\mbox{II}}, \bm{Y}_{\scriptsize\mbox{II}})^T$ and
\begin{equation}
K=\left( \begin{array}{cccccc} 
I_1&\bm{0}&\cdots& \bm{0} &L_{1\scriptsize\mbox{I}}&L_{1\scriptsize\mbox{II}}\\
\bm{0}&I_2 &\cdots&\bm{0}&L_{2\scriptsize\mbox{I}}&L_{2\scriptsize\mbox{II}}\\
\vdots&\vdots&\ddots&\vdots&\vdots&\vdots\\
\bm{0}&\bm{0}&\cdots&I_{M}&L_{M\scriptsize\mbox{I}}&L_{M\scriptsize\mbox{II}}\\
L_{1\scriptsize\mbox{I}}&L_{2\scriptsize\mbox{I}}&\cdots&L_{M\scriptsize\mbox{I}}&I_{\scriptsize\mbox{I}}&\bm{0}\\
L_{1\scriptsize\mbox{II}}&L_{2\scriptsize\mbox{II}}&\cdots&L_{M\scriptsize\mbox{II}}&\bm{0}&I_{\scriptsize\mbox{II}}\\ 
\end{array}\right),
\end{equation}
where the components are $2\times2$ matrices given by
\begin{equation}
I_m=\left( \begin{array}{cc}  
-\kappa_m &\omega_m\\
-\omega_m&-\kappa_m\\
\end{array}\right), \: 
L_{mn}=\left( \begin{array}{cc}  
0 &0\\
-2G_{mn}&0\\
\end{array}\right), \nonumber
\end{equation}
\begin{equation} 
I_{n}=\left( \begin{array}{cc}  
-\gamma_n &\Omega_n\\
-\Omega_n&-\gamma_n\\
\end{array}\right),
\end{equation}
and $\bm{0}$ is a $2\times2$ zero matrix.
Note that we have used the index $m=1,2,\cdots,M$ for the cavity modes and $n=\scriptsize\mbox{I},\scriptsize\mbox{II}$ for the bacterial modes. 
We split the last term in the matrix equation into two parts, representing the noise and pumping respectively, i.e. $l(t)=\eta(t)+p(t)$ where 
\begin{equation}
\frac{\eta(t)}{\sqrt{2}}=\left( \begin{array}{c} 
\sqrt{\kappa_1}\: X_1(t)\\ 
\sqrt{\kappa_1}\: Y_1(t)\\ 
\vdots\\
\sqrt{\kappa_{M}}\: X_{M}(t)\\ 
\sqrt{\kappa_{M}}\: Y_{M}(t)\\ 
\sqrt{\gamma_{\scriptsize\mbox{I}}}\: X_{\scriptsize\mbox{I}}(t)\\ 
\sqrt{\gamma_{\scriptsize\mbox{I}}}\: Y_{\scriptsize\mbox{I}}(t)\\
\sqrt{\gamma_{\scriptsize\mbox{II}}}\: X_{\scriptsize\mbox{II}}(t)\\ 
\sqrt{\gamma_{\scriptsize\mbox{II}}}\: Y_{\scriptsize\mbox{II}}(t)\\
\end{array}\right),
\frac{p(t)}{\sqrt{2}}=\left( \begin{array}{c} 
\bm{E}_{1}\cos{\Lambda_{1}t}\\ 
-\bm{E}_{1}\sin{\Lambda_{1}t}\\ 
\vdots \\
\bm{E}_{M}\cos{\Lambda_{M}t}\\ 
-\bm{E}_{M}\sin{\Lambda_{M}t}\\ 
0\\ 
0\\
0\\ 
0\\
\end{array}\right).
\end{equation}
We have also used quadratures for the noise terms, i.e. through $F_m=(X_m+iY_m)/\sqrt{2}$ and $Q_n=(X_n+iY_n)/\sqrt{2}$.

The solution to the Langevin equations is given by 
\begin{eqnarray}\label{AEQ_Lsol}
u(t)&=&W_+(t)u(0)+W_+(t)\int_0^t dt^{\prime}  W_-(t^{\prime})l(t^{\prime}),
\end{eqnarray}
where $W_{\pm}(t)=\exp{(\pm Kt)}$. 
This allows numerical calculation of expectation value of the quadratures as a function of time, i.e. $\langle u_i(t)\rangle$ is given by the $i{\text{th}}$ element of 
\begin{equation}\label{EQ_ui}
W_+(t)\langle u(0)\rangle+W_+(t)\int_0^t dt^{\prime}  W_-(t^{\prime})p(t^{\prime}),
\end{equation}
which is obtained as follows.
Since every component of $p(t)$ is not an operator, we have $\langle p_k(t)\rangle=\mbox{tr}(p_k(t)\rho)=p_k(t)$. 
Also, we have used the fact that the noises have zero mean, i.e. $\langle \eta_k(t)\rangle=0$.

\subsection{Covariance matrix}

Covariance matrix of our system is defined as $V_{ij}(t)\equiv \langle \{ \Delta u_i(t),\Delta u_j(t)\}\rangle/2=\langle u_i(t)u_j(t)+u_j(t)u_i(t)\rangle/2-\langle u_i(t)\rangle \langle u_j(t)\rangle$ where we have used $\Delta u_i(t)=u_i(t)-\langle u_i(t)\rangle$. 
This means that $p(t)$ does not contribute to $\Delta u_i(t)$ (and hence the covariance matrix) since $\langle p_k(t)\rangle=p_k(t)$. 
We can then construct the covariance matrix at time $t$ from Eq. (\ref{AEQ_Lsol}) without considering $p(t)$ as follows
\begin{eqnarray}
V_{ij}(t)&=&\langle u_i(t)u_j(t)+u_j(t)u_i(t)\rangle/2-\langle u_i(t)\rangle \langle u_j(t)\rangle \nonumber \\
\label{EQ_l3}
V(t)&=&W_+(t)V(0)W_+^T(t) \nonumber \\
&&+W_+(t)\int_0^t dt^{\prime} W_-(t^{\prime})DW_-^T(t^{\prime}) \:W_+^T(t) ,
\end{eqnarray}
where $D=\mbox{Diag}[\kappa_1,\kappa_1,\cdots,\kappa_M,\kappa_M,\gamma_{\scriptsize\mbox{I}},\gamma_{\scriptsize\mbox{I}},\gamma_{\scriptsize\mbox{II}},\gamma_{\scriptsize\mbox{II}}]$ and we have assumed that the initial quadratures are not correlated with the noise quadratures such that the mean of the cross terms are zero. 
A more explicit solution of the covariance matrix, after integration in Eq. (\ref{EQ_l3}), is given by 
\begin{eqnarray}\label{EQ_Ct}
KV(t)+V(t)K^T&=&-D+KW_+(t)V(0)W_+^T(t) \nonumber \\
&&+W_+(t)V(0)W_+^T(t)K^T \nonumber \\
&&+W_+(t)DW_+^T(t),
\end{eqnarray}
which is linear and can be solved numerically. 

Time evolution of important quantities can then be calculated from the covariance matrix, e.g., entanglement and excitation number.
If one is only interested in the steady state, it is guaranteed when all real parts of the eigenvalues of $K$ are negative.
In this case the covariance matrix satisfies Lyapunov-like equation
\begin{equation}\label{EQ_Css}
K\:V(\infty)+V(\infty)\:K^T+D=0,
\end{equation}
where $D=\mbox{Diag}[\kappa_1,\kappa_1,\cdots,\kappa_M,\kappa_M,\gamma_{\scriptsize\mbox{I}},\gamma_{\scriptsize\mbox{I}},\gamma_{\scriptsize\mbox{II}},\gamma_{\scriptsize\mbox{II}}]$.
Note that the steady-state covariance matrix does not depend on the initial conditions, i.e. $V(0)$. 
Moreover, as the Langevin equations are linear and due to the gaussian nature of the quantum noises, the dynamics of the system is preserving gaussianity. 
Therefore the steady state is a continuous variable gaussian state completely characterised by $V(\infty)$.

\subsection{Quantum entanglement}

\subsubsection{Entanglement from covariance matrix}

Logarithmic negativity is chosen as entanglement quantifier and below we provide the details on how this quantity is calculated.
 
The covariance matrix $V$ describing our system can be written in block form 
\begin{equation}\label{EQ_COV}
V=\left( \begin{array}{cccc} 
B_{11}&B_{12}&\cdots&B_{1Z}\\
B_{12}^T&B_{22}&\cdots&B_{2Z}\\ 
\vdots&\vdots&\ddots&\vdots \\ 
B_{1Z}^T&B_{2Z}^T&\cdots&B_{ZZ}\\
\end{array}\right),
\end{equation}
where $Z$ is the total number of modes, which is $M+2$ in our case. 
The block component, here denoted as $B_{jk}$, is a $2\times 2$ matrix describing local mode correlation when $j=k$ and intermodal correlation when $j\ne k$.
A $Z$-mode covariance matrix has symplectic eigenvalues $\{\nu_k\}_{k=1}^{Z}$ that can be computed from the spectrum of matrix $|i\Omega_{Z} V|$ \cite{weedbrook2012gaussian} where 
\begin{equation}
 \Omega_{Z}=\bigoplus^{Z}_{k=1} \left( \begin{array}{cc} 0&1\\ -1 &0\end{array}\right).
 \end{equation}
For a physical covariance matrix $2 \nu_k\ge 1$. 

Entanglement is calculated as follows. For example, the calculation in the partition $12:34$ only requires the covariance matrix of modes $1,2,3$, and $4$:
\begin{equation}
V=\left( \begin{array}{cccc} 
B_{11}&B_{12}&B_{13}&B_{14}\\
B_{12}^T&B_{22}&B_{23}&B_{24}\\ 
B_{13}^T&B_{23}^T&B_{33}&B_{34} \\ 
B_{14}^T&B_{24}^T&B_{34}^T&B_{44}\\
\end{array}\right),
\end{equation}
that can be obtained from Eq. (\ref{EQ_COV}).
If the covariance matrix $\tilde V$, after partial transposition with respect to mode $3$ and $4$ (this is equivalent to flipping the sign of the operator $\bm{Y}_3$ and $\bm{Y}_4$ in $V$) is not physical, then our system is entangled. 
This unphysical $\tilde V$ is shown by its minimum symplectic eigenvalue $\tilde \nu_{\text{min}}<1/2$.
Entanglement is then quantified by logarithmic negativity as follows $E_{12:34}=\mbox{max}[0,-\log_2{(2\tilde \nu_{\text{min}})}]$ \cite{negativity, adesso2004extremal}.\footnote{Note that we have used logarithm with base 2, instead of natural logarithm as in Ref.~\cite{bacteria}, to maintain consistencies in this thesis.}
Note that the separability condition, when $\tilde \nu_{\text{min}}\ge1/2$, is sufficient and necessary when one considers bipartitions with one mode on one side \cite{werner2001bound}, e.g., partition between bacterial modes $\mbox{I}:\mbox{II}$. 

\subsubsection{Simulation parameters}

The parameters used in our simulations are taken, wherever possible, from the experiments of Ref.~\cite{bacteria-exp}.
We place the bacteria in a single-sided Fabry-Perot cavity of length $L=518$ nm (cf. Fig. \ref{FIG_pr_setup}). 
The refractive index due to aqueous bacterial solution embedded in the cavity is $n_r=\sqrt{\epsilon_r} \approx 1.33$, which gives the frequency of the $m\text{th}$ cavity mode $\omega_m=m\pi c/n_rL\approx 1.37m\times 10^{15}$ Hz.
The reflectivities of the mirrors are engineered such that $R_2 = 100\%$ and $R_1 = 50\%$.
We assume the reflectivities are the same for all the optical modes, giving $\kappa_m \approx 7.5\times 10^{13}$ Hz through the finesse $\mathcal{F}_i=-2\pi/\ln{(R_1R_2)}=\pi c/2 \kappa_m n_r L$.
The decay rate of the excitons can be calculated as $\gamma_n=1/2\tau_n$ where $\tau_n=2h/\Gamma_n$ is the coherence time with $\Gamma_n$ being the full-width at half-maximum (FWHM) of the bacterial spectrum \cite{bajoni2012}. 
We approximate the spectrum in Fig. 1b of Ref. \cite{bacteria-exp} as a sum of two Lorentzian functions centred at $\Omega_{\scriptsize\mbox{I}}$ and $\Omega_{\scriptsize\mbox{II}}$ having FWHM of $(\Gamma_{\scriptsize \mbox{I}},\Gamma_{\scriptsize \mbox{II}})=(130,600)$ meV, giving $(\gamma_{\scriptsize\mbox{I}},\gamma_{\scriptsize\mbox{II}}) \approx (0.78,3.63)\times 10^{13}$ Hz respectively.
Note that the decay rate solely depends on the coherence time, i.e., we assume only homogenous broadening of the spectral lines.

All the spectral components of the driving laser are assumed to have the same power $\bm{P}_m=50$mW and frequency $\Lambda_{m}=\omega_m$. 
By using the mode volume $V_m=2\pi L^3/m(1-R_1)$ \cite{modeV}, we can express the interaction strength as $G_{mn}= m\tilde G_n$, where we define $\tilde G_n \equiv \mu_n \sqrt{c(1-R_1)N/4\hbar n_r^3 \epsilon_0L^4}$. 
This quantity is a rate that characterises the base collective interaction strength of the cavity mode and the $n\text{th}$ bacterial mode. 
Instead of fixing the value of $\tilde G_n$, we vary this quantity $\tilde G_n=[0,0.2]\: 10^{15}\: \mbox{Hz}$, which is within experimentally achievable regime (cf. Refs. \cite{bacteria-exp,bacteria-th}).

\subsubsection{Entanglement dynamics}

Let us consider as initial the time right before the bacteria are inserted into the cavity.
Then all the cavity modes and the bacteria are completely uncorrelated and do not interact.
The dynamics is then started by placing the bacteria in the cavity.
In what follows, as an example of the dynamics we start with vacuum state for the cavity modes and ground state for the bacteria. 
The initial state of the bacteria is justified by the fact that $\hbar \Omega_n \gg k_BT$, even at room temperature. 

Fig.~\ref{FIG_pr_dynamics} shows the resulting entanglement dynamics.
Panel (a) displays existence of steady-state entanglement between cavity modes $1,2$ and $3,4$, which is not altered heavily if the calculations take into account five and six cavity modes in total. 
Therefore, we consider $4$ cavity modes in all other calculations.
In recent experiments, the rate $\tilde G_{\scriptsize \mbox{I}}$ was shown to be $3.9\times 10^{13}\:\mbox{Hz}$ \cite{bacteria-exp} and the corresponding $\tilde G_{\scriptsize \mbox{II}}=6\times 10^{13}\:\mbox{Hz}$.
In our calculations we vary the latter rate as in panels (b) and (c) (also see Fig. \ref{FIG_pr_steadystate}).
As expected the higher the rate the more entanglement gets generated.
It is also apparent that entanglement between the cavity modes and bacteria $E_{1234:\scriptsize \mbox{I}\: \scriptsize \mbox{II}}$ grows faster than entanglement between the cavity modes.
More precisely, nonzero $E_{12:34}$ implies nonzero $E_{1234:\scriptsize \mbox{I}\: \scriptsize \mbox{II}}$.

\begin{figure}[h!]
\centering
\includegraphics[scale=0.6]{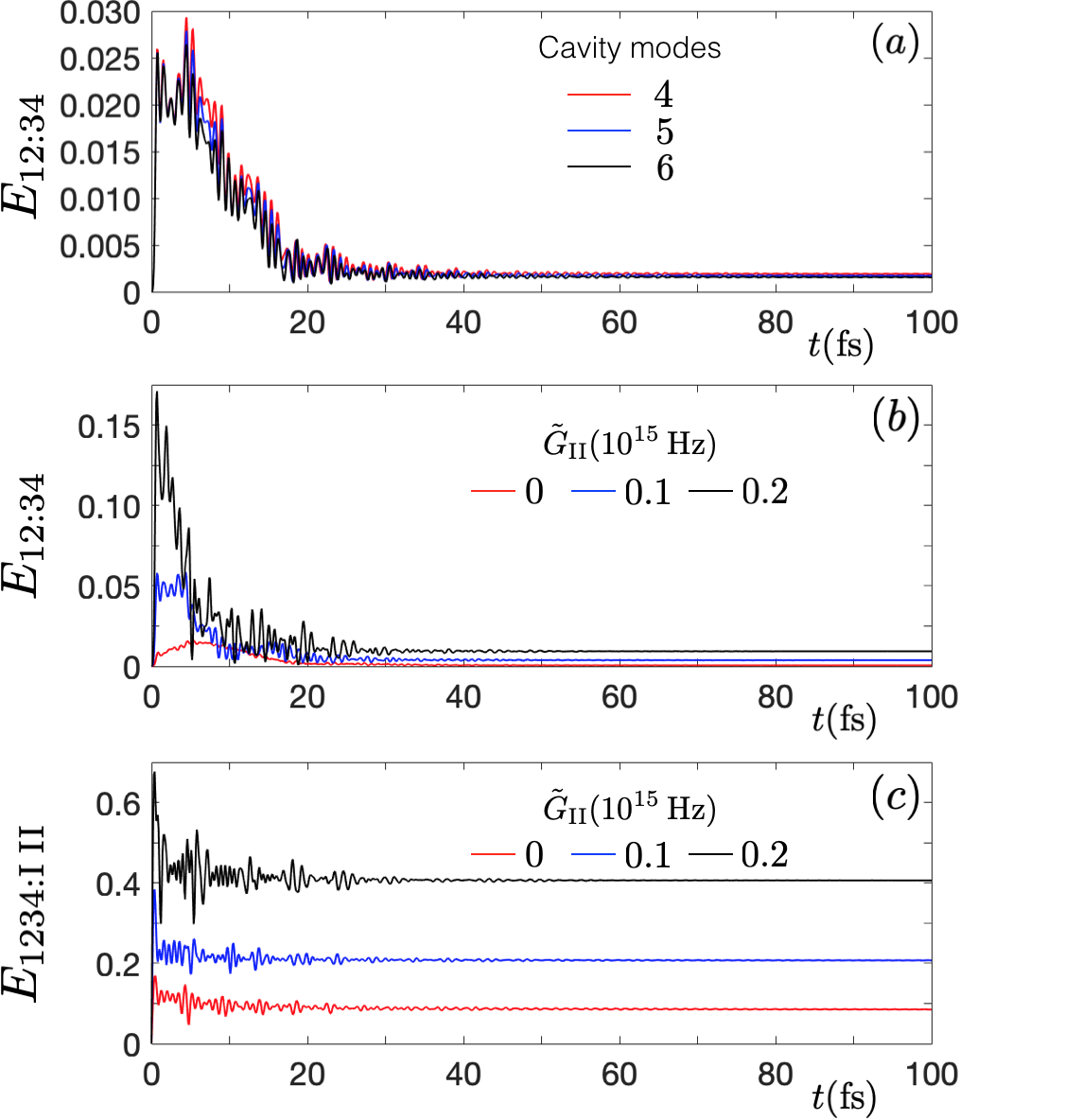}
\caption{Entanglement dynamics in the field-bacteria-field system.
(a) Entanglement $E_{12:34}$ between the cavity modes $\{1,2\}$ and $\{3,4\}$.
The entanglement is shown up to $4$, $5$, and $6$ cavity modes in the cavity. 
Higher modes do not contribute to the steady state entanglement due to higher detuning from the bacterial system.
We therefore only consider $4$ cavity modes for the rest of the dynamics.
(b) Entanglement between the cavity modes for varied interaction strength with the bacterial system.
(c) Cavity modes--bacteria entanglement.
The correlation from this partition is higher and showing faster initial growth compared to the others.
The coupling parameters used are $\tilde G_{\scriptsize\mbox{I}}=3.9\times 10^{13}\:\mbox{Hz}$ for all graphs and $\tilde G_{\scriptsize\mbox{II}}=6 \times10^{13}\:\mbox{Hz}$ for (a). 
Note that the steady state of all dynamics is reached within $\sim 100$~fs and we have used the logarithmic negativity as entanglement quantifier.}
\label{FIG_pr_dynamics} 
\end{figure}

\subsubsection{Steady state entanglement}

We consider four cavity modes as the addition of higher modes shows negligible effects to the steady state entanglement.
In the steady state regime, we calculate entanglement between the cavity modes $E_{12:34}$, between the cavity modes and bacteria $E_{1234:{\scriptsize\mbox{I}}\:{\scriptsize\mbox{II}}}$, 
and between the bacterial modes $E_{{\scriptsize\mbox{I}}:{\scriptsize\mbox{II}}}$, cf. Fig.~\ref{FIG_pr_steadystate}. 
This steady state regime is reached in $\sim 100 \:\mbox{fs}$ (see Fig.~\ref{FIG_pr_dynamics}), which is faster than relaxation processes ($\sim$ ps) occuring within green sulphur bacteria \cite{chlorosomes}.
Our results show that the steady state entanglement $E_{12:34}$ is always accompanied by $E_{1234:{\scriptsize\mbox{I}}\:{\scriptsize\mbox{II}}}$, i.e., the bacteria are non-classically correlated with the cavity modes. 
This is in agreement with the general detection method of Ref.~\cite{krisnanda2017} as entanglement is a stronger type of quantum correlation than discord, i.e., nonzero $E_{1234:\scriptsize \mbox{I}\: \scriptsize \mbox{II}}$ implies nonzero cavity modes-bacteria discord $D_{1234|\scriptsize \mbox{I}\: \scriptsize \mbox{II}}$.
Our results also show that the entanglement dynamics of $E_{12:34}$ is dominated by modes $2$ and $3$ since other modes are further off resonance with the bacterial modes.
Moreover, there is entanglement generated within the bacteria.
This requires both $\tilde G_{{\scriptsize\mbox{I}}}$ and $\tilde G_{{\scriptsize\mbox{II}}}$ to be nonzero and relatively high.
We see that the bacteria can be strongly entangled with the cavity modes, much stronger than entanglement between the cavity modes. 
While the latter is in the order of $10^{-2} - 10^{-3}$, we note that entanglement in the range $10^{-2}$ has already been observed experimentally between mechanical motion and microwave cavity field~\cite{evalue1}.
We have also indicated, as black dots in Fig. \ref{FIG_pr_steadystate}, the coupling strengths $\tilde G_{{\scriptsize\mbox{I}}}=3.9\times10^{13}\:\mbox{Hz}$ from Ref. \cite{bacteria-exp} and the corresponding $\tilde G_{{\scriptsize\mbox{II}}}=6\times10^{13}\:\mbox{Hz}$, which is estimated as follows.
From the relation $\mu_n^2\propto \int f(\omega)d\omega /\omega_n$ \cite{houssier1970circular}, where $f$ is the extinction coefficient, one can obtain the ratio $\tilde G_{\scriptsize \mbox{II}}/\tilde G_{\scriptsize \mbox{I}}=\mu_{\scriptsize \mbox{II}}/\mu_{\scriptsize \mbox{I}}\approx 1.53$.

\begin{figure}[h!]
\centering
\includegraphics[scale=0.5]{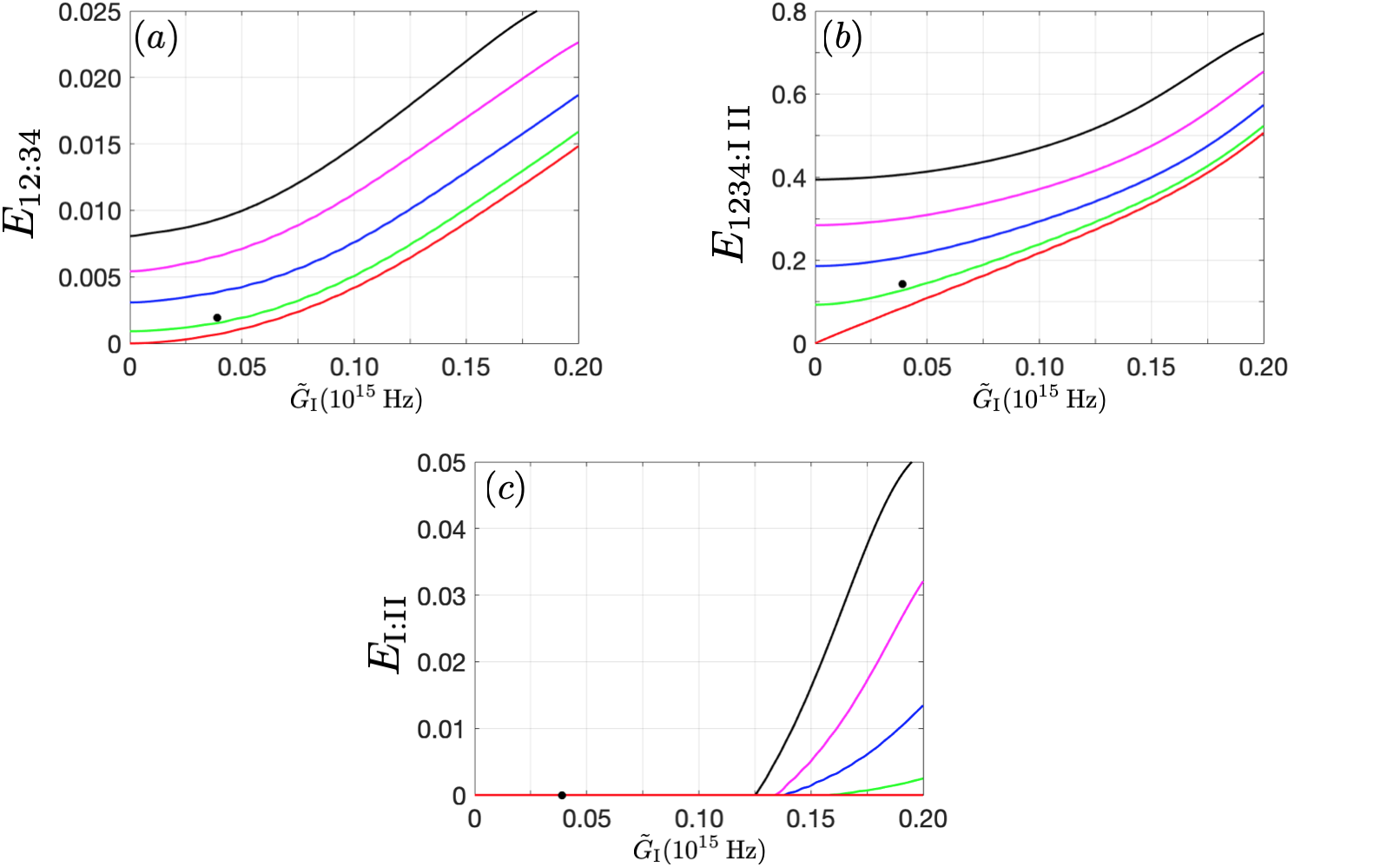} 
\caption{Steady state entanglement in field-bacteria-field system for varied interaction strengths.
In all panels, the horizontal axis is $\tilde G_{{\scriptsize\mbox{I}}}=[0,0.2]\: 10^{15}$~Hz, while $\tilde G_{{\scriptsize\mbox{II}}}$ is varied in $10^{15}$~Hz as: $0$ (red curves), $0.05$ (green curves), $0.1$ (blue curves), $0.15$ (magenta curves), and $0.2$ (black curves).
Black dots in the panels show the values of entanglement for recently realised coupling strengths $\tilde G_{{\scriptsize\mbox{I}}}=3.9\times10^{13}\:\mbox{Hz}$ by Coles \emph{et al}.~\cite{bacteria-exp} and the corresponding $\tilde G_{{\scriptsize\mbox{II}}}=6\times10^{13}\:\mbox{Hz}$.
We note from our simulations that the entanglement $E_{1234:\scriptsize \mbox{I}\: \scriptsize \mbox{II}}$ is nonzero whenever $E_{12:34}$ is present. 
Furthermore, entanglement in the bacterial system, panel (c), also persists for stronger couplings.
}
\label{FIG_pr_steadystate}
\end{figure}

\subsection{Excitation number}

We point out that the covariance matrix $V(t)$, and hence the entanglement, does not depend on the power of the lasers.
This is a consequence of the dipole-dipole coupling and classical treatment of the driving field.
Therefore, the system gets entangled also in the absence of the lasers.
There is no fundamental reason why this entanglement with vacuum could not be measured, but practically it is preferable to pump the cavity in order to improve the signal-to-noise ratio.
Of course quantities other than entanglement may depend on driving power, for example the light intensity inside the cavity.

This finding is quite different from results in optomechanical system where the covariance matrix depends on laser power \cite{paternostro,krisnanda2017}. 
The origin of this difference is the nature of the coupling. 
For example, in an optomechanical system consisting of a single cavity mode $A$ and a mechanical mirror $B$ the coupling is proportional to $a^{\dagger}a\bm{X}_B$, which is a third-order operator \cite{optmech1}. 
This results in the effective coupling strength being proportional to the classical cavity field intensity $\alpha$ after linearisation of the Langevin equations. 
This classical signal enters the covariance matrix via the effective coupling strength and introduces the dependence on the driving power. 

In order to justify the low atomic excitation limit we first note that the number of steady-state photons for the $m\text{th}$ cavity mode without the presence of the bacteria is given by $\bm{E}_m^2 / \kappa_m^2 \propto \bm{P}_m$. 
When one considers the bacteria in the cavity having the base interaction strength $\tilde G_n$ and a decay rate $\gamma_n$ in the same order as the cavity decay rate, the number of excitations of the bacterial modes would also be in the order of $\bm{E}_m^2 / \kappa_m^2 $, which in our case is $10^{3}$. 
With $\sim10^8$ actively coupled dipoles in the cavity \cite{bacteria-exp}, this gives $\sim10^{-3}\%$ excitation, which justifies the low-excitation approximation.
We also plotted the evolution of excitation numbers of the bacterial modes (together with the number of photons in different cavity modes) within our model in Fig.~\ref{FIG_pr_number}. 
It shows that excitation numbers are oscillating in the ``steady state". 
The oscillations are caused by the combination of interactions between the light and bacteria (Rabi-like oscillations) and the time-dependent driving laser.
Setting the interactions $G_{mn}=0$ or the driving off ($\bm{P}_m=0$) indeed produces constant steady-state value.
We observe that the excitation number of the bacterial system is always below $2000$, which is in agreement with the statement above.

The excitation number of the cavity modes and bacteria as a function of time can be calculated from $\langle u_i(t)\rangle$ and $V_{ii}(t)$. For example, the mean excitation number for the first cavity mode is given by
\begin{equation}
\bar N_{1}(t)=\langle a_1^{\dagger}(t)a_1(t)\rangle =\frac12 ( V_{11}(t)+V_{22}(t) +\langle u_1(t)\rangle^2+\langle u_2(t)\rangle^2-1). 
\end{equation}

We present the evolution of photon number of the cavity modes and excitation of the bacterial modes in Fig. \ref{FIG_pr_number}. Note that photon number of the third cavity mode (solid magenta line) is showing oscillations well below its ``off-interaction" value (dashed magenta line).
This is because $\omega_3$ is almost in resonance with the frequency of the atomic transition $\Omega_{\scriptsize \mbox{II}}$.

\begin{figure}[h]
\centering
\includegraphics[scale=0.55]{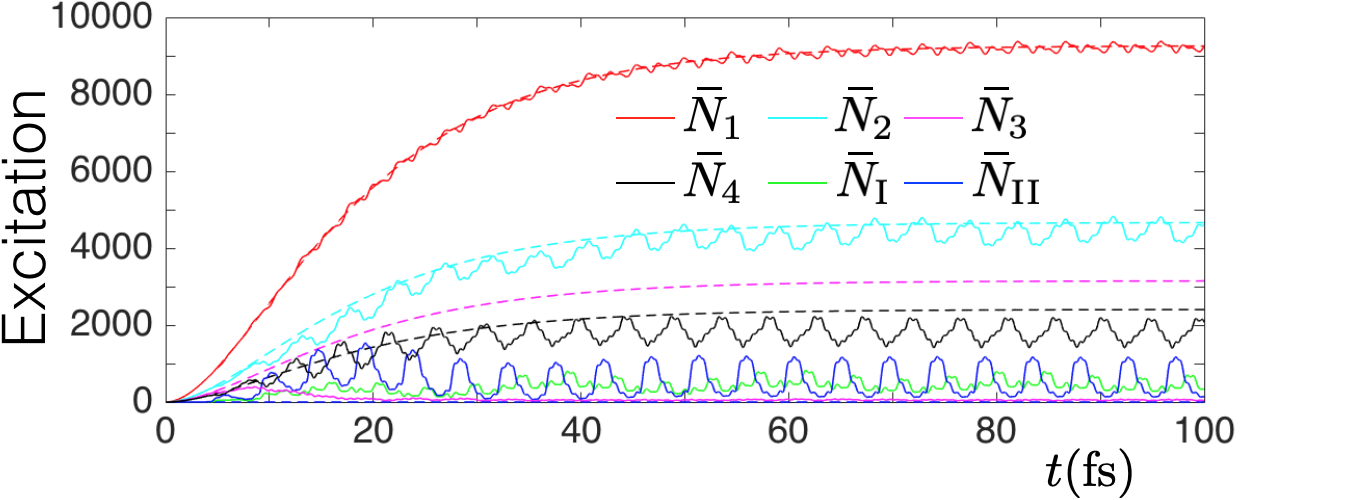}
\caption{Dynamics of the bacterial excitation numbers and photon number of the cavity modes.
Solid curves represent the interacting system: $\bar N_1$, $\bar N_2$, $\bar N_3$, $\bar N_4$ for the cavity modes and $\bar N_{\scriptsize\mbox{I}},\bar N_{\scriptsize\mbox{II}}$ for the bacterial system.
The couplings are taken as $\tilde G_{\scriptsize\mbox{I}},\tilde G_{\scriptsize\mbox{II}}=(3.9,6)\times 10^{13}$~Hz.
The corresponding dynamics for non-interacting system, i.e., $G_{mn}=0$, are given by the dashed curves. 
}
\label{FIG_pr_number} 
\end{figure}

\section{Entanglement as a witness of ultra-strong coupling}

We also performed similar calculations in which we neglected the counterrotating terms in Eq. (\ref{EQ_H}), the model known as Tavis-Cummings. 
This resulted in no entanglement generated in the steady state and can be intuitively understood as follows.
Since the steady-state covariance matrix does not depend on the initial state and on the power of the driving lasers, we might start with all atoms in the ground state, vacuum for the light modes, and no driving.
Under such circumstances there is no interaction between bacterial modes and light modes as every term in the interaction Hamiltonian contains an annihilation operator.
In physical terms, since we begin with the lowest energy state and the interaction Hamiltonian preserves energy, the ground state will be the state of affairs at any time.
Therefore, nonzero entanglement observed in experiments will provide evidence of the existence of the counterrotating terms, showing a signature of the ultra-strong coupling regime.

\section{Summary}
Recent experimental studies have shown that the regime of strong coupling between photosynthetic organisms such as green sulphur bacteria and light is accessible.
In this chapter, we put forward a proposal for probing a quantum property of the bacteria (as characterised by quantum correlations) without directly measuring them. 
Our general proposal bypasses the need of modelling both the bacteria and the interactions they have with external world.
The proposed scheme utilises the bacteria as mediators between two groups of non-interacting cavity light modes. 
In order to show feasibility, we modelled the light-sensitive part of the bacteria and their interactions with the cavity modes and environments using recent experimental parameters.
Within this model, our simulations show that, in most cases, quantum entanglement endures the environmental noises and remains present in the steady state.
Our simulations also confirmed that entanglement between the cavity light modes reveals a quantum property of the bacteria, i.e., quantum entanglement between the bacterial modes and cavity modes, and for stronger interactions, even quantum entanglement within the bacterial modes.
We note that the presence of entanglement also provides independent evidence of the ultra-strong coupling regime.



\chapter{Other applications} 

\label{Chapter7} 

\lhead{Chapter 7. \emph{Other applications}} 

\emph{In this chapter, I will present other applications that result from the principles of distribution of correlations between two objects whose interactions are mediated by an ancillary system -- investigated in Chapters~\ref{Chapter_revealing} to \ref{Chapter_speedup}.\footnote{Parts of this chapter are reproduced from our published articles of Refs.~\cite{krisnanda2017,nondecompos}, $\copyright$ [2019] American Physical Society. Where applicable, changes made will be indicated.}
First, I will apply our non-classicality detection method to an experimentally relevant platform of optomechanics, in particular, the membrane-in-the-middle setting.
Our protocol indirectly detects the quantum property of the mechanical membrane by observing the dynamics of entanglement between probing light modes.
Next, I present a protocol to reveal the correlation between a system and its environment in a scenario of open system dynamics. 
Then, high gain of correlations will be demonstrated in a dynamics of two cavity fields interacting with a two-level atom. 
This gain will be shown to imply a quantum trait -- non-commutativity of interaction Hamiltonians.
Finally, a brief discussion will be presented that shows an application of our method to estimate dimensionality of a quantum object.
}

\clearpage
\section{Membrane-in-the-middle optomechanics}

In this section, we address the practical implication of our criteria from Chapter~\ref{Chapter_revealing} to experiments of cavity optomechanics~\cite{RMP.86.1391}. This is a paradigmatic open mesoscopic quantum system for which the non-classicality detection method in Chapter~\ref{Chapter_revealing} holds the potential to be practically significant. 
In fact, one of the goals of optomechanics is to infer the non-classicality of the state of a massive mechanical system, in similar spirit as ``certification" in Refs. \cite{certification1,certification2}, without affecting its (in general fragile) state. 
A possible setting for such a task is given by the so-called membrane-in-the-middle configuration, where a mechanical oscillator (\emph{a membrane}) is suspended at the centre of a two-sided optical cavity~\cite{paternostro}.
By driving the cavity with laser fields from both its input mirrors, respectively, we realise a situation completely analogous to that in Fig.~\ref{FIG_re_ABC} (cf. Fig.~\ref{FIG_re_optmech}).
We now show that our scheme detects non-classicality of the membrane without measuring it.

\subsection{Experimental setup}


The hamiltonian of our setup (Fig.~\ref{FIG_re_optmech} below) in a rotating frame with frequency of the lasers can be written as $H=H_{\text{loc}}+H_{\text{int}}$ where \cite{paternostro}:
\begin{eqnarray}\label{Eq_hloc}
H_{\text{loc}}&=&\hbar \Delta_{0A} a^{\dagger} a+\hbar \Delta_{0B} b^{\dagger} b+\frac{\hbar \omega_C}{2}( \bm{P}_C^2+ \bm{X}_C^2) \nonumber \\ 
&&+i\hbar \bm{E}_A(a^{\dagger}- a)+i\hbar \bm{E}_B( b^{\dagger}- b) 
\end{eqnarray}
and 
\begin{equation}\label{Eq_hint}
H_{\text{int}}=-\hbar G_{0A} a^{\dagger}a \:\bm{X}_C+\hbar G_{0B}  b^{\dagger} b \:\bm{X}_C,
\end{equation}
where the annihilation (creation) operator of field $J=A,B$ is denoted by the corresponding lowercase letter $j$ ($ j^{\dagger}$) with $[ j, j^{\dagger}]=1$, 
$\bm{P}_C$ and $\bm{X}_C$ are dimensionless quadratures of the membrane with $[\bm{X}_C, \bm{P}_C]=i$, 
$\bm{E}_J$ is the driving strength of laser $J$ with $|\bm{E}_J|=(2\bm{P}_J\kappa_J/\hbar \omega_{lJ})^{1/2}$, where $\bm{P}_J$ is the laser power and $\omega_{lJ}$ denotes its frequency. $\kappa_J=\pi c/2\mathcal{F}_{\text{i},J}l_J$ is decay rate of cavity $J$ with finesse $\mathcal{F}_{\text{i},J}$.
Cavity-laser detuning is defined as $\Delta_{0J}\equiv \omega_J-\omega_{lJ}$, where $\omega_J$ is the frequency of the cavity
and $G_{0J}= (\omega_J/l_J)(\hbar/m_C \omega_C)^{1/2}$ represents field-membrane coupling strength, 
where $l_J$ is the length of the cavity, $m_C$ is the mass of the membrane and $\omega_C$ is its natural frequency.
Note that $H_{\text{loc}}$ is local in $A:B:C$ partition and the two terms in $H_{\text{int}}$ represent coupling in the partition $A:C$ and $B:C$ respectively. 
All the other interactions are local, i.e., $A$ ($B$) is coupled to its own environment $A'$ ($B'$) and $C$ is coupled to its thermal phonon reservoir $C'$, responsible for the Brownian motion of the membrane.
Thus, our Theorem~\ref{TH_revealing} directly applies here and we can implement the detection method of Fig.~\ref{FIG_re_method}a.

\begin{figure}[h]
\centering
\includegraphics[scale=0.45]{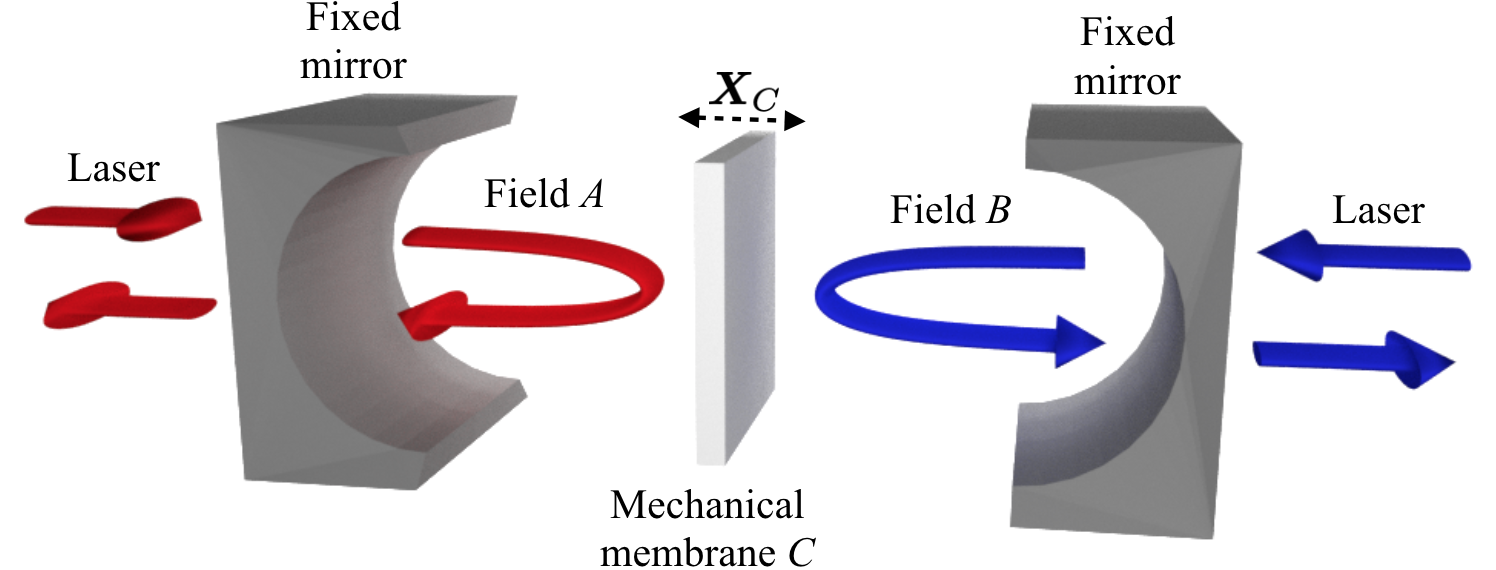} 
\caption{Membrane-in-the-middle optomechanical setup.
Cavity field $A$ is interacting indirectly with cavity field $B$ via a perfectly reflecting mechanical membrane $C$.
The interactions of the movable mechanical mirror with its thermal environment result in the Brownian noise.
The cavities are driven by lasers and experience energy dissipation through the fixed input mirrors.
}
\label{FIG_re_optmech} 
\end{figure}

\subsection{Dynamics of the system}

\subsubsection{Langevin equations}

The dynamics of the operators, adding into account noise and damping terms (also local), can be well written by a set of Langevin equations in Heisenberg picture
\begin{eqnarray}\label{Eq_lgvn}
\dot a&=&-(\kappa_A+i\Delta_{0A})a+iG_{0A}a\bm{X}_C+\bm{E}_A+\sqrt{2\kappa_A}\:a_{\text{in}} \nonumber \\
\dot b&=&-(\kappa_B+i\Delta_{0B})b-iG_{0B}b\bm{X}_C+\bm{E}_B+\sqrt{2\kappa_B}\:b_{\text{in}} \nonumber \\
\dot {\bm{X}_C}&=&\omega_C \bm{P}_C \nonumber \\
\dot {\bm{P}_C}&=&-\omega_C \bm{X}_C +G_{0A}a^{\dagger}a- G_{0B}b^{\dagger}b-\gamma_C \bm{P}_C+\xi
\end{eqnarray}
where $\gamma_C$ is damping rate of the membrane. Also $j_{\text{in}}$ is input noise of field $J$ associated with cavity-input mirror interface and has only correlation function\\ $\langle j_{\text{in}}(t)k_{\text{in}}^{\dagger}(t^{\prime}) \rangle=\delta_{jk}\delta(t-t^{\prime})$ \cite{walls2007quantum}, whereas $\xi$ is Brownian noise of the membrane and has correlation function $\langle \xi(t) \xi(t^{\prime})+\xi(t^{\prime})\xi(t)\rangle/2\approx \gamma_C(2\bar n+1)\delta(t-t^{\prime})$ in the limit of interest that is large mechanical quality of the membrane, i.e., $\omega_C/\gamma_C\gg 1$ \cite{giovannetti2001phase,benguria1981quantum}. The mean phonon number of the membrane reads $\bar n=1/(\exp{(\hbar \omega_C/k_BT)}-1)$.

The linearised Langevin equations can be obtained by splitting the operators into steady state values and fluctuating terms. In particular we write $\bm{X}_C=\bm{X}_{Cs}+\delta \bm{X}_C$, $\bm{P}_C=\bm{P}_{Cs}+\delta \bm{X}_C$, and $j=\alpha_{Js}+\delta j$. By inserting these into Eq. (\ref{Eq_lgvn}) and ignoring nonlinear terms $\delta j^{\dagger}\delta j$ and $\delta j\: \delta \bm{X}_C$ one gets a set of linear Langevin equations for the fluctuations of the quadratures
\begin{eqnarray}\label{Eq_llgvn}
\delta \dot {\bm{X}_A}&=&-\kappa_A \delta \bm{X}_A+\Delta_A \delta \bm{Y}_A+\sqrt{2\kappa_A}\:x_{\text{in},A} \nonumber \\
\delta \dot {\bm{Y}_A}&=&-\kappa_A \delta \bm{Y}_A-\Delta_A \delta \bm{X}_A+G_A\delta \bm{X}_C+\sqrt{2\kappa_A}\:y_{\text{in},A} \nonumber \\
\delta \dot {\bm{X}_B}&=&-\kappa_B \delta \bm{X}_B+\Delta_B \delta \bm{Y}_B+\sqrt{2\kappa_B}\:x_{\text{in},B} \nonumber \\
\delta \dot {\bm{Y}_B}&=&-\kappa_B \delta \bm{Y}_B-\Delta_B \delta \bm{X}_B-G_B\delta \bm{X}_C+\sqrt{2\kappa_B}\:y_{\text{in},B} \nonumber \\
\delta \dot {\bm{X}_C}&=&\omega_C \delta \bm{P}_C \nonumber \\
\delta \dot {\bm{P}_C}&=&-\omega_C\delta \bm{X}_C-\gamma_C \delta \bm{P}_C+G_A\delta \bm{X}_A -G_B\delta \bm{X}_B+\xi 
\end{eqnarray}
where effective detuning $\Delta_A\equiv \Delta_{0A}-G_{0A}\bm{X}_{Cs}$, $\Delta_B\equiv \Delta_{0B}+G_{0B}\bm{X}_{Cs}$, and effective coupling $G_J\equiv \sqrt{2}G_{0J}\alpha_{Js}$. The steady state values are given by $\bm{P}_{Cs}=0$, $\bm{X}_{Cs}=(G_{0A}|\alpha_{As}|^2-G_{0B}|\alpha_{Bs}|^2)/\omega_C$, and $\alpha_{Js}=|\bm{E}_J|/\sqrt{\kappa_J^2+\Delta_J^2}$. The quadratures of the field $\bm{X}_J$ and $\bm{Y}_J$ are related to the field operator $j$ through $j=(\bm{X}_J+i\bm{Y}_J)/\sqrt{2}$. This relation also applies for the input noise, i.e., $j_{\text{in}}=(x_{\text{in},J}+iy_{\text{in},J})/\sqrt{2}$.

For simplicity one can re-write Eqs. (\ref{Eq_llgvn}) as a single matrix equation $\dot u(t)=Ku(t)+n(t)$ where the vector $u(t)=(\delta \bm{X}_A,\delta \bm{Y}_A,\delta \bm{X}_B,\delta \bm{Y}_B,\delta \bm{X}_C,\delta \bm{P}_C)^T$, 
\begin{equation}
n(t)=(\sqrt{2\kappa_A}\: x_{\text{in,A}},\sqrt{2\kappa_A}\: y_{\text{in,A}},\sqrt{2\kappa_B}\: x_{\text{in,B}},\sqrt{2\kappa_B}\: y_{\text{in,B}},0,\xi)^T,
\end{equation}
and
\begin{equation}
K=\left( \begin{array}{cccccc} -\kappa_A&\Delta_A&0 &0 &0 &0\\-\Delta_A&-\kappa_A&0 &0 &G_A &0\\0&0&-\kappa_B &\Delta_B &0 &0\\0&0&-\Delta_B &-\kappa_B &-G_B &0\\0&0&0 &0 &0&\omega_C\\G_A&0&-G_B &0 &-\omega_C &-\gamma_C \\ \end{array}\right).
\end{equation}
The solution to this linearised Langevin equation is then given by $u(t)=M(t)u(0)+\int_0^t ds M(s)n(t-s)$ where $M(t)=\exp{(Kt)}$.

\subsubsection{Dynamics of covariance matrix}

The quantum state of the fluctuations is fully characterised by covariance matrix $V_{ij}(t)\equiv \langle u_i(t)u_j(t)+u_j(t)u_i(t)\rangle/2-\langle u_i(t)\rangle \langle u_j(t)\rangle$. Note that the Gaussian nature of the initial state is maintained since we have linear dynamics and the noises involved are zero mean Gaussian noises. One can show that the covariance matrix at time $t$ is $V(t)=M(t)V(0)M^T(t)+\int_0^t ds\; M(s)DM^T(s)$ where $D=\mbox{Diag}[\kappa_A,\kappa_A,\kappa_B,\kappa_B,0,\gamma_C(2\bar n+1)]$. A more explicit solution of the covariance matrix, after integration, is given by
\begin{eqnarray}
KV(t)+V(t)K^T&=&-D+KM(t)V(0)M^T(t) \nonumber \\
&&+M(t)V(0)M^T(t)K^T \nonumber \\
&&+M(t)DM^T(t),
\end{eqnarray}
which is linear and can easily be solved numerically. 
For our simulations of the dynamics below, we take the initial state to be thermal state for $c$ and coherent state for field $j$, this gives $V(0)=\mbox{Diag}[1,1,1,1,2\bar n+1,2\bar n+1]/2$. If one is only interested in steady state solution, it is guaranteed when all real parts of eigenvalues of $K$ are negative, giving $M(\infty)=0$ such that the steady state covariance matrix can be calculated from a simpler equation $KV(t_{s})+V(t_{s})K^T=-D$.

\subsubsection{Quantum entanglement: dynamics and steady state}

In order to calculate entanglement, we utilise here the method presented in Chapter~\ref{Chapter1}. 
In particular, the covariance matrix $V$ describing our three-mode optomechanical system can be written in block form 
\begin{equation}
V_{ABC}=\left( \begin{array}{ccc} L_{AA}&L_{AB}&L_{AC}\\ L_{AB}^T&L_{BB}&L_{BC}\\ L_{AC}^T&L_{BC}^T&L_{CC} \end{array}\right)
\end{equation}
where for $j,k=A,B,C$ the block component $L_{jk}$ is a $2\times 2$ matrix describing local mode correlation when $j=k$ and intermodal correlation when $j\ne k$. An $N$-mode covariance matrix has symplectic eigenvalues $\{\nu_k\}_{k=1}^N$ that can be computed from the spectrum of matrix $|i\Omega_N V|$ \cite{weedbrook2012gaussian} where 
\begin{equation}
 \Omega_N=\bigoplus^N_{k=1} \left( \begin{array}{cc} 0&1\\ -1 &0\end{array}\right).
 \end{equation}
For a physical covariance matrix $2 \nu_k\ge 1$. For an entangled system, e.g., in the partition $AB:C$, the covariance matrix will not be physical after partial transposition with respect to mode $C$ (this is equivalent to flipping the sign of the membrane's momentum fluctuation operator $\delta \bm{P}_C$ in $V$). For our system, this unphysical $V^{T_C}$ is shown by one of its three symplectic eigenvalues $\tilde \nu_{\text{min}}<1/2$. Entanglement is then quantified by logarithmic negativity as follows $E_{AB:C}=\mbox{max}[0,-\log_2{(2\tilde \nu_{\text{min}})}]$ \cite{negativity, adesso2004extremal}.\footnote{Note that we have used logarithm with base 2, instead of natural logarithm as in Ref.~\cite{krisnanda2017}, to maintain consistencies in this thesis.}
Note that the separability condition, when $V^{T_C}$ has $\tilde \nu_{\text{min}}\ge1/2$, is sufficient and necessary for $1:N$ mode partition \cite{werner2001bound}. Entanglement $E_{A:B}$ is calculated in similar manner by only considering system $AB$ where the covariance matrix is now
\begin{equation}
V_{AB}=\left( \begin{array}{cc} L_{AA}&L_{AB}\\ L_{AB}^T&L_{BB}\\  \end{array}\right).
\end{equation}

In order to independently confirm the non-classicality of the membrane and demonstrate that there is considerable entanglement to be detected we calculate the ensuing entanglement dynamics.
As mentioned above, we start with the experimentally natural state where $C$ is in a thermal state and $A$ and $B$ are coherent states, and calculate the dynamics of entanglement $E_{A:B}$ and $E_{AB:C}$.
Since initially there is no entanglement, the first step in Fig.~\ref{FIG_re_method}a can be omitted.

The parameters used in our simulations all adhere to present-day technology \cite{groblacher2009observation}.
This includes $m_C=145 \: \mbox{ng}$, $T=300\: \mbox{mK}$, $\l_J=25\: \mbox{mm}$, and $(\omega_C, \omega_{lJ}, \gamma_C)=2\pi(947\times 10^3, 2.8\times10^{14}, 140)\: \mbox{Hz}$. Finesse of each cavity is $1.4\times10^4$.
The results of our analysis are presented in Fig.~\ref{FIG_re_evstime} for varying power of the right laser. 
We see that non-zero $E_{A:B}(\tau)$ is always accompanied by non-zero $E_{AB:C}$ at some time $(0,\tau)$.
Note that entanglement is a stronger type of quantum correlations than discord.
We have also performed similar calculations by varying the power of the left laser as well as the frequencies of the lasers within experimentally accessible ranges and observed consistent results.

\begin{figure}[h]
\centering
\includegraphics[scale=0.6]{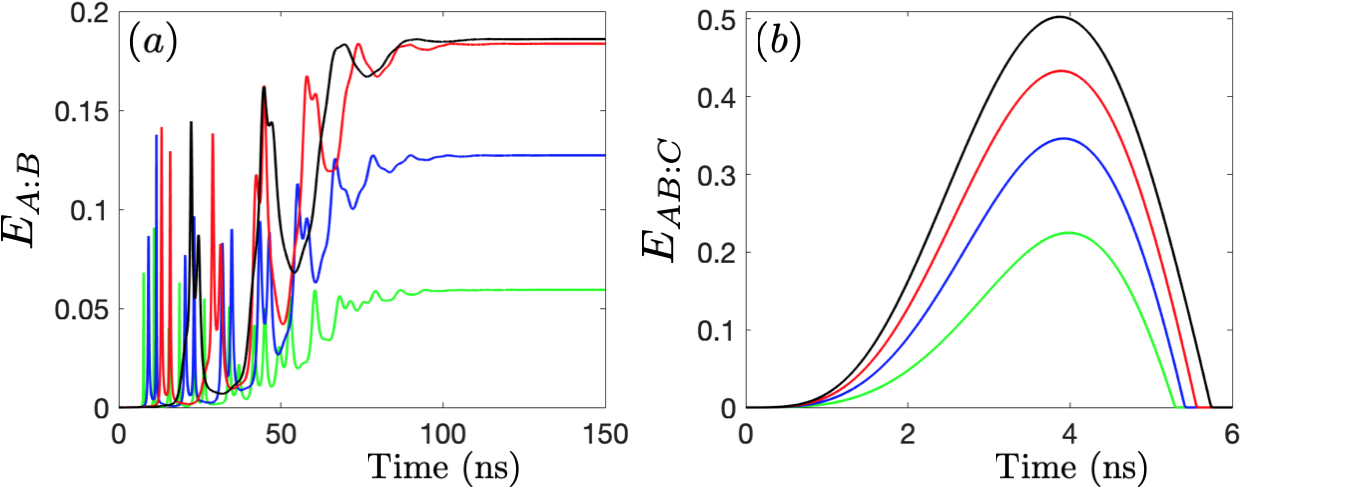}
\caption{Exemplary dynamics of entanglement in the membrane-in-the-middle optomechanical setup.
We take the feasible parameters from recent experiments by Gr\"{o}blacher \emph{et al}.~\cite{groblacher2009observation}.
In particular, the mirror has mass $m_C=145$ ng, damping rate $\gamma_C=2\pi \times 140$~Hz, natural frequency $\omega_C=2\pi \times947$~kHz, and environmental temperature $T=0.3$~K. 
Each cavity has length $25$~mm, finesse $1.4\times10^4$, and is driven by $1064$~nm laser. 
We fixed the power of the left laser to be $\bm{P}_A=100$~mW, and the detunings $\Delta_A=\omega_C$ and $\Delta_B=-\omega_C$, while varying the power of the right laser as $\bm{P}_B=20$~mW (green curves), $40$~mW (blue curves), $60$~mW (red curves), and $80$~mW (black curves).
Our simulations show that positive entanglement $E_{A:B}$ implies that the mirror is entangled with the cavity fields, i.e., nonzero $E_{AB:C}$ during the evolution.
}
\label{FIG_re_evstime} 
\end{figure}

If one is interested only in the steady state regime, Fig.~\ref{FIG_re_ess} shows the corresponding entanglement $E_{A:B}$ while $E_{AB:C}$ is zero in this range (not shown). Note that red colour has been used in the plots for parameters that do not correspond to steady state solution.

\begin{figure}[h]
\centering
\includegraphics[scale=0.45]{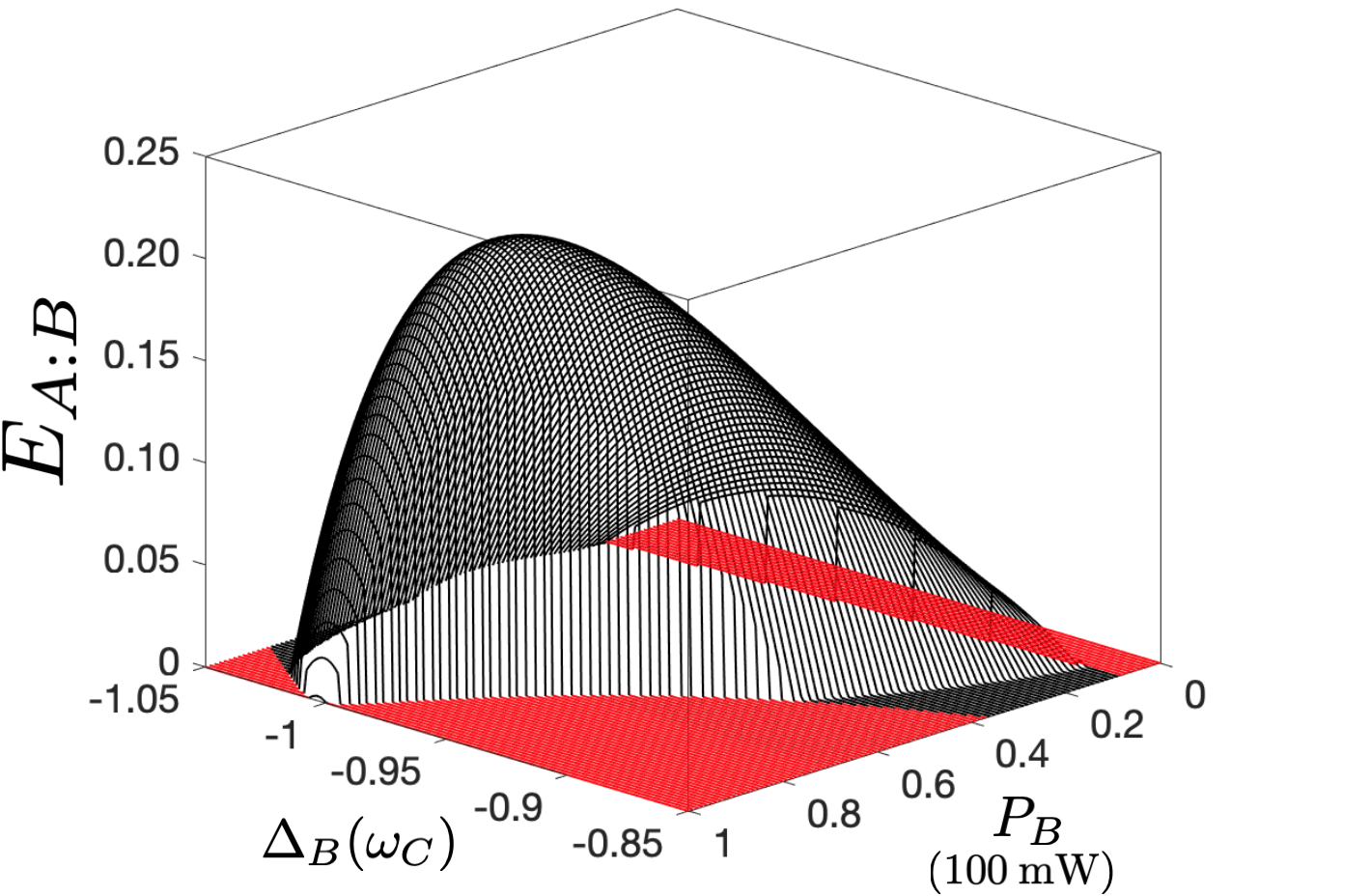}
\caption{Exemplary steady state entanglement between the cavity fields in the membrane-in-the-middle optomechanical setup.
The power of the right laser is varied as $\bm{P}_B=[0,100]$~mW and the detuning $\Delta_B=[-1.05,-0.85]\: \omega_C$.
Other parameters are the same as those stated in Fig. \ref{FIG_re_evstime}.
Red colour region corresponds to non-steady state regime. 
}
\label{FIG_re_ess} 
\end{figure}


\section{Detecting system-environment correlations}

In this section we present a brief review of schemes for detecting system-environment correlations.
We show that our protocol is capable of such detection without the need for state tomography.

\subsection{Previous detection schemes}
A vast body of literature exists on the study of the influence of initial system-environment correlations (SECs) on the evolution of the open system~\cite{breuer2002theory}.
Proposals for the detection of SECs based on monitoring the dynamics of distinguishability \cite{laine2011witness, smirne2010initial, dajka2010distance, dajka2011distance, wissmann2013detecting} or purity \cite{kimura2007,rossatto2011purity} of the open system have been put forward.
Such proposals have been implemented experimentally by means of quantum tomography \cite{smirne2011experimental, li2011experimentally}.
Moreover, the possible non-classical nature of SECs was linked to the impossibility of describing the evolution of an open system through completely positive maps~\cite{rodriguez2008completely}.
Hence detection schemes of quantum discord in the initial system-environment state have been proposed~\cite{gessner2011detecting, gessner2013local} 
and recently assessed experimentally \cite{gessner2014local, tang2015experimental, cialdi2014two}. 

\subsection{Our protocol}

In order to apply our method in this scenario, let us consider again a closed-system dynamics and, in line with the assumed inaccessibility of the mediator, focus the attention to the probes only (see Fig.~\ref{FIG_re_ABC}).
We could thus think of $C$ as an environment in contact with the open system $AB$, as seen in Fig.~\ref{FIG_oa_sec}.

\begin{figure}[h]
\centering
\includegraphics[scale=0.5]{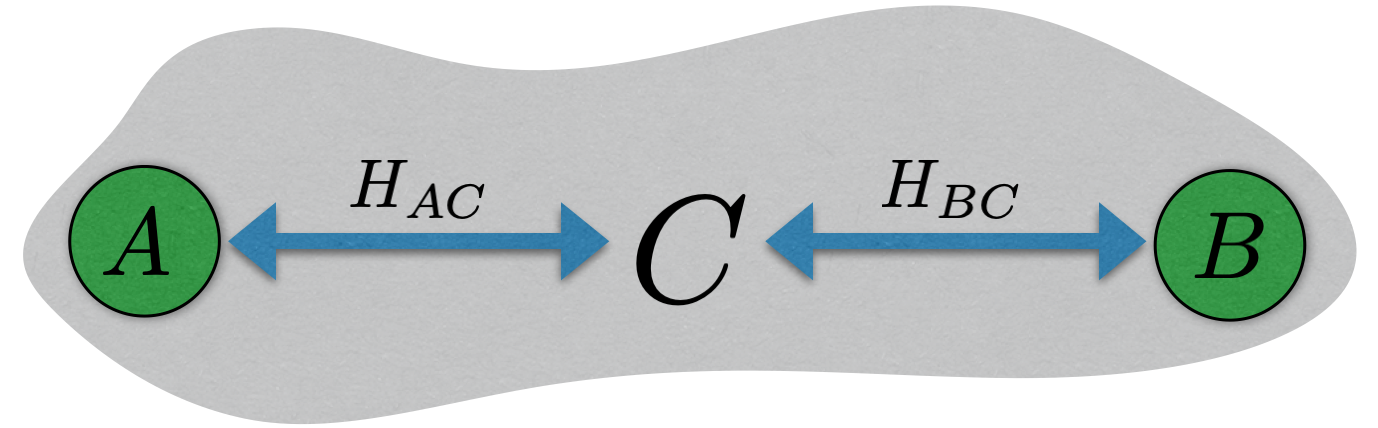}
\caption{A scenario to reveal system-environment correlations. 
Two mutually non-interacting objects $A$ and $B$ are in contact with a common environment $C$.
We show that correlations between the open system $AB$ and the environment can be inferred from the gain of entanglement between $A$ and $B$.
Note that the setup considered here is different from the one in Fig.~\ref{FIG_re_ABC} of Chapter~\ref{Chapter_revealing}, where now $C$ includes local or global environments or mediating objects.
}
\label{FIG_oa_sec} 
\end{figure}

Now we use our scheme of Fig.~\ref{FIG_re_method}a to reveal SECs, with the advantage that the applications of an entanglement breaking channel and state tomography are not necessary. 
This is a consequence of Lemma~\ref{LM_revealing1} that, if rephrased in terms of the aim in this section, states that quantum entanglement in the open system between two objects $A$ and $B$ cannot increase if the environment is not correlated with them at all times.
Note that the amount of entanglement can be quantified by any entanglement monotone.

Recall also that our detection protocol of Fig.~\ref{FIG_re_method}a is used to infer non-classicality of correlations, i.e., quantum discord, between the environment and the open system.
We note that previous schemes detect the non-classicality of the system \cite{gessner2011detecting, gessner2013local}, i.e., presence of $D_{C|AB}$, whereas our schemes ascertain the non-classicality of the environment, $D_{AB|C}$,
which is perhaps a prime example of an inaccessible object.

\section{Two cavity fields coupled via a two-level atom}\label{SC_faf}

In this section, we demonstrate, with concrete dynamics generated by non-commuting Hamiltonians, that the correlation capacity bounds derived in Chapter~\ref{Chapter_detecting} can be violated.
We also provide similar calculations for commuting Hamiltonians.

\subsection{Witnessing non-commutativity and discord}

Consider a two-level atom $C$, i.e., $d_C=2$, mediating interactions between two cavity fields $A$ and $B$ as illustrated in Fig.~\ref{FIG_de_faf}. 
A similar scenario has been considered and implemented, for example, in Refs. \cite{rauschenbeutel2001,messina2002,browne2003,hamsen}.
The interaction between the atom and each cavity field is taken to follow the Jaynes-Cummings model,
\begin{equation}
H = \hbar g ( a \sigma^+ + a^{\dagger} \sigma^-)+\hbar g ( b \sigma^+ + b^{\dagger} \sigma^-),
\label{EQ_JC}
\end{equation}
where $ a$ ($ b$) is the annihilation operator of field $A$ ($B$), while $ \sigma^+$ ($ \sigma^-$) is the raising (lowering) operator of the two-level atom.
For simplicity, we have assumed that the interaction strengths between the two-level atom and the fields are the same.
Note that $H$ is of the form $H_{AC}+H_{BC}$ with non-commuting components.

\begin{figure}[h]
\centering
\includegraphics[scale=0.6]{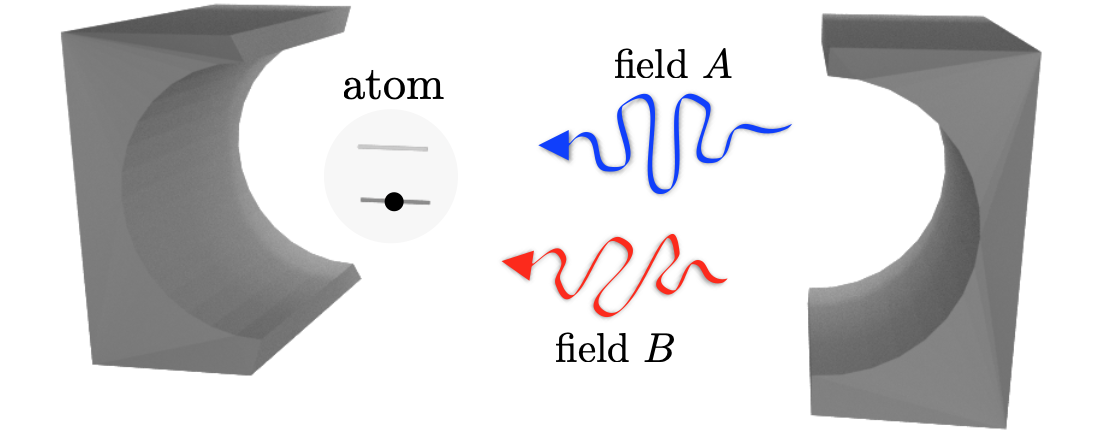}
\caption{Two orthogonally polarised cavity fields interact via a two-level atom, but not with each other.
In the text, we show that correlation gain between $A$ and $B$ is linked with non-commutativity of interaction Hamiltonians, or more generally, non-decomposability of time evolution.
}
\label{FIG_de_faf}
\end{figure}

The resulting correlation dynamics are plotted in Fig. \ref{FIG_de_dynamics}. 
Mutual information and negativity were calculated directly, whereas for the classical correlation and the relative entropy of discord, we provide the lower bounds $\tilde C_{A:B}$ and $-S_{A|B}$, respectively.
$\tilde C_{A:B}$ is calculated as the mutual information of the state resulting from projective local measurements in the Fock basis (no optimization over measurements performed).
The negative conditional entropy $-S_{A|B}$ is a lower bound on the distillable entanglement \cite{devetak2005distillation}, which in turn is a lower bound on the relative entropy of entanglement $E_{A:B}$ \cite{horodecki2000limits}. Therefore, we note the chain of inequalities $-S_{A|B} \le E_{A:B} \le D_{A|B} \le I_{A:B}$, where the last two inequalities follow from \cite{modi2010unified}.
Already these lower bounds can beat the limit set by decomposable evolution, and therefore, all mentioned correlations can detect non-decomposability of the evolution. 
Since we consider closed systems, this infers non-commutativity of the Jaynes-Cummings couplings.
We also note another non-classical feature of the studied dynamics: since Fig. \ref{FIG_de_dynamics} shows entanglement gain, according to Chapter~\ref{Chapter_revealing} there must be quantum discord $D_{AB|C}$ during the evolution.

It is apparent that the detection is easier (faster and with more pronounced violation) with a higher number of photons in the initial states of the cavity fields.
We offer an intuitive explanation.
Consider, for example, $| m n 0 \rangle$ as the initial state of $ABC$.
By defining $ \xi = ( a + b )/\sqrt{2}$, the Hamiltonian of Eq. (\ref{EQ_JC}) becomes $\sqrt{2} \hbar g( \xi  \sigma^+ +  \xi^{\dagger} \sigma^-)$ and it is straightforward to obtain the unitary evolution~\cite{scully-book}.
One finds that the quantum state of the fields oscillates incoherently between $\sum^{m+n}_{j=0}c_j(t) |j\rangle_A |m+n - j\rangle_B$ and $\sum^{m+n-1}_{j=0}d_j(t) |j\rangle_A |m+n-1-j\rangle_B$.
Both of these states are superpositions of essentially $m+n$ bi-orthogonal terms giving rise to high entanglement and, therefore, also other forms of correlations.

Figure~\ref{FIG_de_dynamics} illustrates that different correlation quantifiers have different detection capabilities and it is not clear at this stage whether there is a universal measure with which non-commutativity is detected, e.g., the fastest.
For most initial states we studied mutual information detected non-commutativity the most rapidly, but there are exceptions, as shown by the black curve corresponding to the initial state $| 1 0 1\rangle$.
With this initial state the mutual information never violates its bound, but the negativity does.

\begin{figure}[h]
\centering
\includegraphics[scale=0.4]{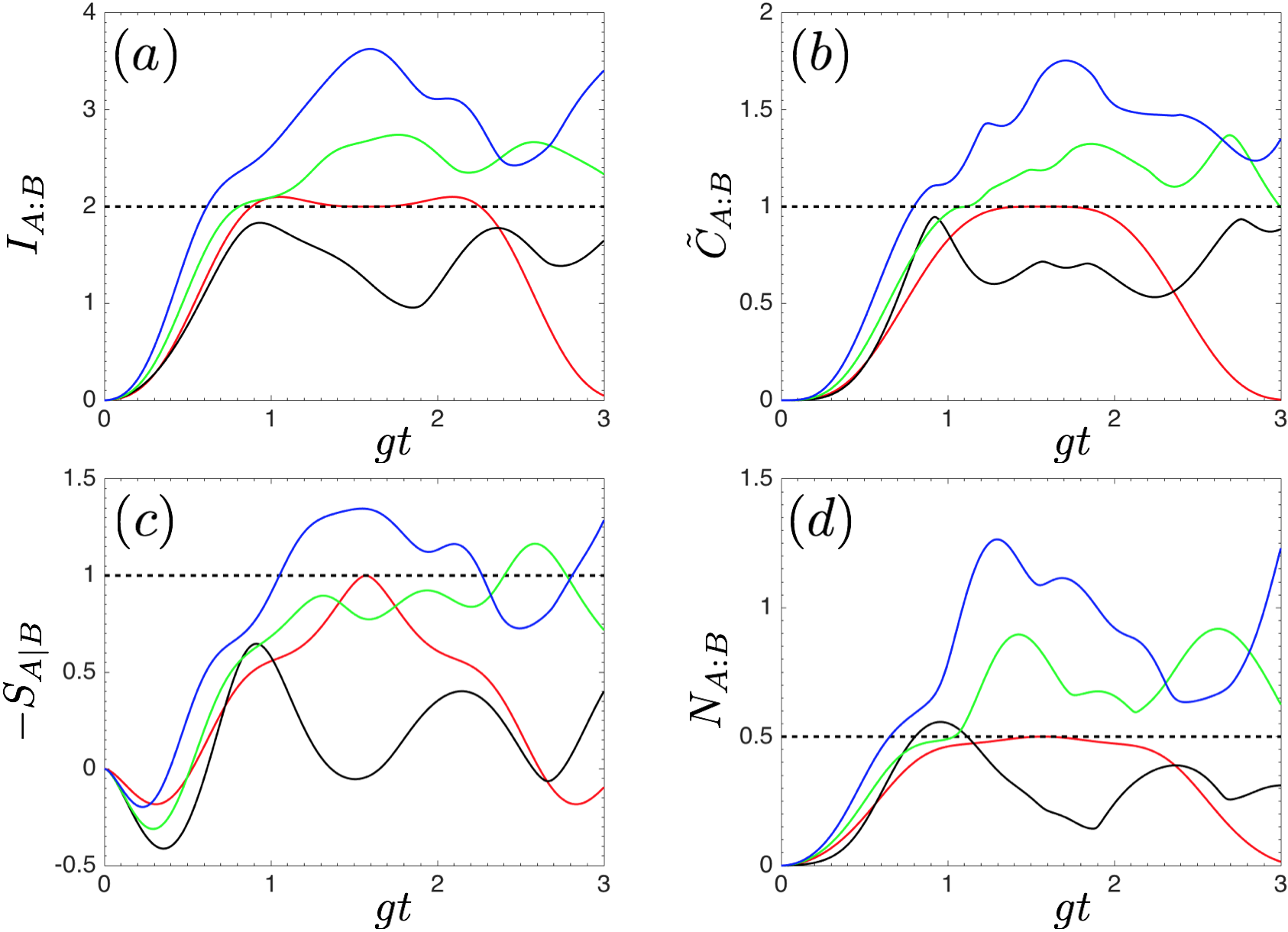}
\caption{Correlation dynamics in the field-atom-field system.
The solid curves represent the Jaynes-Cummings coupling while the dashed lines show the corresponding correlation bound for decomposable dynamics.
In the panels, we show correlations between field $A$ and field $B$: (a) the mutual information, (b) lower bound on the classical correlation, (c) lower bound on the relative entropy of discord, and (d) the entanglement as quantified by negativity.
The varied parameter is the initial state of the tripartite system $ABC$: $\ket{110}$ (red curves), $\ket{101}$ (black curves), $\ket{210}$ (green curves), and $\ket{220}$ (blue curves).
Note that the horizontal axis is given by the dimensionless time $gt$.
}
\label{FIG_de_dynamics}
\end{figure}

\subsection{Strong interactions and bounded entanglement}

Our results imply that the non-commutativity (non-decomposability in general) is a desired feature of interactions in the task of correlation distribution, which is important for quantum information applications.
As a contrasting physical illustration, we consider the strong dipole-dipole interactions in our field-atom-field example.
The Hamiltonian reads
\begin{equation}
H^{\prime}=\hbar g ( a+ a^{\dagger}) ( \sigma^+ +  \sigma^-)+\hbar g ( b+ b^{\dagger})( \sigma^+ +  \sigma^-),
\label{EQ_DD}
\end{equation}
with commuting components, i.e., $[H_{AC},H_{BC}]=0$.
One can verify that the results of this model are in agreement with all the bounds we derived in Chapter~\ref{Chapter_detecting}.
Furthermore, we prove below that, with this coupling, the state of $AB$ at time $t$ is effectively given by a two-qubit separable state. This makes $N_{A:B}(t)=0$ and $I_{A:B}(t)\le 1$. 
Note the counter-intuitive result that strong interactions produce bounded correlations between the probes, while weak interactions (Jaynes-Cummings coupling) can increase the correlations above the bounds. 

Let us define ${\xi} = ( a + b )/\sqrt{2}$. 
The dipole-dipole Hamiltonian, Eq. (\ref{EQ_DD}), is reformulated as $H^{\prime}=\sqrt{2}\hbar g( \xi + \xi^{\dagger}) \sigma^x$, where $ \sigma^x= \sigma^+ + \sigma^-$ and $[ \xi,  \xi^{\dagger}]=\openone$. 
The unitary evolution operator is given by
\begin{eqnarray}\label{EQ_hddunitary}
 U_t&=&e^{-\frac{iH^{\prime}t}{\hbar}} \\
&=&\frac{1}{2} [(\openone- \sigma^x)e^{i\sqrt{2} gt( \xi +  \xi^{\dagger})}+(\openone+ \sigma^x)e^{-i\sqrt{2} gt( \xi +  \xi^{\dagger})}]\nonumber \\
&=&\frac{1}{2} [(\openone- \sigma^x) D_a(\alpha) D_b(\alpha) +(\openone+ \sigma^x) D_a(-\alpha) D_b(-\alpha)], \nonumber
\end{eqnarray}
where $\alpha=igt$ and, e.g., $ D_a (\alpha)=\exp(\alpha a^{\dagger}-\alpha^{\ast} a)$. 
Given an initial state $|mn0\rangle$, the state at time $t$ reads
\begin{eqnarray}
| \psi_t \rangle & = & \frac{1}{4} [ (d^{(mn)}_{++}|D^{(m)}_+,D^{(n)}_+\rangle + d^{(mn)}_{--}|D^{(m)}_-,D^{(n)}_-\rangle )|0\rangle \nonumber \\
& - & (d^{(mn)}_{+-}|D^{(m)}_+,D^{(n)}_-\rangle+d^{(mn)}_{-+}|D^{(m)}_-,D^{(n)}_+\rangle)|1\rangle], \nonumber
\end{eqnarray}
where
\begin{eqnarray}
d^{(mn)}_{\pm\pm} & = & 2\sqrt{[1\pm e^{-2|\alpha|^2} L_m(4|\alpha|^2)][1\pm e^{-2|\alpha|^2} L_n(4|\alpha|^2)]},\nonumber \\
|D^{(n)}_{\pm}\rangle & = & \frac{1}{\sqrt{d^{(nn)}_{\pm\pm}}}[{D}(\alpha)\pm{D}(-\alpha)]|n\rangle . \nonumber
\end{eqnarray}
Note that $\langle D^{(n)}_+|D^{(n)}_-\rangle=0$ and $\langle D^{(n)}_{\pm}|D^{(n)}_{\pm}\rangle=1$. $L_n(|\alpha|^2)$ is the Laguerre polynomial, which comes from the relation
$\langle n|{D}(\alpha)|n\rangle=e^{-|\alpha|^2/2}L_n(|\alpha|^2)$.
After tracing-out of the atomic mode $C$, the state of the fields is effectively given by a two-qubit state,
\begin{eqnarray}
\frac{1}{16}
\begin{pmatrix}
(d^{(mn)}_{++})^2 & 0& 0& d^{(mn)}_{++}d^{(mn)}_{--} \\
0 & (d^{(mn)}_{+-})^2 & d^{(mn)}_{+-}d^{(mn)}_{-+} & 0 \\
0 & d^{(mn)}_{+-}d^{(mn)}_{-+} & (d^{(mn)}_{-+})^2 & 0 \\
d^{(mn)}_{++}d^{(mn)}_{--} & 0 & 0 & (d^{(mn)}_{--})^2 \nonumber
\end{pmatrix},\\
\end{eqnarray}
which is positive under partial transposition and, hence, separable \cite{peres1996,horodecki1996m}. The same result follows for initial state $\ket{mn1}$.

\subsection{Estimating dimension of mediators}

Last but not least, we note an application of our correlation capacity bounds to estimate the dimension of the mediator; see, e.g., Refs.~\cite{dimwit1,dimwit2,dimwit3} for other dimension witnesses.
For decomposable evolution (including discrete sequential operators considered in Refs.~\cite{cubitt,bounds1,bounds2,exp0,edssexp1,edssexp2,edssexp3}), the amount of correlation between the probes is bounded by the correlation capacity $\sup_{\ket{\psi}} Q_{A:C}$, which is a function of $d_C$.
If one observes $Q_{A:B}(t)$ value that is larger than the correlation capacity of a certain $d_C$, then the dimension of the mediator must be larger than $d_C$.


\chapter{Conclusion and future work} 

\label{Chapter_conclusion} 

\lhead{Chapter 7. \emph{Conclusion and future work}} 

\emph{In this chapter, I present the conclusion of the work that has been reported in this thesis.
Additionally, some immediate questions arise, which have not yet been settled.
I will discuss these key questions with preliminary results, wherever possible.
In particular, this includes a natural question on the general entanglement bound for the indirect interaction scenario, a quantitative bound on the amount of non-commutativity of Hamiltonians (in general, non-decomposability of time evolution operator), a request on the protocol capable of witnessing the presence of quantum objects (i.e., the ones that can be given a description within the quantum framework such as living in a Hilbert space), and a strict bound on entangling time for initial states where the mediator is decoupled. 
}

\clearpage
\section{Conclusion}
We have studied the creation of correlations, mostly quantum entanglement, between two principal quantum objects that are continuously interacting via a mediator -- the indirect interaction setting.
The research in this thesis has been devoted to understanding the factors that are crucial in this scenario.
It includes the required property of the mediator, the amount of correlation gain, and the speed at which entanglement is created.
Our work also resulted in key applications ranging from the current technologically available platform detecting non-classicality of an optomechanical mirror to a proposal on revealing a quantum nature of gravity.

In particular, Chapter~\ref{Chapter_revealing} has shown that, for initially uncorrelated objects, quantum discord between the mediator and the principal objects is a necessary condition for distribution of quantum entanglement. 
The contrapositive of our theorem is therefore the revelation of quantum discord (a form of non-classical feature of the mediator) from the observation of entanglement gain between the principal objects.
Next, in Chapter~\ref{Chapter_detecting}, we have made a connection between the amount of distributed correlations and non-commutativity of Hamiltonians or, more generally, non-decomposability of the dynamical operator (another form of non-classicality).
The speed of entanglement creation has been investigated in Chapter~\ref{Chapter_speedup}. 
We have presented a lower bound on entangling time for both direct and indirect interaction settings.
We proved that entanglement cannot be indirectly distributed faster than the direct time bound. 
However, some examples, all of which require an initially correlated mediator, have been shown to saturate the bound.
We note that in Chapters~\ref{Chapter_revealing}--\ref{Chapter_speedup}, wherever possible, we have used minimalistic assumptions in order to make our theorems as general as possible.
This includes the relaxation of knowledge of the dimensionality of all the objects under scrutiny, the initial state of the system, and details of interactions. 
All the objects can also be open systems.

Applications resulted from this work have been reported in three chapters.
First, Chapter~\ref{Chapter_gravity} focussed on our proposal on the observation of quantum entanglement between two masses.
This has been argued as a potential route towards the detection of a quantum nature of gravitational interactions.
Next, Chapter~\ref{Chapter_probing} has shown another proposal on probing of non-classicality of photosynthetic organisms, in particular, the green sulphur bacteria, without measuring them directly.
From our model of the bacterial modes and their interactions, we have confirmed our probing scheme.
Finally, other applications were presented in Chapter~\ref{Chapter7}.
This includes, among others, the revelation of a non-classical property of a mechanical mirror that is mediating interactions between two otherwise non-interacting cavity light modes.

\section{Future work}

\subsection{Entanglement bound for indirect continuous interactions}

In Chapter~\ref{Chapter_revealing}, we considered the indirect interaction setting where the interactions are continuous. 
An illustration can be seen in Fig.~\ref{FIG_re_ABC} where an ancillary object $C$ is mediating interactions between two principal objects $A$ and $B$.
In this case, we have proven that classical mediator (zero discord $D_{AB|C}$ at all times) cannot increase quantum entanglement, i.e., $E_{A:BC}(\tau)\le E_{A:BC}(0)$.
See Theorem~\ref{TH_revealing} for details.
However, we note that having positive discord is not sufficient for entanglement gain.

Moreover, we often deal with imperfections in real situations.
It has been shown that states with zero discord are impossible to prepare in experiments, i.e., small perturbations will drive a zero-discord state into a state having positive, albeit small, discord~\cite{ferraro2010almost}.
As an example, let us consider a state of the form
\begin{equation}\label{EQ_nzd}
\rho=p\: \rho_{AB}\otimes \ket{\alpha}\bra{\alpha} +(1-p)\: \rho_{AB}^{\prime}\otimes \ket{\alpha^{\prime}}\bra{\alpha^{\prime}},
\end{equation}
where $p$ stands for probability, $\{\rho_{AB},\rho_{AB}^{\prime}\}$ are states of system $AB$, and $\{\ket{\alpha},\ket{\alpha^{\prime}}\}$ are coherent states of $C$. 
Here we note that the states $\ket{\alpha}$ and $\ket{\alpha^{\prime}}$ are non-orthogonal. 
In the limit where the expectation values of position, $\langle \alpha| x| \alpha \rangle$ and $\langle \alpha^{\prime} |x| \alpha^{\prime}\rangle$ are far apart, the states will be \emph{effectively} orthogonal.
Theoretically, however, quantum discord $D_{AB|C}$ of the state in (\ref{EQ_nzd}) will never be exactly zero.
Yet, one expects that not much entanglement can be gained with the effectively orthogonal states.
This calls for a bound on entanglement gain in terms of quantum discord $D_{AB|C}$, or other related quantities.\footnote{For example, one might consider quantum coherence of object $C$ (see Ref. \cite{roqcoherence} for a recent review on quantum coherence).}
Indeed, an observation of finite entanglement gain would certify that useful quantum discord (i.e., finite) formed during the dynamics.  

We have performed preliminary investigation towards this direction. 
First, we note from the example in Fig.~\ref{FIG_re_exp1} that entanglement increment in the partition $A:BC$ within $\omega t=\pi/4$ ($E_{A:BC}=1$ at this time), is not bounded by the maximum discord ($\approx 0.81$) or the average discord ($\approx0.38$).
For a special class of initial states and Hamiltonian, we have obtained a simple bound.
In particular, the gain of entanglement (as quantified by REE) is always bounded by entanglement with the mediator, i.e., 
\begin{equation}
|E_{A:BC}(\tau)-E_{A:BC}(0)|\le E_{AB:C}(\tau).
\end{equation}
This applies for a coherent dynamics with commuting interaction Hamiltonians, $[H_{AC},H_{BC}]=0$, and pure decoupled states as the initial condition \cite{edmundfyp}.
The present aim is therefore to generalise this result, taking into account more general states and Hamiltonians with non-commuting components.
One could also make it experimentally friendly by considering incoherent interactions with environments.

\subsection{The strength of non-decomposability and correlation gain}

Here we discuss a future direction that is an extension of Chapter~\ref{Chapter_detecting}.
Previously we have established, in the indirect interaction setting, that the gain of correlations between the objects $A$ and $B$ exceeding certain bounds, would imply that the the Hamiltonians are non-commuting (for closed systems), or more generally, non-decomposability of the dynamical operator (for open systems).

Now let us consider the following bound on the correlation between $A$ and $B$ (as quantified by a quantifier $Q$)
\begin{equation}\label{EQ_strengthnd}
Q_{A:B}(\tau)-Q_{A:B}(0)\le \eta,
\end{equation}
where $\eta$ is the to-be-proven bound that is independent of initial state.
One can then think of the Trotterized process, as depicted in Fig. ~\ref{FIG_de_illu}.
For a general evolution, suppose one writes the dynamical operator as
\begin{equation}
\Lambda=\Lambda_{BC}^{(n)}\Lambda_{AC}^{(n)} \cdots \Lambda_{BC}^{(2)}\Lambda_{AC}^{(2)}\: \Lambda_{BC}^{(1)}\Lambda_{AC}^{(1)},
\end{equation}
where the superscript denotes the number of sequences where $C$ first interacts with $A$ and then with $B$.
Next, one can apply the bound of Eq.~(\ref{EQ_strengthnd}) to the sequences above and obtain
\begin{eqnarray}
Q_{A:B}(\Delta t)-Q_{A:B}(0)&\le& \eta \nonumber \\
Q_{A:B}(2\Delta t)-Q_{A:B}(\Delta t)&\le& \eta \nonumber \\
&\vdots& \nonumber \\
Q_{A:B}(n\Delta t)-Q_{A:B}((n-1)\Delta t)&\le& \eta.
\end{eqnarray}
The sum would give 
\begin{equation}\label{EQ_boundeta}
Q_{A:B}(\tau)-Q_{A:B}(0)\le n\eta,
\end{equation}
where $\tau=n\Delta t$. 
In real situation, the observation of correlation gain that is exceeding certain $n$, would imply that the corresponding evolution operator cannot be modelled by $n$ sequences of decomposable operator of the form $\Lambda_{BC}\Lambda_{AC}$. 
We propose this as characterisation of the strength of non-decomposability of the actual dynamical operator $\Lambda$.
In particular, the higher the observed correlation $Q_{A:B}(\tau)$ the stronger the non-decomposability in the dynamics. 

\begin{figure}[h]
\centering
\includegraphics[scale=0.5]{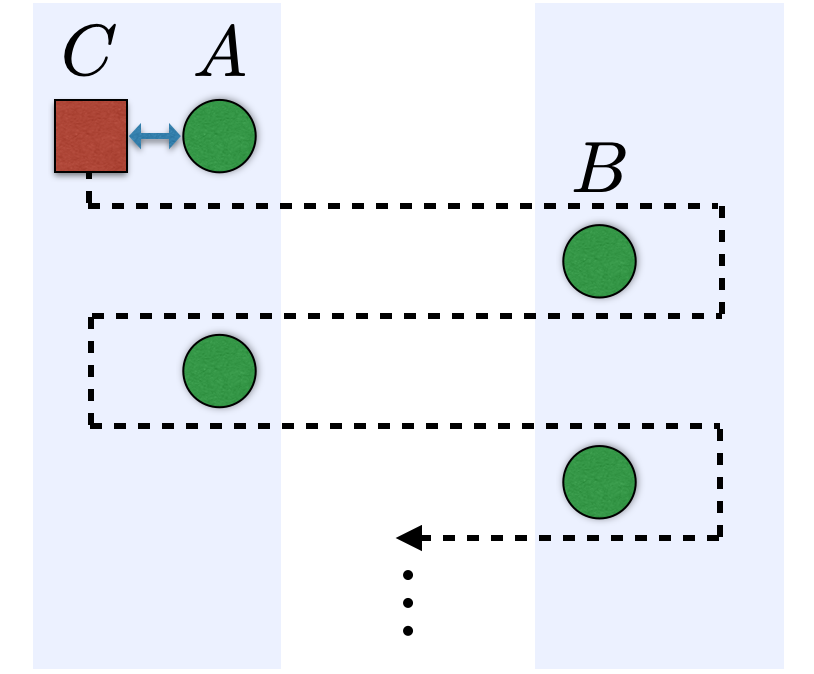} 
\caption{Illustration of Trotterized dynamics. 
The evolution of the whole system consists of sequences of interactions between the mediator object $C$ with $A$, followed by $C$ with $B$, each for a time $\Delta t$.
}
\label{FIG_de_illu} 
 \end{figure}

We now provide a simple example of the bound in Eq.~(\ref{EQ_strengthnd}).
Consider the mutual information as the correlation quantifier.
The corresponding bound is given by the following Lemma.

\begin{lemma}
For a tripartite system with a dynamical operator having the form $\Lambda_{BC}\Lambda_{AC}$, the mutual information follows
\begin{equation}
I_{A:B}(t)-I_{A:B}(0) \le 4\log_2(d_C),
\end{equation}
where $d_X$ is the dimension of object $X$ and we have assumed $d_A,d_B>d_C$ for simplicity.
\end{lemma}
\begin{proof}
Let us begin with the monotonicity of mutual information under local operations to arrive at 
\begin{equation}\label{EQ_prepre}
I_{A:B}(t)\le I_{A:BC}(t)=I_{A:BC}(\Lambda_{BC}\Lambda_{AC}[\rho(0)])\le I_{A:BC}(\Lambda_{AC}[\rho(0)]).
\end{equation}
Next, we have the following arguments
\begin{eqnarray}\label{EQ_bound4}
I_{A:B}(t)-I_{A:B}(0)&\le& I_{A:BC}^{\prime}-I_{A:B}(0) \nonumber \\
&=&I_{AC:B}^{\prime}-I_{B:C}^{\prime}+I_{A:C}^{\prime}-I_{A:B}(0) \nonumber \\
&\le&I_{AC:B}(0)-I_{A:B}(0)-I_{B:C}^{\prime}+I_{A:C}^{\prime} \nonumber \\
&\le&I_{B:C|A}(0)+I_{A:C}^{\prime} \nonumber \\
&\le&4\log_2(d_C),
\end{eqnarray}
where the prime symbol denotes the state $\Lambda_{AC}[\rho(0)]$. 
We justify the steps as follows.
Eq.~(\ref{EQ_prepre}) has been used in the first line.
The second line is apparent by utilising the definition of mutual information between two parties, i.e., $I_{X:Y}=S_X+S_Y-S_{XY}$, where for example $S_X$ is the von Neumann entropy of object $X$.
The third line uses the monotonicity property under the operation $\Lambda_{AC}$ that is local in the partition $AC:B$.
We have used the definition of conditional entropy and the positivity of mutual information in the fourth line.
The mutual information between $B$ and $C$ (conditioned on $A$) is bounded by $2\log_2(d_C)$. 
Note that we assume $d_B>d_C$.
This is also true for mutual information between $A$ and $C$, which completes the proof.
\end{proof}

Therefore, the corresponding bound for $n$ sequences of decomposable dynamics reads 
\begin{equation}
I_{A:B}(t)-I_{A:B}(0)\le 4n\log_2(d_C).
\end{equation}
To demonstrate violations of the bound for certain $n$, one can make use of the setup described in Section~\ref{SC_faf}, where two cavity fields are interacting via a two-level atom.
Our simulations show that high correlations can be obtained by having pure states with more photons as the initial condition.
Future work is aimed at reducing the bound in Eq.~(\ref{EQ_bound4}) as our simulations show the gain of mutual information being bounded by only $2\log_2(d_C)$.
Also, one might consider other correlation quantifiers.

\subsection{Protocols witnessing the presence or absence of quantum objects}

Chapter~\ref{Chapter_gravity} has brought to our attention the need of a protocol for witnessing quantum mediators. 
To some extend, this has been discussed in Section~\ref{SC_directmediated}.
In particular, if quantum entanglement between the principal objects $A$ and $B$ increases and the state of system $AB$ is pure at all times, then the dynamics is only possible through direct interactions. 
This shows the absence of a mediating object. 
An experimentally friendly scheme requires one to generalise this method to mixed states, open system dynamics, and finite number of measurements.

Moreover, we are also interested in a protocol showing the presence of a mediating object.
Based on the study of the speed of distributing quantum entanglement in Chapter~\ref{Chapter_speedup}, a simple way can be proposed as follows.
If one starts with mixed initial state for system $AB$, by considering closed dynamics, one cannot get maximum entanglement in the partition $A:B$.
This is because unitary operator preserves the purity of the state and that maximum entanglement is given by a pure state.
Therefore, the observation of maximum entanglement between the principal objects, would reveal the presence of a mediating object. 
This mediating object results in the change of purity of the state of $AB$, making it possible to reach maximum entanglement.
This protocol is in general hard to implement due to the requirement of maximum entanglement.
We aim at making this scheme also experimentally friendly. 
A potential direction would be to consider an entanglement bound for a given purity.
If the entanglement $E_{A:B}$ reaches a higher value, this would imply the presence of a quantum mediator.

\subsection{Complete link between initial correlations with mediators and entangling time limit}

In Chapter~\ref{Chapter_speedup}, we have proven that the speed limit of entanglement creation in the indirect interaction setting cannot be faster than the direct setting.
See Theorem~\ref{TH_untimate} for details.
However, the saturation of the speed limit is possible and we have shown this with exemplary dynamics. 
Our examples, e.g., the ones in Section~\ref{SC_satlimit}, show that initial correlation in the partition $AB:C$ is important for the task.
Therefore, for completeness, we aim to show that for initial state of the form $\rho_{AB}\otimes \rho_C$, the entangling time for the indirect setting follows a strict bound $T>\Gamma_{\text{di}}$.

We also performed preliminary simulations with coherent dynamics in the indirect interaction setting. 
For initial assessment, we consider three qubits $A$, $B$, and $C$.
We have taken random initial states (separable in $A:B$ partition) and interaction Hamiltonians.\footnote{Note that the resource equality has been implemented, i.e., $\min \{ \langle M\rangle,\Delta M\}=1$.}
Our simulations suggest a time bound
\begin{equation}
T\ge \frac{\pi}{2}.
\end{equation}
Therefore, in this scenario, it follows that $T\ge 2\Gamma_{\text{di}}$ for the case of three qubits.
A simple example saturating this limit is given as follows
\begin{eqnarray}
\ket{\psi(0)}&=&\ket{110}, \nonumber \\
H&=&\frac{\hbar \Omega}{\sqrt{2}} \left( \sigma^+_A\otimes \sigma^-_C+\sigma^-_A\otimes \sigma^+_C+ \sigma^+_B\otimes \sigma^-_C+\sigma^-_B\otimes \sigma^+_C\right),
\end{eqnarray}
where $\sigma^{+(-)}_X$ is the raising (lowering) operator of object $X$.
Indeed, maximum entanglement (negativity $N_{A:B}=0.5$) is reached in $T=\Omega t=\pi/2$.

\subsection{Gravity-mediated entanglement: Full treatment}

The results we presented in Chapter~\ref{Chapter_gravity} have shown the possibility of having entanglement gain between two spherical masses that are interacting via gravitational potential.
The gravitational energy is taken as 
\begin{eqnarray}
H_{\text g}&=&-\frac{Gm^2}{ (L + x_B - x_A)} \label{EQ_fullhg}\\
&\simeq&-\frac{Gm^2}{L}\left(1+\frac{(x_A - x_B)}{L}+\frac{(x_A - x_B)^2}{L^2}\right),\label{EQ_trun}
\end{eqnarray}
where we have applied the approximation $(x_A - x_B)\ll L$.
As gravitational coupling is weak in practically any experimental situation, the resulting entanglement is either very small or requires a long time to accumulate. 
One might wonder if there is more entanglement generated by taking the full Hamiltonian as in Eq.~(\ref{EQ_fullhg}).
This is intuitive since higher order terms in the expansion provide coupling between the two masses, although at the expense of strength. 

In this direction, we aim to perform simulations of the dynamics considering full gravitational energy.
We note that logarithmic negativity is a valid entanglement quantifier only for Hamiltonians that have up to quadratic terms in operators, e.g., Eq.~(\ref{EQ_trun}).
This is because these Hamiltonians would preserve Gaussianity of the state, allowing logarithmic negativity to be a monotone quantifier of entanglement. 
However, we note that the simulation with full term has to converge to the same dynamics in the correct limit, i.e., $(x_A - x_B)\ll L$.
For general regime, one can use other quantifiers. 
As an example, we propose the use of pure state for the initial condition. 
For a closed dynamics, one can then calculate the entropy of one of the masses as a quantifier of entanglement.


\appendix
\chapter{Trotter expansion}\label{A_trotter}

A simple derivation of the Trotter formula~\cite{trotter1959product} is provided in Theorem \ref{THA_trotter} below.

\begin{theorem}\label{THA_trotter}
For arbitrary square matrices $A$ and $B$, the following holds
\begin{equation}
e^{A+B}=\lim_{n\rightarrow \infty}\left( e^{\frac{A}{n}}e^{\frac{B}{n}}\right)^n.
\end{equation}
\end{theorem}

\begin{proof}
We recall the Baker-Campbell-Hausdorff equation
\begin{equation}\label{EQA_bch}
\ln{\left(e^{A}e^B\right)}=A+B+\frac{1}{2}[A,B]+\frac{1}{12}[A,[A,B]]+\frac{1}{12}[B,[B,A]]+\cdots
\end{equation}
and the identity $\lim_{x}e^{f(x)}=e^{\lim_x f(x)}$.

Taking the natural logarithm of $(e^{A/n}e^{B/n})^n$ gives 
\begin{equation}\label{EQA_bch2}
n\ln{\left(e^{\frac{A}{n}}e^{\frac{B}{n}}\right)}=n\left( \frac{A}{n} +\frac{B}{n}+\frac{1}{2n^2}[A,B]+\frac{1}{12n^3}[A,[A,B]]+\frac{1}{12n^3}[B,[B,A]]+ \cdots \right),
\end{equation}
where we have used Eq. (\ref{EQA_bch}). 

By applying the limit $n\rightarrow \infty$ to the exponent of Eq. (\ref{EQA_bch2}) one has
\begin{eqnarray}\label{EQA_bch3}
\lim_{n\rightarrow \infty}\left(e^{\frac{A}{n}}e^{\frac{B}{n}}\right)^n &=&e^{\lim_{n\rightarrow \infty} \left( A +B+\frac{1}{2n}[A,B]+\frac{1}{12n^2}[A,[A,B]]+\frac{1}{12n^2}[B,[B,A]]+ \cdots \right)}\nonumber \\
&=&e^{A +B},
\end{eqnarray}
where we have utilised the limit identity in the first line.
\end{proof}

\afterpage{\blankpage} 
\chapter{Entanglement localisation: Prescription}\label{A_eloc}

As shown in Chapter \ref{Chapter_revealing}, entanglement distribution between remote objects $A$ and $B$ is possible through a classical mediator $C$. 
Here I provide an insight into this example and prescribe a method for applying it to a more general setting.
We also note that our example has been realised recently~\cite{elocexp}.

For simplicity, let us assume that the Hamiltonian is of the form 
\begin{equation}
H=H_A\otimes H_{C_1}+H_B\otimes H_{C_2},
\end{equation}
where $[H_{C_1},H_{C_2}]=0$, i.e., these components share the same eigenbasis.
Now we take a quantum--classical initial state in the partition $AB:C$, 
\begin{equation}
\rho(0)=\sum_c p_c \: \rho_{AB|c} \otimes \ket{c}\bra{c},
\end{equation}
where $\{\ket{c}\}$ form the common eigenbasis of $H_{C_1}$ and $H_{C_2}$.
Let us also use, as the eigenvalues (assumed dimensionless), $E^c_{1}$ and $E^c_{2}$ respectively.

The state after a short time $\Delta t$ reads
\begin{eqnarray}
\rho(\Delta t)&=&\left(\openone -\frac{i\Delta t}{\hbar}H\right)\: \rho(0)\: \left(\openone+\frac{i\Delta t}{\hbar}H\right) \nonumber \\
&=&\sum_c p_c \Big( \rho_{AB|c}-\frac{i\Delta t }{\hbar} E^c_1H_A \: \rho_{AB|c}-\frac{i\Delta t }{\hbar}E^c_2 H_B \:\rho_{AB|c} \nonumber \\
&&+\frac{i\Delta t}{\hbar}  \rho_{AB|c}\: E^c_1H_A+\frac{i\Delta t}{\hbar}  \rho_{AB|c}\:E^c_2 H_B \Big) \otimes \ket{c}\bra{c} \nonumber \\
&=&\sum_c p_c \big(\rho_{AB|c}-\frac{i\Delta t}{\hbar}[(E^c_1H_A+E^c_2H_B),\rho_{AB|c}]\big) \otimes \ket{c}\bra{c} \nonumber \\
&=&\sum_c p_c \:\rho_{AB|c}^{\Delta t} \otimes \ket{c}\bra{c},
\end{eqnarray}
where we have used $\rho_{AB|c}^{\Delta t}$ to denote the short coherent evolution of $\rho_{AB|c}$ with weighted local Hamiltonians $E^c_1H_A+E^c_2H_B$.
Note that we have ignored terms proportional to $(\Delta t)^2$.
The evolution for a time $t$ simply requires successive applications of the above process, i.e., $\rho(t)=\sum_c p_c \:\rho_{AB|c}^{t} \otimes \ket{c}\bra{c}$.

We show that entanglement in the partition $A:BC$ stays constant as follows
\begin{equation}
E_{A:BC}(t)=\sum_c p_c\: E_{A:B}(\rho_{AB|c}^t)=\sum_c p_c\: E_{A:B}(\rho_{AB|c})=E_{A:BC}(0),
\end{equation} 
where we have used the flags condition \cite{flags} on the quantum--classical state in the first and last equality, and the monotonicity of entanglement under local unitary operations in the second equality.

Now let us consider entanglement between $A$ and $B$.
In particular we have
\begin{equation}
E_{A:B}(t)=E_{A:B}\big(\sum_c p_c\: \rho_{AB|c}^t\big)\le \sum_c p_c\: E_{A:B}(\rho_{AB|c}^t)=\sum_c p_c\: E_{A:B}(\rho_{AB|c}),
\end{equation}
where we have used the convexity of entanglement for the inequality and the monotonicity for the last equality.
One immediately realises that, in order to have entanglement gain in the partition $A:B$ with this method, there has to be entanglement already present initially, i.e., in the state $\rho_{AB|c}$.

Note that the necessary requirements above are satisfied in the example given in Chapter \ref{Chapter_revealing}, where $H=(\sigma^x_A\otimes \openone \otimes \sigma^x_C+\openone \otimes \sigma^x_B \otimes \sigma^x_C)\hbar \omega$, and the initial state reads $\tfrac{1}{2} \ket{\psi_+}\bra{\psi_+} \otimes \ket{+}\bra{+} + \tfrac{1}{2} \ket{\phi_+}\bra{\phi_+} \otimes \ket{-}\bra{-}$. 
In particular, the classical basis of system $C$ is the eigenbasis of the Pauli matrix $\sigma_x$ and there is maximum entanglement in the initial states $\{\ket{\psi_+},\ket{\phi_+}\}$.
Furthermore, the evolution of the states for a time $\tau=\pi/8\omega$ follows
\begin{eqnarray}
\ket{\psi_+}\bra{\psi_+}&\rightarrow& \rho_{AB|1}^{\tau}= \rho_{\text{max}}\nonumber \\
\ket{\phi_+}\bra{\phi_+}&\rightarrow& \rho_{AB|2}^{\tau}= \rho_{\text{max}},
\end{eqnarray}
where $\rho_{\text{max}}$ is, up to a universal phase factor, a density matrix of a pure maximally entangled state 
\begin{equation}
\frac{1}{\sqrt{2}}\left( \ket{0}\ket{y_+}+\ket{1}\ket{y_-} \right),
\end{equation}
with $\{\ket{y_+},\ket{y_-}\}$ being the eigenbasis of the $\sigma_y$ Pauli matrix.
Therefore, one can see that $E_{A:B}(0)=0$ and $E_{A:B}(\tau)=1$, where we have used REE as the entanglement quantifier.

\label{Bibliography} 
\lhead{\emph{Bibliography}} 
\bibliographystyle{unsrtnat} 

\end{document}